\newcommand{\draft}{0}
\newcommand{\llncs}{0}
\ifnum\llncs=1\relax
  \documentclass[runningheads,orivec]{llncs}
\else
  
  \documentclass[11pt]{article}
  \usepackage[margin=1in]{geometry}
\fi

\usepackage[T1]{fontenc}
\usepackage{xspace,comment}
\usepackage{lmodern,amsmath,amssymb,mathtools,booktabs,array,microtype}
\usepackage{xcolor}

\ifnum\llncs=1\relax
\else
  \usepackage{amsthm}
\fi
\usepackage[hidelinks]{hyperref}
\usepackage{cleveref}

\ifnum\llncs=1\relax
  \counterwithin{theorem}{section}
  \makeatletter
  \@for\pke@env:=lemma,proposition,corollary,remark\do{%
    \expandafter\let\csname c@\pke@env\endcsname\c@theorem
    \expandafter\let\csname the\pke@env\endcsname\thetheorem
    \expandafter\let\csname theH\pke@env\endcsname\theHtheorem
  }
  \makeatother
  \spnewtheorem{observation}[theorem]{Observation}{\bfseries}{\itshape}
  \spnewtheorem{openproblem}{Open Problem}{\bfseries}{\upshape}
\else
  \theoremstyle{plain}
  \newtheorem{theorem}{Theorem}[section]
  \newtheorem{lemma}[theorem]{Lemma}
  \newtheorem{proposition}[theorem]{Proposition}
  \newtheorem{corollary}[theorem]{Corollary}
  \newtheorem{observation}[theorem]{Observation}
  \newtheorem{definition}{Definition}
  \theoremstyle{remark}
  \newtheorem{remark}[theorem]{Remark}
  \theoremstyle{definition}
  \newtheorem{openproblem}{Open Problem}

  \newcommand{\titlerunning}[1]{}
  \newcommand{\authorrunning}[1]{}
  \newcommand{\inst}[1]{\textsuperscript{#1}}
  \newcounter{affiliation}
  \makeatletter
  \newcommand{\institute}[1]{%
    \gdef\@date{\small
      \begingroup
      \setcounter{affiliation}{1}%
      \renewcommand{\and}{\par\stepcounter{affiliation}\textsuperscript{\arabic{affiliation}}}%
      \textsuperscript{\arabic{affiliation}}#1\par
      \endgroup}%
  }
  \makeatother
  \date{}
\fi
\crefname{observation}{observation}{observations}
\Crefname{observation}{Observation}{Observations}
\crefname{openproblem}{open problem}{open problems}
\Crefname{openproblem}{Open Problem}{Open Problems}

\usepackage{tikz}
\usetikzlibrary{arrows.meta,calc,fit}
\usepackage[normalem]{ulem}

\newcommand{\bits}{\{0,1\}}
\newcommand{\ket}[1]{\lvert #1\rangle}
\newcommand{\bra}[1]{\langle #1\rvert}
\newcommand{\negl}{\mathsf{negl}}
\newcommand{\poly}{\mathsf{poly}}
\newcommand{\Dec}{\mathsf{Dec}}
\newcommand{\PSPACE}{\mathsf{PSPACE}}
\newcommand{\KeyGen}{\mathsf{KeyGen}}
\newcommand{\Enc}{\mathsf{Enc}}
\newcommand{\Obf}{\mathsf{Obf}}
\newcommand{\Eval}{\mathsf{Eval}}
\newcommand{\blank}{\square}

\DeclareMathOperator{\SD}{SD}
\newcommand{\Image}{\operatorname{Im}}
\DeclareMathOperator{\supp}{supp}
\DeclareMathOperator{\Interp}{Interp}
\newcommand{\numberthis}{\refstepcounter{equation}\tag{\theequation}}

\numberwithin{equation}{section}
\ifnum\draft=1
    \newcommand{\amit}[1]{{\color{violet}[\textbf{Amit}: #1]}}
    \newcommand{\mo}[1]{{\color{blue}[\textbf{Mohammed}: #1]}}
\else
    \newcommand{\amit}[1]{}
    \newcommand{\mo}[1]{}
\fi

\titlerunning{Quantum Cryptomania with Classical Communication without One-Way Functions}
\authorrunning{Mohammed Barhoush and Amit Behera}
\ifnum\llncs=0
    \AddToHook{env/theorem/begin}{\crefalias{theorem}{theorem}}
    \AddToHook{env/lemma/begin}{\crefalias{theorem}{lemma}}
    \AddToHook{env/proposition/begin}{\crefalias{theorem}{proposition}}
    \AddToHook{env/corollary/begin}{\crefalias{theorem}{corollary}}
    \AddToHook{env/observation/begin}{\crefalias{theorem}{observation}}
    \AddToHook{env/remark/begin}{\crefalias{theorem}{remark}}
\fi

\begin{document}
\title{Quantum Cryptomania with Classical Communication without One-Way Functions
}
\author{Mohammed Barhoush\inst{1} \and Amit Behera\inst{2}}
\institute{
IonQ Inc.\\
\texttt{mohammed.barhoush@ionq.co}
\and
NTT Research\\
\texttt{amitbehera1767@gmail.com}
}
\maketitle
\begin{abstract}
In classical cryptography, one-way functions are necessary for most of cryptography. In the quantum world, however, Kretschmer at {TQC'21} and subsequent works have demonstrated that quantum versions of various primitives in minicrypt, such as signatures and private-key encryption, do not require quantum-computable one-way functions (\textsf{qOWF}s). However, the question of the necessity of \textsf{qOWF}s for primitives in Cryptomania like public-key encryption with classical public keys and ciphertexts (\textsf{qPKE}) has remained open.  

We provide a comprehensive picture of the relationship between \textsf{qPKE} and \textsf{qOWF}. First, our main result is a black-box separation between \textsf{qOWF}s and \textsf{qPKE}. In particular, we show that relative to a unitary, real, and self-inverse oracle, there exists a \textsf{qPKE} scheme (with classical keys and ciphertexts), but \textsf{qOWF}s do not exist. Further, our oracles admit a point-function obfuscation construction, thus ruling out black-box construction of \textsf{qOWF} from point function obfuscations. 

On the positive side, we construct \textsf{qPKE}
from quantum indistinguishability obfuscation for quantum circuits with classical obfuscated descriptions and efficiently verifiable one-way puzzles, which is known to be a potentially weaker assumption than \textsf{qOWF}s, due to an existing black-box separation. Our construction is the quantum analog of the classical construction of \textsf{PKE} from one-way functions and indistinguishability obfuscation in the quantum computation with classical communication (\textsf{QCCC}) setting. 
\end{abstract}
\newpage
\tableofcontents

\newpage
\section{Introduction}
One-way functions (\textsf{OWF}s) form the foundation of classical computational cryptography~\cite{IL89}. They are necessary for most non-trivial primitives, including signatures, many-time encryption, and bit commitments. This is reflected in Impagliazzo's five worlds~\cite{I95}, where
\textsf{OWF}s form the basis of Minicrypt, the world that supports
basic computational cryptography. Furthermore, they are implied by more powerful cryptographic primitives, like public-key encryption (\textsf{PKE}), a central primitive of Cryptomania\footnote{{Cryptomania is the world in which secure key agreement over a public channel exists~\cite{I95}; \textsf{PKE} suffices for such key agreement.}}. \textsf{PKE} is one of the most fundamental primitives in cryptography, enabling secure communication without shared keys. Hence, understanding how \textsf{PKE} relates to other assumptions has been  extensively studied in classical cryptography. In particular, it is known that classical public-key encryption (\textsf{PKE}) implies private-key encryption, which in turn is equivalent to  \textsf{OWF}s~\cite{IL89,HILL99}.  {In the other direction}, \textsf{OWF}s may not be sufficient for \textsf{PKE}, as Impagliazzo and Rudich ruled out fully black-box constructions of classical \textsf{PKE} from \textsf{OWF}s~\cite{IR89}. This gives a familiar view of the two primitives: one-wayness is necessary for public-key encryption, but may not be sufficient.

Quantum computation changes the relationships between cryptographic
primitives. Classical proofs often rely on fixing random coins,
copying internal states, or rewinding a computation, and these steps
do not automatically extend to quantum algorithms. Indeed,  prior
constructions and oracle separations show that, unlike in the classical setting, commitments, one-time symmetric encryption, and one-time signatures can exist
relative to a quantum oracle even when quantum computable one-way functions (\textsf{qOWF}s) do not~\cite{Kre21,AQY22,MY22,BBO+24}. 
It may appear that the source of these anomalies is that the quantum versions of these minicrypt primitives allow for quantum communication. However, even when restricted to entirely classical communication, the relationships remain different. Indeed, commitments, symmetric encryption, and digital signatures with classical communication do not imply \textsf{qOWF} in a black-box manner~\cite{BNY25,BBO+24,Bar26,BEHMV26}, unlike their classical counterparts. Much of the recent work has focused on the relationships of \textsf{qOWF} with the quantum versions of Minicrypt primitives~\cite{BBO+24,BNY25,Bar26,BEHMV26}, but its corresponding relationship with quantum Cryptomania, particularly public-key encryption with classical communication (denoted as \textsf{qPKE}), is somewhat less understood. Previous work has shown that known classical black-box barriers to constructing \textsf{qPKE} from \textsf{OWF}s can be adapted to the quantum setting~\cite{LLLL25,LLLL25b}. Nevertheless, in the reverse direction, the more basic question of whether \textsf{qOWF}s are necessary for \textsf{qPKE} has remained open.

{Allowing quantum public keys changes the picture.
Encryption with quantum public keys and classical ciphertexts can
be constructed from post-quantum \textsf{OWF}s~\cite{KMNY23},
and even from potentially weaker assumptions, such as pseudorandom
function-like states with proofs of destruction~\cite{BGH23}
and short pseudorandom states (\textsf{PRS}s)~\cite{BBO+24}.
Thus, the black-box barrier for classical public keys does not
extend to quantum public keys in general~\cite{LLLL25}.
However, quantum public keys do not provide the same form of
public availability as classical keys. A classical public key
can be published once and freely copied by anyone, whereas
quantum public keys generally require distributing quantum states
to each recipient~\cite{BGH23}.
This distinction motivates studying whether the assumptions
needed for public-key encryption can also be weakened while
retaining classical public keys and entirely classical
communication.}
This leads us to ask the following natural question.
\begin{openproblem}\label{op:sep}
    Does public-key encryption still require quantum-computable one-way functions in the quantum setting, even if the keys and communication remain classical?
\end{openproblem}

\subsection{Our results}

In this work, we answer \Cref{op:sep} in the negative. Our main result is  {a separation of} \textsf{qOWF} from \textsf{qPKE}. 

\begin{theorem}[Informal: oracle separations]\label{thm:inf-main-sep}
There exists a quantum real and self-inverse unitary oracle relative to which perfectly correct, \textsf{IND-CPA}-secure \textsf{qPKE} exists but {\textsf{qOWF}s do not}.
\end{theorem}

{Our separation also rules out quantum-computable pseudorandom
generators (\textsf{qPRG}) and quantum-computable pseudorandom functions
{(\textsf{qPRF}s), since both} {imply \textsf{qOWF}s}. In fact, we even rule out a weaker notion of these functions that allow quantum input sampling, as explored in \mbox{\cite{BNY25}}, as long as they are deterministic on every input except with negligible probability.}

Interestingly, our oracles (via a similar proof to that for our \textsf{qPKE} construction) also admit the construction of point-function obfuscation, which is a quantum algorithm that takes a classical encoding of a point function and outputs a classical encoding of the obfuscation that satisfies single-program virtual black-box (\textsf{VBB}) security without correlated auxiliary inputs (see \Cref{vbb:def}). Therefore, we also separate \textsf{qOWF}s from point-function obfuscation.

Finally, since our oracles are real, self-inverse, and unitary, our result {rules out black-box constructions} that use unitary, inverse, controlled, transpose, and complex-conjugate operations.

Our result gives a very interesting relationship between \textsf{qOWF}s and \textsf{qPKE}. {Recall that Li, Li, Li, and Liu} rule out black-box constructions of perfectly correct \textsf{qPKE} from \textsf{qOWF}s~\cite{LLLL25}. Together with \cref{thm:inf-main-sep}, this gives {a two-way black-box separation} between perfectly correct \textsf{qPKE} and \textsf{qOWF}s. This highlights the complexity of the quantum landscape as opposed to the hierarchical structure in the classical setting.  

\paragraph{Public-key encryption from \textsf{EV-OWPuzz} and \textsf{qiO}.} \textsf{qOWF} is not the only one-way primitive in quantum cryptography. There are several one-way notions believed to be weaker due to black-box separations, such as one-way state generators \cite{MY22} and {efficiently verifiable one-way puzzles} (\textsf{EV-OWPuzz}) \cite{CGG24}. Given that our separation result (\Cref{thm:inf-main-sep}) shows that \textsf{qOWF}s may not be necessary for \textsf{qPKE}, it is therefore natural to ask if \textsf{qPKE} can be built from weaker forms of one-way primitives that exist in the quantum world. 

Note that classically, a well-known construction shows that \textsf{OWF}s together with indistinguishability obfuscation (\textsf{iO}) give \textsf{PKE}~\cite{SW14}. \textsf{EV-OWPuzz} are a natural candidate to replace \textsf{OWF}s as a weaker quantum one-way primitive in a quantum analog of this construction, since they are already implied by, and hence necessary for, \textsf{qPKE}~\cite{CGG24}. {An \textsf{EV-OWPuzz} consists of} two \textsf{QPT} algorithms: a sampler $S$, which outputs a classical puzzle and answer $(p,a)$, and a verifier $V$, {which accepts honestly sampled pairs except with negligible probability}. One-wayness requires that, for every \textsf{QPT} algorithm $A$,
\[
 \Pr_{(p,a)\leftarrow {S}(1^\lambda)}
 \bigl[V(p,A(p))=1\bigr]
 \le \negl(\lambda).
\]

Chung, Goldin, and Gray~\cite[Section~3]{CGG24} asked whether
\textsf{EV-OWPuzz} suffice to construct secret-key encryption
with quantum computation and classical communication.
We study a related question for public-key encryption:
do \textsf{EV-OWPuzz} suffice when combined with indistinguishability obfuscation?

For our second result, we show that \textsf{EV-OWPuzz} suffices to construct \textsf{qPKE} when combined with quantum indistinguishability obfuscation (\textsf{qiO}), {for pseudodeterministic quantum circuits with classical obfuscated descriptions, under Definition~\ref{def:semantic-qio}. Obfuscation of this class of circuits is studied in~\cite{BGMS26}.}
\begin{theorem}[Informal: construction]
\textsf{EV-OWPuzz} and \textsf{qiO} imply \textsf{IND-CPA}-secure \textsf{qPKE} with classical communication and keys.
\end{theorem}

\paragraph{Framework for different notions of generalized \textsf{qOWF}} 
Next, we take a closer look at the very definition of \textsf{qOWF}. The definition most commonly used in the literature~\cite{AC02,CLS01,DMS00,GKT26,HY20,IT03a,IT03b,KKK+05,KQT25,MY23,OTU00} requires that determinism holds for every input, except with negligible probability, and that the one-wayness game is with respect to the uniform distribution on the input space. However, there are natural generalizations to relax the above definition with respect to the determinism requirement and the security requirement through the distribution of inputs considered in the one-wayness game. For instance, should evaluation be near-deterministic on every input or only on average, or should inversion hardness be measured on uniform inputs or on inputs produced by a specified quantum sampler? Indeed, Barhoush, Nishimaki, and Yamakawa studied the role of quantum input sampling, i.e., where the input of the function in the security experiment can be sampled by an arbitrary efficient quantum algorithm~\cite{BNY25}. They showed that allowing quantum samplers can yield a potentially weaker assumption.  Thus, one may suspect that our separation result (\Cref{thm:inf-main-sep}) is due to similar differences in the definition of \textsf{qOWF}, and not a fundamental barrier towards \textsf{qOWF} being necessary for \textsf{qPKE}.

Therefore, to better understand the role of these different potential generalizations of \textsf{qOWF} in our separation result (\Cref{thm:inf-main-sep}), we provide a general framework for different generalizations of \textsf{qOWF}.
In our framework, a \textsf{qOWF} consists of a pair of QPT algorithms $(S,E)$, where $S$ is the input sampler, and $E$ is the evaluator. Let $f$ be the function that maps $x$ to the {most likely output of $E(x)$}. In each notion, one-wayness requires that, for every QPT inverter $A$,
\[
 \Pr_{\substack{
      X\leftarrow S(1^\lambda)\\
      Y\leftarrow E(X)}}
 \bigl[E(A(Y))=Y\bigr]
 \le \negl(\lambda).
\]


In this framework, we discuss two natural generalizations of the existing \textsf{qOWF} definition and finally consider an even weaker definition that combines the relaxations of the two generalizations.
\begin{enumerate}
\item \textbf{Generalized-$\textsf{qOWF}_1$.}\label{it:def-1} The sampler $S$ is {an efficient quantum algorithm}, and
\[
 \forall x\in\{0,1\}^{m(\lambda)},\qquad
 \Pr[E(x)\ne f(x)]\le\negl(\lambda).
\]

In other words, this definition generalizes the \textsf{qOWF} definition by replacing the uniform distribution with arbitrary and \emph{potentially quantum} sampling distributions for the one-wayness experiment. Thus, it is a quantum analog of the general classical \textsf{OWF}s whose one-wayness is with respect to  {arbitrary efficiently classically samplable distributions}; see~\cite{PW11,ABH+26}. Classically, they are equivalent to traditional \textsf{OWF}s, since the randomness can be pulled out from any classical sampler, {but this argument does not extend to quantum samplers. Barhoush, Nishimaki, and Yamakawa~\cite{BNY25} give black-box separations demonstrating this distinction for pseudorandom generators}. 

\item \textbf{Generalized-$\textsf{qOWF}_2$.}\label{it:def-2} The sampler $S$ is uniform, and
\[
 \Pr_{X\leftarrow\{0,1\}^{{m(\lambda)}}}
 [E(X)\ne f(X)]\le\negl(\lambda).
\]

This definition generalizes the existing \textsf{qOWF} definition by allowing \emph{determinism to hold on average over uniform inputs}, though the one-wayness game remains in terms of the uniform inputs. For instance, this notion allows the evaluator to be completely non-deterministic on a negligible fraction of inputs. However, on an overwhelming fraction of the inputs, the function is deterministic up to negligible error. {This notion of determinism has been used to define quantum-computable pseudorandom generators; see~\cite{Bar26,BEHMV26}}.

\item \textbf{Generalized-$\textsf{qOWF}_3$.}\label{it:def-3} The sampler $S$ is arbitrary, and
\[
 \Pr_{X\leftarrow S(1^\lambda)}
 [E(X)\ne f(X)]\le\negl(\lambda).
\]
This definition is naturally obtained by combining the two somewhat incomparable relaxations allowed in the first two definitions. Unlike the above two definitions, this notion may allow the evaluator to behave non-deterministically on a majority of inputs, raising the question of whether this definition should even be called a function. Nevertheless, we study this notion as a natural generalization of the first two definitions for the theoretical exploration of the limits of our separation result.
\end{enumerate}

For brevity, we will just use $\textsf{qOWF}_1$, $\textsf{qOWF}_2$ and $\textsf{qOWF}_3$ to refer to generalized-$\textsf{qOWF}_1$, generalized-$\textsf{qOWF}_2$ and generalized-$\textsf{qOWF}_3$, respectively, from now on.

The distinction between these notions of \textsf{qOWF} suggests a more nuanced question: how does {the relationship between \textsf{qOWF} and \textsf{qPKE}} change across the different notions of \textsf{qOWF}?

We provide a clear picture of this question. Firstly, we note that our separation result is even stronger than the version stated in \Cref{thm:inf-main-sep}, and in particular we show the following fine-grained version of it.

\begin{theorem}[Informal: fine-grained version of oracle separations]\label{thm:inf-main-sep-finer}
There exists a quantum real and self-inverse unitary oracle relative to which perfectly correct, \textsf{IND-CPA}-secure \textsf{qPKE} exists but neither $\textsf{qOWF}_1$s nor $\textsf{qOWF}_2$s exist.
\end{theorem}

Thus, this {gives black-box separations} between \textsf{qPKE} and the two notions $\textsf{qOWF}_1$ and $\textsf{qOWF}_2$, but what about the last and weakest notion, $\textsf{qOWF}_3$? Interestingly, we find that $\textsf{qOWF}_3$ has a different relation with \textsf{qPKE} and \textsf{EV-OWPuzz}. Our third result presents this relationship.

\begin{theorem}[Informal: equivalence with puzzles]
{$\textsf{qOWF}_3$s exist if and only if \textsf{EV-OWPuzz} exist}. Since \textsf{qPKE} implies \textsf{EV-OWPuzz}~\cite{CGG24}, we conclude that {\textsf{qPKE} implies $\textsf{qOWF}_3$s}.
\end{theorem}


This equivalence may seem surprising since $\textsf{qOWF}_3$ still requires nearly deterministic evaluation on sampled inputs, whereas in an \textsf{EV-OWPuzz}, there is no such determinism requirement. Nevertheless, we show that since $\textsf{qOWF}_3$ allows average-input determinism for an arbitrary quantum sampleable distribution, we can build it from \textsf{EV-OWPuzz}. We note that {for a relativizing equivalence, it is necessary} that we allow both average-input correctness and security relative to quantum input sampling. Indeed, if we remove even one of these two relaxations, we get either $\textsf{qOWF}_1$ or $\textsf{qOWF}_2$, which by \Cref{thm:inf-main-sep-finer} are separated from \textsf{qPKE} and, hence, are also separated from \textsf{EV-OWPuzz}. 

We summarize all our results in Figure~\ref{fig:primitive-relations}. 
\begin{figure}[!htbp]
\centering
\begin{tikzpicture}[
  >=Latex,
  x=\linewidth,y=1cm,
  primitive/.style={
    draw=black!65,rounded corners=2pt,align=center,
    font=\footnotesize,inner sep=4pt,minimum height=9mm
  },
  implication/.style={->,line width=.8pt},
  ourimplication/.style={
    implication,draw=blue!70!black,text=blue!70!black
  },
  separation/.style={
    ->,densely dashed,line width=.9pt,draw=red!75!black
  },
  priorseparation/.style={
    separation,draw=red!55!white
  },
  resultref/.style={
    font=\scriptsize,align=center,fill=white,inner sep=1pt
  },
  cross/.style={
    fill=white,inner sep=0.5pt,text=red!75!black,font=\large
  }
]

\node[primitive,text width=.38\linewidth] (owf) at (0,3)
  {\textbf{\textsf{qOWF}}\\
   (Fixed-input determinism;\\uniform-input one-wayness)};

\node[primitive,text width=.22\linewidth] (d1) at (-.37,0)
  {\textbf{$\textsf{qOWF}_1$}\\
   Fixed-input\\determinism;\\
   Arbitrary-sampler\\one-wayness};

\node[primitive,text width=.12\linewidth] (pke) at (0,0)
  {\textbf{\textsf{qPKE}}};

\node[primitive,text width=.22\linewidth] (d2) at (.37,0)
  {\textbf{$\textsf{qOWF}_2$}\\
   Uniform-input\\determinism;\\
   Uniform-input\\one-wayness};

\node[fit=(pke),inner sep=2.5mm] (positive) {};

\node[anchor=south,font=\footnotesize\bfseries]
  at ([yshift=5mm]positive.north)
  {Quantum Cryptomania};

\node[primitive,text width=.24\linewidth] (ev) at (0,-2.2)
  {\textbf{\textsf{EV-OWPuzz}}};

\node[primitive,text width=.34\linewidth] (d3) at (0,-4.4)
  {\textbf{$\textsf{qOWF}_3$}\\
   Average-input determinism\\and one-wayness\\
   under an arbitrary sampler};

\draw[implication]
  (owf.south west) -- (d1.north);
\draw[implication]
  (owf.south east) -- (d2.north);

\coordinate (sep-level) at (0,1.65);
\coordinate (lll-start) at ([xshift=7mm]owf.south);
\coordinate (lll-entry) at ([yshift=5mm]pke.north east);
\coordinate (our-entry) at ([yshift=5mm]pke.north west);
\coordinate (our-target) at ([xshift=-7mm]owf.south);

\draw[priorseparation,rounded corners=4pt]
  (lll-start)
  -- node[cross,pos=.5,text=red!55!white] {$\times$}
     node[resultref,pos=.5,right=4pt] {\cite{LLLL25}}
  (lll-start |- sep-level) -- (.22,1.65)
  |- (lll-entry) -- (pke.north east);

\draw[separation,rounded corners=4pt]
  (pke.north west) -- (our-entry) -| (-.22,1.65)
  -- (our-target |- sep-level)
  -- node[cross,pos=.5] {$\times$}
     node[resultref,pos=.5,left=4pt]
       {Cor.\\\ref{cor:unitary-access}}
  (our-target);

\draw[separation]
  (pke.west)
  -- node[cross] {$\times$}
     node[resultref,above=7pt]
       {Cor.\\\ref{cor:unitary-access}}
  (d1.east);

\draw[separation]
  (pke.east)
  -- node[cross] {$\times$}
     node[resultref,above=7pt]
       {Cor.\\\ref{cor:unitary-access}}
  (d2.west);

\draw[implication]
  (pke.south)
  -- node[resultref,pos=.45,right=2pt] {\cite{CGG24}}
  ([xshift=-5mm]ev.north);

\draw[ourimplication]
  (ev.west) to[out=180,in=-90]
  node[resultref,pos=.55,left=2pt]
    {$+\,\textsf{qiO}$\\Thm.~\ref{thm:ev-qio-pke}}
  (pke.south west);

\draw[implication,rounded corners=6pt]
  (d1.south) -- (-.37,-4.4) -- (d3.west);
\draw[implication,rounded corners=6pt]
  (d2.south) -- (.37,-4.4) -- (d3.east);

\draw[ourimplication,<->,line width=.9pt]
  (ev.south)
  -- node[resultref,right=3pt]
       {Prop.~\ref{prop:def-three-ev}}
  (d3.north);

\end{tikzpicture}
\caption{Solid arrows denote implications, the double arrow denotes
equivalence, and crossed dashed arrows denote black-box impossibilities.
Blue arrows denote reductions proved in this work. Dark red separations
are ours; the light red separation is from prior work.
The \textsf{qOWF}-to-\textsf{qPKE} impossibility assumes perfect correctness of
\textsf{qPKE}~\cite{LLLL25}.}
\label{fig:primitive-relations}
\end{figure}
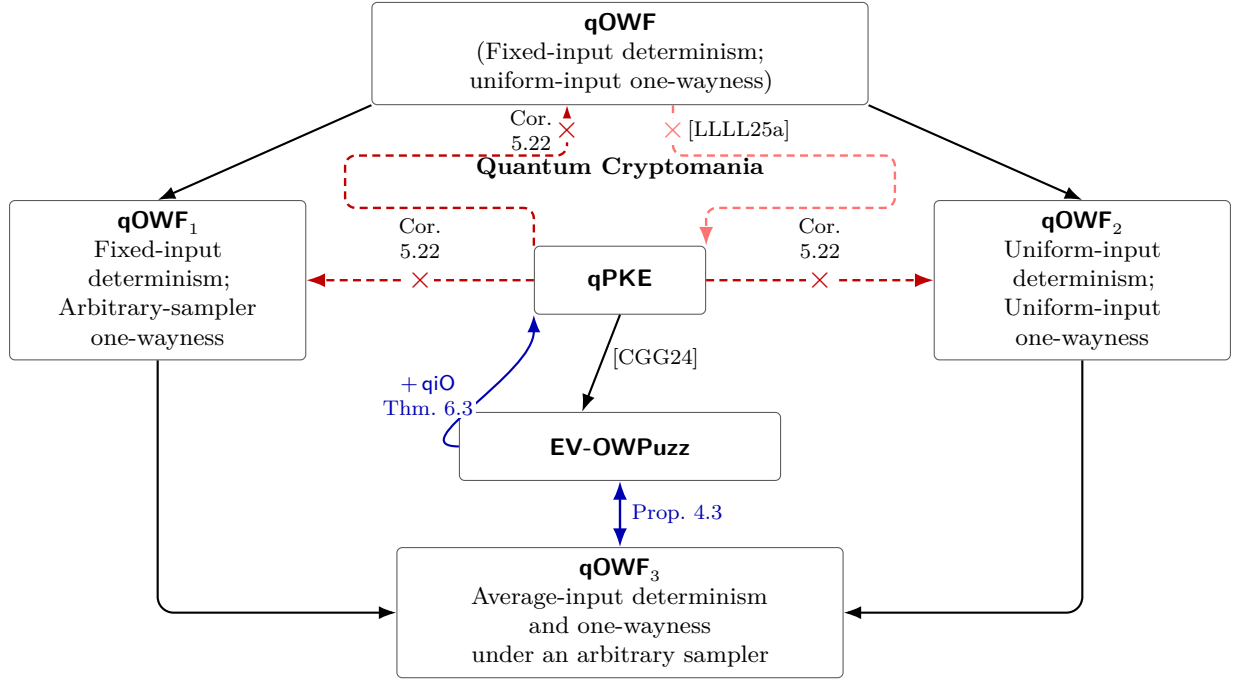
\subsection{Related works}

In the classical setting, \textsf{PKE} can be constructed from indistinguishability obfuscation and \textsf{OWF}s~\cite{SW14}. Whether \textsf{OWF}s alone suffice remains open, but we know that classical black-box constructions are ruled out by~\cite{IR89}. An optimal quadratic-query attack on classical key agreement in the random oracle model was later given in~\cite{BM09}.

Extending this impossibility to the quantum setting has required a separate line of work. Earlier results were conditional or applied to restricted constructions~\cite{ACC+22,BGVV23,LLLL24,LLLL25b}. Li, Li, Li, and Liu subsequently gave a polynomial-query attack on every perfectly correct public-key encryption scheme with classical keys in the quantum random oracle model, even when ciphertexts are quantum~\cite{LLLL25}. Their later work extends this to perfectly correct key agreement with arbitrarily many rounds of classical communication~\cite{LLLL26}. These attacks allow unbounded computation between queries and give black-box barriers to constructing public-key encryption or key agreement from \textsf{OWF}s. They do not rule out constructions with negligible correctness error in full generality. Our separation addresses the reverse question: whether \textsf{qPKE} implies \textsf{qOWF}s when the algorithms may be quantum.

Two earlier oracle separations are particularly relevant to this reverse direction. Kretschmer, Qian, and Tal construct a classical oracle relative to which $\textsf{P}=\textsf{NP}$, but trapdoor \textsf{qOWF}s, \textsf{qPKE} with classical keys and ciphertexts, and signatures with classical keys and signatures exist~\cite{KQT25}. Their oracle rules out classically computable \textsf{OWF}s, while retaining quantum-computable ones. Goldin, Morimae, Mutreja, and Yamakawa construct a unitary oracle relative to which $\textsf{BQP}=\textsf{QCMA}$ but non-interactive key exchange with classical communication exists. This rules out \textsf{qOWF}s, but the parties in their key-exchange construction retain quantum private states~\cite{GMM+24}. Our \textsf{qPKE} separation has classical public keys, secret keys, and ciphertexts, and rules out both $\textsf{qOWF}_1$s and $\textsf{qOWF}_2$s: every-input determinism with arbitrary quantum input sampling, and average determinism with uniform input sampling.

Our proof is also related to separations involving short \textsf{PRS}s and determinism. Barhoush rules out black-box constructions of \textsf{qPRG}s from short \textsf{PRS}s and from $\bot$-\textsf{qPRG}s, allowing unitary and inverse access. His black-box formulation requires the correctness of the target primitive to follow from the correctness of the source~\cite{Bar26}. Bouaziz--Ermann, Hhan, Muguruza, and Vu give a unitary oracle relative to which logarithmic-output pseudorandom function-like states exist, but \textsf{qPRG}s with negligible correctness error do not. Their proof uses a barrier theorem for products of Haar-distributed states~\cite{BEHMV26}. We use the same broad strategy of forcing a candidate to fail correctness or become simulatable. To apply this strategy to \textsf{qPKE}, we continuously change the $O_{\mathrm{samp}}$, $O_{\mathrm{enc}}$, and $O_{\mathrm{dec}}$ operations and show that a nearly deterministic evaluator computes the same function for many choices of the oracle. We then show that an oracle can be chosen for which encryption remains secure.

Quantum input sampling was studied by Barhoush, Nishimaki, and Yamakawa~\cite{BNY25}. They separate \textsf{qPRG}s with uniform seeds from \textsf{qPRF}s with quantum key sampling, using unitary oracles with inverse access. 

Kretschmer~\cite{Kre21} showed that there is a quantum oracle relative to which long \textsf{PRS}s exist but $\textsf{BQP}=\textsf{QMA}$, ruling out \textsf{qOWF}s. Black-box constructions of many-time signatures with classical secret keys and signatures from \textsf{PRS}s with at least linearly many output qubits are also ruled out~\cite{CM24}, even allowing quantum public keys; this separation does not allow inverse queries. Together with the signature construction from short \textsf{PRS}s, this rules out the corresponding black-box shortening of \textsf{PRS}s~\cite{CM24,BBO+24}. A separate oracle separation between linear-output \textsf{PRS}s and $c\log\lambda$-qubit \textsf{PRS}s, for sufficiently large $c$, is given in~\cite{BM24}.

\section{Technical Overview}
\label{sec:technical-overview}

We will give a brief overview of our separation and the construction of \textsf{qPKE} from puzzles and obfuscation. 

\subsection{Oracle Separation}
In this section, we discuss the main ideas behind our oracle separation result. For simplicity, we focus only on the separation of $\textsf{qOWF}_1$s and $\textsf{qOWF}_2$s from  \textsf{qPKE}.

Our separation exploits a key distinction between \textsf{qPKE} and \textsf{qOWF}s: \emph{determinism}. Key generation and encryption may produce different outputs each time they are run. Meanwhile, an \textsf{qOWF} evaluator must produce one fixed output on a given or at least an overwhelming fraction of the inputs. $\textsf{qOWF}_1$ requires this on every input, while $\textsf{qOWF}_2$ allows exceptions on a negligible fraction of uniformly distributed inputs.
This suggests designing oracles whose random outputs don't affect encryption, but since they can't be derandomized, any nearly deterministic evaluator can't use them non-trivially to construct an \textsf{qOWF}.

\paragraph{First attempt for the Oracles} We begin with a natural construction of oracles that admits a \textsf{qPKE} construction. At each length $n$, sample a random table, where each row indexed by $j\in\bits^n$ contains a uniformly random string $s_j\in\bits^{8n}$, and two random disjoint sets of ciphertext strings, one for each message bit. Each set has $2^{2n}$ elements from $\bits^{12n}$. We call $j$ a \emph{p-key} and $s_j$ an \emph{s-key}, reserving public key and secret key for the encryption scheme constructed from the oracle. The contents of each row are sampled independently of those of other rows. Based on such a random table, we consider three oracles:
\begin{itemize}
\item \textbf{Key sampling oracle $O_{\mathrm{samp}}$.} Prepare a uniform superposition over pairs $(j,s_j)$. Measuring this state samples a uniform p-key and its s-key. To make it unitary, we implement it as a swapping oracle that swaps the uniform superposition state with a dummy state, as used in previous works of~\cite{BMM+24,GZ25}.
\item \textbf{Ciphertext sampling oracle $O_{\mathrm{enc}}$.} Given a p-key $j$ and a message bit $b$, prepare a uniform superposition over the ciphertext set for $(j,b)$. Measuring this state samples a ciphertext string $c$.
\item \textbf{Decryption oracle $O_{\mathrm{dec}}$.} Given $(j,s_j,c)$ for an honestly sampled ciphertext string $c$, return its message bit. With an incorrect s-key, the measurement returns an unbiased bit.
\end{itemize}
Note that the three oracle operations are real, unitary, and self-inverse. These oracles naturally give a secure \textsf{qPKE} construction as follows. Key generation samples $(j,s_j)$ using $O_{\mathrm{samp}}$, publishes $j$, and keeps $(j,s_j)$ as the secret key. Encryption samples $c$ using $O_{\mathrm{enc}}(j,b)$ and outputs $(j,c)$; decryption uses $O_{\mathrm{dec}}(j,s_j,c)$. Correctness is immediate, and all keys and ciphertexts are classical. Security follows by standard quantum query hybrids. Let $J$ be the challenge public key. First, remove $(J,s_J)$ from the key-sampling state. Next, replace decryption at $J$ with the unbiased response. Finally, replace the challenge ciphertext in the ciphertext-sampling state with a fresh string, keeping the ciphertext given to the adversary unchanged. The exponentially large p-key and s-key spaces and ciphertext sets make these replacements indistinguishable to polynomial-query adversaries. In the final hybrid, the joint distribution of the challenge ciphertext and the oracle is independent of the message bit, thus completing the security proof.
Next, to rule out \textsf{qOWF}, we also add a \textsf{PSPACE} oracle, i.e., an oracle for a fixed \textsf{PSPACE}-complete language, which has \emph{no access to the random tables}, in order to preserve the security of our \textsf{qPKE} scheme. 

\paragraph{The \textsf{qPKE} construction also provides an \textsf{qOWF}.}
The construction above gives a secure \textsf{qPKE}, but unfortunately it already yields a $\textsf{qOWF}_1$. In particular, the oracle above allows us to efficiently verify if a given s-key, p-key pair is indeed a valid honestly sampled pair. This is because encrypting message one with the p-key and then decrypting the ciphertext with the s-key returns one with probability exactly $1$ if they are a sampled key pair and $1/2$ for every other pair. Hence, repeating this process sufficiently many times (say, $\lambda$ times) allows us to distinguish sampled key pairs from other pairs, which lets us evaluate the following $\textsf{qOWF}_1$:
 $f(j,s)=
 \begin{cases}
 (0,j),&s=s_j,\\
 \bot,&s\ne s_j.
 \end{cases}$ It is easy to check that the above verification process leads to deterministic evaluation (up to negligible error) for every input. In fact, the determinism holds even if we change the distribution on the oracle tables. For one-wayness, we consider the quantum input sampler that samples $(j,s_j)$ according to the \textsf{qPKE} scheme's key-generation algorithm. By the correctness of the \textsf{qPKE} scheme, the image of a typical $(j,s_j)$ under $f$ is $(0,j)$. This implies that any inverter against $f$ with respect to this quantum input sampler essentially computes the secret key $(j,s_j)$ given a sample public key $j$ for the \textsf{qPKE} scheme, thus breaking the \textsf{qPKE} security.  
Moreover, this correspondence holds even if we change the oracle distribution! Thus, for our separation, we must tweak the oracle operations themselves, not just the distribution over the oracle parameters (such as the random tables).

\paragraph{Oracles with two tables.} We tweak the oracle operation by introducing a second table, so that the above-mentioned evaluator for $f$ can fail fixed-input determinism while the \textsf{qPKE} construction remains correct and secure. In particular, at each length $n$, sample two independent tables $A,B$ with the same p-key indices, each assigning its own s-key and ciphertext sets to every p-key. Independently sample uniform $i\in\{1,2,3\}$ and $t\in[0,1]$, and set $\tau=(i-1+t)/3$, which is uniform in $[0,1]$. Informally, the oracle operations use an interpolation of the two tables determined by $\tau$. These choices are made once, independently at each length, and remain fixed throughout every execution.
At $\tau=0$, the oracle operations use only $A$; at $\tau=1$, they use only $B$. For intermediate values, sampling can use weighted superpositions of the states specified by $A$ and $B$. For example, at $\tau=1/2$, conditioned on obtaining p-key $j$, the s-key register is in the state $(\ket{s_j^A}+\ket{s_j^B})/\sqrt2$ assuming $s_j^A,s_j^B$ are distinct. The precise dependence of all three operations on $\tau$ is given in Section~\ref{sec:base-oracle}.
The \textsf{qPKE} security argument above also applies to this two-table construction using similar arguments as above\footnote{after ignoring the event of overlapping ciphertext sets at the challenge p-key (which happens with negligible probability)}; see Lemma~\ref{lem:fresh-record-pke}.
Crucially, fix a p-key $j$ with distinct s-keys and disjoint ciphertext sets across the two tables. For $\tau\in[0,1/3]$, both sampling operations use $A$, and honestly sampled keys decrypt correctly. Yet, decrypting an $A$-encryption of one with $s_j^B$ returns one with a probability $p(\tau)$ that varies continuously from $1/2$ to $1$\footnote{provided their ciphertext sets at $j$ are disjoint, which happens except with negligible probability}. Thus, decrypting a message under a p-key using a given s-key can no longer separate the sampled s-key from the s-key from the other table $B$. As a result, for any polynomial number of repetitions, there is an inverse-polynomial fraction of values of $\tau$ for which the above evaluator outputs $(0,j)$ and $\bot$ each with probability at least $1/3$ on the fixed input $(j,s_j^B)$; hence, the above evaluation algorithm for $f$ fails fixed-input determinism.

\paragraph{Ruling out all \textsf{qOWF}s by conditioning the oracle distribution.}
However, only ruling out the particular $\textsf{qOWF}_1$ $f$ is not enough; we need to rule out \emph{every} $\textsf{qOWF}_1$ and every $\textsf{qOWF}_2$. We construct it by \emph{conditioning}: sample only oracles conditioned on satisfying some event, such that no $\textsf{qOWF}_1$ or $\textsf{qOWF}_2$ exists conditioned on this event.\footnote{With the original one-table operations, the $\textsf{qOWF}_1$ above would remain valid whenever \textsf{qPKE} remained secure, regardless of the conditioning.} Bouaziz--Ermann et al.~\cite[Section~5.2]{BEHMV26} use this approach to rule out candidate quantum-evaluable pseudorandom generators, considering each candidate infinitely often and ensuring either a correctness failure or a successful attack. However, we need to be careful, since if we condition on an event with too small a weight, then after conditioning on such an event, the \textsf{qPKE} security may not hold. For example, conditioning on every s-key in both tables being the all-zero string makes every honestly generated secret key publicly known. Applying this restriction at infinitely many lengths destroys \textsf{qPKE} security.

There are two ways to rule out a candidate evaluator: either we condition on those oracle parameters for which it is not deterministic enough to satisfy the correctness condition of $\textsf{qOWF}_1$ or $\textsf{qOWF}_2$ and we condition on that event, or we show somehow that the evaluator can be inverted using the \textsf{PSPACE} oracle. Thus one option is that we categorize each evaluator into two categories, according to the determinism requirement with respect to both $\textsf{qOWF}_1$ and $\textsf{qOWF}_2$, averaged over the oracle randomness: 
\begin{enumerate}
    \item the first category where the failure probability of determinism, i.e., the probability over the oracles that the \emph{determinism} does not hold, is at least a threshold $t$, in which case we condition on the event that they are not deterministic.\footnote{The notion of determinism depends on whether we consider $\textsf{qOWF}_1$ or $\textsf{qOWF}_2$.} 
    \label{it:first-category}
    \item the second category where this failure probability is smaller than $t$ with respect to either $\textsf{qOWF}_1$ or $\textsf{qOWF}_2$.\label{it:barrier-theorem-agreement} 
\end{enumerate}
If we can prove that, using a PSPACE oracle, we can invert all \textsf{qOWF}s corresponding to the evaluators in the second category, then we can rule out all $\textsf{qOWF}_1$s and $\textsf{qOWF}_2$s.
However, the PSPACE oracle does not have access to the oracle data, namely $(A,B,\tau)$. So it seems like the only way we can rule out evaluators in the second category is to show that these evaluators do not use the tables in a non-trivial way! Our intuition here is that since any such evaluator in the second category is deterministic according to the requirements of either $\textsf{qOWF}_1$ or $\textsf{qOWF}_2$, it agrees on a function for an overwhelming fraction of the inputs. Hence, since the three oracles $O_{\mathrm{samp}},O_{\mathrm{enc}},O_{\mathrm{dec}}$ are randomized, any such evaluator satisfying such an agreement should not depend on the oracles too strongly, i.e., the evaluator computes the \emph{same function} under a large fraction of the oracles.

We formalize this intuition and prove a general result which can be seen next. This can be seen as a generalization of the barrier theorem proven by Bouaziz--Ermann et al.~\cite[Theorem~5.3 and Lemma~5.8]{BEHMV26}, who obtain a similar agreement albeit on output \emph{bits} in the common Haar random model. Hence we call it the general barrier theorem.
\paragraph{The Generalized Barrier Result.}
Consider a distribution on some space of tables $\mathcal{T}$, and say we consider a map encoding a given table $A$ to an oracle $\mathcal{O}(A)\in \cal{X}$, where $\cal X$ is a space of oracles. Finally, suppose there is a way to get interpolation between two tables $A,B\in \mathcal{T}$ to get an interpolated oracle, using an interpolation parameter $\tau\in [0,1]$ represented as $\Interp(A,B,\tau)$ such that $\Interp(A,B,0)=\mathcal{O}(A)$ and $\Interp(A,B,1)=\mathcal{O}(B)$. Moreover, suppose changing $\tau$ changes the interpolated oracle $\Interp(A,B,\tau)$ at a bounded speed $v_0$, i.e., for any $s,t\in [0,1]$ the distance (as per operator norm) between the oracles is bounded by $|\Interp(A,B,s)-\Interp(A,B,t)|\leq v_0|s-t|$. (Later, we show that our oracles also change in a bounded fashion as we vary $\tau$; see Lemma~\ref{lem:edge-properties}.) Then, we show:

\begin{theorem}[Formalized in Corollary~\ref{cor:barrier-large-set}]\label{thm:general-barrier-informal}
Let $\nu$ be a distribution on the set of tables $\cal T$ and draw $A,B$ independently from $\nu$, and independently sample $\tau$ from $[0,1]$ uniformly. Let $\mu$ be the resulting distribution $\Interp(A,B,\tau)$ obtained on the oracles. Suppose
\begin{enumerate}
    \item $S_1,\ldots,S_r$ are sets of oracles that are pairwise separated by distance at least $\delta>0$,
    \item a sample $O$ drawn from $\mu$ lies outside their union with probability $\varepsilon_{\mathrm{bad}}\leq1/2$,
\end{enumerate}
 then there exists a set $S_i$, such that the probability that a sample from $\mu$ falls in $S_i$ is at least $1-\left(1+\frac{6v_0}{\delta}\right)\varepsilon_{\mathrm{bad}}$.

\end{theorem}
In our case, we apply the barrier theorem to rule out each evaluator from the second category (see \Cref{it:barrier-theorem-agreement}) according to each of the two definitions. The sets of oracles $S_1,\ldots,S_r$ are defined (separately for each of the two definitions, as we describe in the next paragraph). In the case of $\textsf{qOWF}_1$, each set $S_i$ consists of all the oracles that yield the same function under the evaluator. Hence, for our application, it is going to be doubly exponential in the input length of the functions. Fortunately, the bound in the theorem is \emph{independent} of the number of sets, and thus does not become meaningless even for such a large number of sets. This is one of two axes of generalization compared to the barrier theorem in ~\cite{BEHMV26}, which we need for our separation. The second axis of generalization is that our result holds for an arbitrary metric space $\cal X$, space $\cal T$, and functions $\cal{O}$, $\Interp$ satisfying the conditions above, and with respect to any distribution $\nu$, without needing the specific properties of the Haar random distribution. 

\paragraph{Invoking the barrier theorem.} Next, we complete the separation proof, assuming the generalized barrier theorem. To apply the above generalized barrier theorem, we define the sets of oracles $S_1,\ldots,S_r$ for the two definitions separately. For $\textsf{qOWF}_1$, fix an evaluator $E$ and the security parameter, keeping the operations at short lengths fixed. Let $S_f$ contain the oracles for which $E$ outputs $f(x)$ with probability at least $5/6$ on every input $x$. Since different functions force an output-probability gap of at least $2/3$ at some input $x$, these sets are separated by $1/(3q)$ in the maximum operator-norm distance between corresponding oracle operations, where $q$ denotes the number of queries that the evaluator makes. In the second category (see \Cref{it:barrier-theorem-agreement}) with respect to $\textsf{qOWF}_1$, their union has probability at least $1-t$, so Theorem~\ref{thm:general-barrier-informal} gives a set $S_g$ of oracles with probability weight at least $1-O(qt)$ (Lemma~\ref{cor:whole-function}) such that the evaluator computes the same function on all inputs under any oracle in $S_g$.

Next, after estimating the fixed short operations, and then if the sampled oracle lies in the corresponding large set $S_g$ (which happens with probability at least $1-O(qt)$), the \textsf{PSPACE} oracle can find a preimage (Lemmas~\ref{lem:mean-recovery} and~\ref{lem:strict-inversion}) by averaging over oracles from the base distribution on the oracles, without needing the real oracles. Thus, at each step of our conditioning, which corresponds to a particular evaluator, either the evaluator is of the first category (see \Cref{it:first-category}) and we condition on an event that forces a correctness failure, or we show that the evaluator can be inverted except with a small probability $O(qt)$ (and we do not condition). With an appropriate choice of the parameters, these small probabilities indicating the potential failure to invert at each of the countably many steps have a finite sum. Thus, after conditioning over all countably many steps, we get a distribution that, by the Borel-Cantelli theorem, satisfies that with probability one, no $\textsf{qOWF}_1$ exists (Theorem~\ref{thm:oracle-selection}).

For $\textsf{qOWF}_2$, instead fix an input $x$ and let $S_{x,y}$ contain the oracles for which $\Pr[E^{\mathcal O}(x)=y]\geq5/6$; it is easy to see that these sets also satisfy the hypothesis for the generalized barrier theorem. This lets us invoke \Cref{thm:general-barrier-informal} separately for each reliable input, yielding an agreement guarantee with one function for most oracles, on most inputs. 
  For each evaluator $E$, we choose sufficiently spaced security parameters $1\leq\lambda_1<\lambda_2<\cdots$: at $\lambda_1$, we sample the new oracle data that $E(1^{\lambda_1})$ may query, after potential conditioning with respect to Definition~\ref{it:def-1}; at $\lambda_2$, we do the same with respect to Definition~\ref{it:def-2}, keeping previously sampled data unchanged; and we continue alternating. This way, after countably many such steps, we obtain a distribution under which, with probability one, each evaluator is either invertible at infinitely many security parameters or fails the required determinism at infinitely many security parameters, for each definition separately; see \Cref{thm:oracle-selection}.\footnote{If we consider $\textsf{qOWF}_3$ as discussed before, the above agreement guarantee cannot be achieved because, for $\textsf{qOWF}_3$, there is no suitable definition of the separated sets $S_1,\dots,S_r$ using which we can invoke \Cref{thm:general-barrier-informal}.}

\paragraph{\textsf{qPKE} security after conditioning.}
It remains to show that conditioning the oracle distribution as above preserves \textsf{qPKE} security. We measure the accumulated effect of the conditioning using relative entropy between the conditioned distributions and the base distribution, and choose the parameters so that this quantity stays polynomially bounded on all oracle lengths an efficient adversary can query; in particular, we show that the relative entropy of the conditioned distribution with respect to the base distribution, restricted to oracle lengths at most $m$, is bounded by $2m$. Next, we show that using Pinsker's inequality, this bound allows us to replace the data at a hidden random subset of p-keys containing the challenge p-key by fresh samples from the base distribution, with negligible effect on the adversary's success probability\footnote{The size of the random subset is chosen to be exponential yet negligible as a fraction of the $p$-keys.}; see \Cref{lem:conditioned-pke,lem:hidden-subset-entropy}. Note that the secret data at the challenge public key are now independent of the remaining oracle data, hence the earlier \textsf{qPKE} security argument (shown in \Cref{lem:fresh-record-pke}) applies (see Lemma~\ref{lem:conditioned-pke}).
Using Borel--Cantelli, this gives security with probability one under the conditioned distribution. Moreover, with some positive probability, the two tables' ciphertext sets are disjoint at every p-key and length. Therefore, we conclude that there must exist a fixed oracle with perfectly correct, secure \textsf{qPKE} but neither $\textsf{qOWF}_1$s nor $\textsf{qOWF}_2$s exist (Corollary~\ref{cor:selected-oracle-separation}).

\paragraph{Proving the barrier theorem.} We prove the general barrier theorem using a \emph{path-based framework}; we show that for any independently sampled oracles $O_1,O_2$ from $\mu$, we can connect them using a random path of oracles such that each oracle in the path has marginal distribution $\mu$; the detailed path construction is given in the proof of \Cref{lem:stationary-paths}. Then we show the contrapositive of the above theorem by proving that if no single set $S_i$ has probability weight large enough under $\mu$, then since the sets are mutually separated, the path must spend a sufficient amount of time, i.e., contain a large number of oracles outside the union of the sets $S_1,\ldots,S_r$. However, we know that each oracle in the path is marginally distributed according to $\mu$, and that most of the probability weight of $\mu$ lies in the union of the sets; the expected number of oracles in the path, or the expected amount of time spent outside the union of the sets, must be small, thus arriving at a contradiction. The formal argument is a bit more delicate as the random path is continuous, but the same intuition holds by considering the path to be indexed by continuous time, and then concluding the contradiction on the amount of time spent outside the union of the sets; see \Cref{sec:barrier} for more details.

\subsection{Public-key encryption from puzzles and \textsf{qiO}}

Our starting point is the classical construction of \textsf{OWF}s together with \textsf{iO} implying public-key encryption~\cite{SW14}. The goal is to replace the role of the \textsf{OWF} assumption with \textsf{EV-OWPuzz}.
\paragraph{The qPKE construction.}
We first describe the construction with classical \textsf{iO}. Given such an obfuscation, the idea behind the \textsf{PKE} scheme is to publish an obfuscation of $P_u$ for a uniformly random $u$ as the public key and keep $u$ as the secret key. Here $P_u$ denotes a function which outputs $1$ if the input is $u$ and $0$ elsewhere. The encryption of a message bit $b$ is to simply output an obfuscated program as the ciphertext, where the circuit to be obfuscated returns the message bit when the public-key program accepts its input and zero otherwise, so evaluating it at the secret key $u$ recovers the message. Hence, decryption is just evaluating the cipher-text program at the secret key. It is easy to see that correctness follows from the correctness of the obfuscation scheme. Further it is easy to see that for security, we only need the additional point-hiding property from the $\textsf{iO}$ scheme that the obfuscation of random point function $P_u$ for a uniformly random $u$ is indistinguishable from the obfuscation of the zero circuit $Z$. However, to prove such a property for a $\textsf{iO}$ scheme, we would need an additional hardness assumption like one-way functions, since this cannot follow from \textsf{iO} security unconditionally as \textsf{iO} can even exist when P=NP. A recent work of~\cite{MSY26} provides such a construction of \textsf{PKE} based on \textsf{iO} using the complexity theoretic assumption in the proof. Note that the same template would work to give $\textsf{qPKE}$ even if we can get such a point-hiding guarantee for \textsf{qiO}. Our approach is to prove such a guarantee for \textsf{qiO} but by only relying on the requirement of \textsf{EV-OWPuzz} instead of one-way functions.\footnote{Here we note that the definition of \textsf{qiO} does allow obfuscating quantum circuits implementing classical functions. This is crucial because we would like to use the quantum algorithms of the \textsf{EV-OWPuzz} inside the obfuscated program.}

\paragraph{The obstacle and the suppression lemma.}
In the classical case, the starting point of the proof is that \textsf{iO} is a single-point differing-inputs obfuscation (\textsf{diO}), i.e., a non-negligible distinguisher between obfuscations of two circuits differing on one input allows that input to be recovered from the circuit descriptions~\cite[Section~6.1]{BCP14}. In the proof, we can combine this property with an injective keyed one-way function (which can be built from \textsf{iO} and one-way functions) due to the deterministic and injective property, to ensure that finding a distinguishing input for two 1-point differing inputs, would lead to inverting the keyed one-way function to the unique answer. However, an \textsf{EV-OWPuzz} verifier may accept many answers with intermediate acceptance probabilities so there is no \emph{unique} answer, which would violate the condition of requiring the circuits to differ at a single point. Thus we need a way to take an \textsf{EV-OWPuzz} verifier and suppress all but a unique accepting answer for it.

 In this work, we show that such a suppression lemma indeed holds.
\begin{lemma}[Informal suppression lemma; formalized in Lemma~\ref{lem:ev-isolation}]
For any \textsf{EV-OWPuzz} scheme, given a puzzle $p$ with an answer that the verifier accepts with probability at least $2/3$, an efficient randomized procedure constructs a quantum circuit $W$ based on the puzzle $p$ such that, with inverse-polynomial probability, there is an answer $a^\star$ satisfying $\Pr[W(a^\star)=1]\geq1-\negl(\lambda)$, for every $a\neq a^\star$, $
 \Pr[W(a)=1]\leq\negl(\lambda)$ and the \textsf{EV-OWPuzz} verifier still accepts $a^\star$ for $p$ with probability greater than $1/3$.
\end{lemma}

To prove the point-hiding property for \textsf{qiO} based on the suppression lemma, note that on a successful trial, the circuit produced by suppression implements a point function at an answer $a^\star$. Masking its input by a uniform $r$ makes the accepting point $u=a^\star\oplus r$ uniform. Finally, an extension of the extraction argument from~\cite[Section~6.1]{BCP14} turns a distinguisher between obfuscations of a random point and zero circuits, with a non-negligible advantage, into a procedure recovering $a^\star=u\oplus r$ with inverse-polynomial probability. This contradicts the \textsf{EV-OWPuzz} security since $a^\star$ is still accepted by the \textsf{EV-OWPuzz} verifier with probability at least $1/3$, proving the point hiding guarantee of \textsf{qiO} (see Lemma~\ref{lem:ev-point-hiding}).

\paragraph{Proving the suppression lemma.}
To prove the suppression, we first amplify the verifier by standard independent repetition so that answers originally accepted with probability at least $2/3$ are accepted with overwhelming probability, while those accepted with probability at most $1/3$ are accepted with negligible probability. Next, we guess a hash output length, sample a random hash function, and retain only answers hashing to zero. With inverse-polynomial probability, this hash bucket contains one answer whose acceptance probability exceeds that of every other retained answer by an inverse-polynomial amount. Next, conditioned on this event, we utilize this inverse-polynomial gap to amplify the acceptance probability of the unique answer to $1-$ negligible while on all other inputs, the accepting probability plunges to negligible, via independent repetitions and a suitable choice of threshold on the number of accepting trials for acceptance. For more details on how to choose random hash function and the threshold efficiently, see the proof of Lemma~\ref{lem:ev-isolation}.

\section{Definitions}
\label{sec:definitions}

We collect the conventions and primitive definitions.

\subsection{Conventions}
\label{sec:conventions}

We use a lowercase prefix \textsf{q} for the quantum-computable variant of a classical primitive, as in \textsf{qOWF} and \textsf{qPKE}.

All algorithms and adversaries in the work are uniform quantum polynomial time (QPT). Their classical input and output lengths have fixed polynomial bounds on the security parameter. A function $\eta:\mathbb N\to[0,1]$ is negligible if, for every positive integer $c$, $\eta(\lambda)\leq\lambda^{-c}$ for all sufficiently large $\lambda$. We write $\negl(\lambda)$ for such a bound. We use $[d]=\{1,\ldots,d\}$, identify $[2^n]$ with $\bits^n$, and write $U_m$ for the uniform distribution on $\bits^m$.

We use \textsf{P}, \textsf{NP}, and \textsf{PSPACE} for deterministic polynomial time, nondeterministic polynomial time, and deterministic polynomial space, respectively. We also use bounded-error quantum polynomial time (\textsf{BQP}), quantum Merlin--Arthur (\textsf{QMA}), and quantum classical Merlin--Arthur (\textsf{QCMA}).

\paragraph{Oracle access and gate conventions.} Algorithms with oracle access in this work may query the oracle in superposition, and have forward, inverse, transpose, conjugate, and controlled access. 

\paragraph{Probability distances.}
For distributions $Q,P$, define their relative entropy by $D(Q\Vert P)=\int\log(dQ/dP)\,dQ$ when $Q$ has a density relative to $P$, and set it to infinity otherwise. Here $\log$ denotes the natural logarithm function. The statistical distance, or equivalently, the total variation distance, is $\SD(Q,P)=\sup_A|Q(A)-P(A)|$, where the supremum is over measurable events. 
We use the chain rule for relative entropy, its nonnegativity, and Pinsker's inequality, as recalled below. 
First, we recall the chain rule for relative entropy, expressed in terms of conditional distributions.
\begin{lemma}[Chain rule for relative entropy{~\cite[Section~2.2, Lemma~3]{CP15}}]
\label{lem:kl-chain-rule}
For joint distributions $\mathcal R_{UV}$ and $\mathcal S_{UV}$ with finite relative entropy,
\[
D(\mathcal R_{UV}\Vert\mathcal S_{UV})
=
D(\mathcal R_U\Vert\mathcal S_U)
+
\mathop{\mathbb E}\limits_{u\leftarrow\mathcal R_U}
D\!\left(
\mathcal R_{V\mid U=u}
\middle\Vert
\mathcal S_{V\mid U=u}
\right).
\]
\end{lemma}

Next we note the following standard fact about the relative entropy. Note that no independence is required under $\mathcal R$.
\begin{lemma}[Relative entropy to a product distribution{~\cite[Section~2.2, Lemma~5]{CP15}}]
\label{lem:kl-product}
Let $\mathcal R$ be a joint distribution on $X_1,\ldots,X_m$, with marginals $\mathcal R_1,\ldots,\mathcal R_m$. For any product distribution $\mathcal S=\prod_{j=1}^m\mathcal S_j$,
\[
\sum_{j=1}^m D(\mathcal R_j\Vert\mathcal S_j)
\leq
D\!\left(\mathcal R\middle\Vert\prod_{j=1}^m\mathcal S_j\right).
\]
\end{lemma}
Finally, we recall Pinsker's inequality. 
\begin{lemma}[Pinsker's inequality{~\cite[Theorem~5.15]{watrous}}]
\label{lem:pinsker}
For distributions $\mathcal R,\mathcal S$, it holds that
\[
\operatorname{SD}(\mathcal R,\mathcal S)
\leq
\sqrt{\frac{D(\mathcal R\Vert\mathcal S)}{2}}.
\]
\end{lemma}

\subsection{One-way functions}\label{sec:def-owf}

A quantum computable one-way function or \textsf{qOWF} is given by a QPT evaluator $E(x)$ with $m(\lambda)$-bit inputs and $\ell(\lambda)$-bit classical outputs. Define $f(x)$ to be the most likely output of $E(x)$, with ties broken lexicographically, and let
\[
 e(x)=\Pr[E(x)\ne f(x)].
\]

\begin{definition}[$\textsf{qOWF}_1$: Every-input determinism]
\label{def:one}
A $\textsf{qOWF}_1$ $E$ is associated with a QPT input sampler $S$ such that it satisfies one-wayness as defined below with respect to $S$. Moreover, it satisfies that there is a negligible function $\eta(\lambda)$ such that
\[
 \forall {x\in\bits^{m(\lambda)}}, e(x)\leq\eta(\lambda).
\]
\end{definition}
Setting the sampler $S$ to be the uniform random distribution recovers the standard definition of \textsf{qOWF}, as considered in past works; see Refs.~\cite{AC02,CLS01,DMS00,GKT26,HY20,IT03a,IT03b,KKK+05,KQT25,MY23,OTU00}.

\begin{definition}[$\textsf{qOWF}_2$: Average determinism on uniform inputs]
\label{def:two}
A $\textsf{qOWF}_2$ $E$ satisfies the one-wayness property as defined below with respect to   $S=U_{m(\lambda)}$, and also satisfies
\[
 \mathop{\mathbb E}\limits_{X\leftarrow U_{m(\lambda)}}e(X)\leq\eta(\lambda)
\]
for some negligible function $\eta(\lambda)$.
\end{definition}

We introduce the following \emph{new notion} for this paper, which is implied by the previous two notions. It permits arbitrary quantum input sampling while requiring determinism only on average over that sampler. We refer to it explicitly as $\textsf{qOWF}_3$ to distinguish it from the stronger notions above.

\begin{definition}[$\textsf{qOWF}_3$: Average determinism with arbitrary sampling]
\label{def:three}
A $\textsf{qOWF}_3$ $E$ is associated with a QPT sampler $S$, such that it satisfies the one-wayness as defined below with respect to $S$. Moreover, it satisfies
\[
 \mathop{\mathbb E}\limits_{X\leftarrow S(1^\lambda)}e(X)\leq\eta(\lambda)
\]
for one negligible $\eta$.
\end{definition}

\paragraph{One-wayness.} Given a QPT sampler $S$ of classical inputs, we say one-wayness holds if, for every uniform QPT inverter $I$,
\[
\Pr_{\substack{
X\leftarrow S(1^\lambda)\\
Y\leftarrow E(X)\\
X'\leftarrow I(Y,1^\lambda)
}}
\bigl[E(X')=Y\bigr]
\leq \operatorname{negl}(\lambda).
\]
Here we assume that all malformed inputs are mapped to $\bot$ by $E$.

\subsection{Pseudorandom Generators}

We recall the notion of quantum pseudorandom generators with uniform input sampling.

\begin{definition}[Quantum pseudorandom generator]
\label{def:qprg}
A \textsf{qPRG} is a QPT algorithm $G$
with classical seeds of length $r(\lambda)$ and classical outputs
of polynomial length $s(\lambda)\ge r(\lambda)+1$, satisfying:

\begin{enumerate}
\item \textbf{Correctness.}
There exist deterministic functions
$g_\lambda:\bits^{r(\lambda)}\to\bits^{s(\lambda)}$ such that
\[
 \Pr_{X\leftarrow\bits^{r(\lambda)}}
 \bigl[G(X)\ne g_\lambda(X)\bigr]
 \le \negl(\lambda).
\]

\item \textbf{Pseudorandomness.}
For every QPT distinguisher $D$,
\[
 \left|
 \Pr_{X\leftarrow\bits^{r(\lambda)}}
 \bigl[D(G(X))=1\bigr]
 -
 \Pr_{Y\leftarrow\bits^{s(\lambda)}}
 \bigl[D(Y)=1\bigr]
 \right|
 \le \negl(\lambda).
\]
\end{enumerate}
\end{definition}

\subsection{Efficiently verifiable one-way puzzles}

We recall a microcrypt primitive known as \textsf{EV-OWPuzz}, which was  introduced in \cite{CGG24}.

\begin{definition}[Efficiently verifiable one-way puzzle]\label{def:ev-puzzle}
An efficiently verifiable one-way puzzle (\textsf{EV-OWPuzz}) consists of two QPT algorithms: a sampler $\mathsf{Samp}$ and a verifier $V$. The sampler outputs a classical puzzle and answer $(p,a)\leftarrow\mathsf{Samp}(1^\lambda)$, and the verifier $b\gets V(p,a)$ outputs a bit. Write $v_p(a)=\Pr[V(p,a)=1]$. Then, 
\[
\mathop{\mathbb{E}}\limits_{(p,a)\leftarrow\mathsf{Samp}(1^\lambda)}[1-v_p(a)]\leq\negl(\lambda).
\]
And, for every QPT algorithm $A$, 
$$\Pr_{(p,a)\leftarrow\mathsf{Samp}(1^\lambda)}[V(p,A(p))=1]\leq \negl(\lambda).$$ 
\end{definition}

\subsection{Public-key encryption}

We recall the notions of public-key encryption with quantum computation and classical communication. 

\begin{definition}[Public-key encryption]\label{def:pke}
A public-key encryption scheme with quantum computation and classical communication (\textsf{qPKE}) consists of three QPT algorithms with classical inputs and outputs:
\begin{itemize}
\item $\mathsf{Gen}(1^\lambda)$ outputs a public key $\mathsf{pk}$ and a secret key $\mathsf{sk}$.
\item $\mathsf{Enc}(\mathsf{pk},m)$ takes a message $m\in\mathcal M_\lambda$ and outputs a ciphertext $c$.
\item $\mathsf{Dec}(\mathsf{sk},c)$ outputs a message in $\mathcal M_\lambda$ or $\bot$.
\end{itemize}
Here $\mathcal M_\lambda$ is the message space, whose elements have polynomially bounded length.

\emph{Correctness.} For every $m\in\mathcal M_\lambda$,
\[
 \Pr_{\substack{
   (\mathsf{pk},\mathsf{sk})\leftarrow\mathsf{Gen}(1^\lambda)\\
   c\leftarrow\mathsf{Enc}(\mathsf{pk},m)}}
 \bigl[\mathsf{Dec}(\mathsf{sk},c)=m\bigr]
 \ge 1-\negl(\lambda).
\]

\emph{Security (indistinguishability under chosen-plaintext attacks, \textsf{IND-CPA}).} For every QPT adversary $A=(A_1,A_2)$,
\[
 \left|
 \Pr\!\left[
 b'=b
 \;\middle|\;
 \begin{array}{l}
   (\mathsf{pk},\mathsf{sk})
      \leftarrow\mathsf{Gen}(1^\lambda)\\
   (m_0,m_1,\rho)
      \leftarrow A_1(\mathsf{pk})\\
   b\leftarrow\{0,1\}\\
   c\leftarrow\mathsf{Enc}(\mathsf{pk},m_b)\\
   b'\leftarrow A_2(\mathsf{pk},c,\rho)
 \end{array}
 \right]
 -\frac12
 \right|
 \le \negl(\lambda).
\]
The adversary chooses equal-length messages $m_0,m_1\in\mathcal M_\lambda$ and may retain a quantum state $\rho$ between its two stages. 
\end{definition}

\subsection{Point-function obfuscation}

For a point $a\in\bits^\lambda$, write $P_a(x)=\mathbf 1\{x=a\}$, the function that maps every input to $0$ except for $a$ which is mapped to $1$. Also, let $O_a\ket{x,z}=\ket{x,z\oplus P_a(x)}$. We now define a point-function obfuscator using virtual black-box guarantee.

\begin{definition}[Point-function obfuscation]
\label{vbb:def}
A point-function obfuscator consists of two uniform QPT algorithms with classical inputs and outputs:
\begin{itemize}
\item $\Obf(1^\lambda,a)$ takes a point $a\in\bits^\lambda$ and outputs a classical program $p$.
\item $\Eval(p,x)$ takes a program $p$ and an input $x\in\bits^\lambda$, and outputs a bit.
\end{itemize}

\emph{Correctness.} There is a negligible function $\eta$ such that, for every $\lambda$,
\[
 \max_{a\in\bits^\lambda}
 \mathop{\mathbb E}\limits_{p\leftarrow\Obf(1^\lambda,a)}
 \max_{x\in\bits^\lambda}
 \Pr[\Eval(p,x)\ne P_a(x)]
 \leq\eta(\lambda).
\]
The maximum over evaluation inputs is taken before averaging over programs.

\emph{Security (single-program virtual black-box).} For every uniform QPT adversary $\mathcal A$ with a one-bit output, there exist a uniform QPT simulator $\mathcal S$ and a negligible function $\varepsilon$ such that, for every $\lambda$ and every point $a\in\bits^\lambda$,
\[
 \left|\Pr[\mathcal A(
                  \Obf(1^\lambda,a))=1]
       -\Pr[\mathcal S^{O_a}(1^\lambda)=1]\right|
 \leq\varepsilon(\lambda).
\]
The same simulator works for all points and accesses the hidden point only through quantum queries to $O_a$. 
\end{definition}

\subsection{Quantum indistinguishability obfuscation}
\label{sec:qio-definition}

We consider quantum indistinguishability obfuscation (\textsf{qiO}) for quantum circuits that have an efficient classical encoding and with classical input/outputs.

\begin{definition}[Quantum Indistinguishability Obfuscation]\label{def:semantic-qio}
A \textsf{qiO} obfuscator consists of QPT algorithms $(O,E)$ with the following syntax:
\begin{itemize}
\item $O(1^\lambda,Q)$ takes a classical description of a quantum circuit $Q$ and returns a classical string $c$.
\item {$E(c,x)$ takes the obfuscated program $c$ and a classical input $x$ of the corresponding circuit's input length, and outputs a classical string.}
\end{itemize}

{A family $Q$ is an admissible implementation of a total classical function $f$ if}
\[
\sup_x\Pr[Q(x)\ne f(x)]\leq\epsilon(\lambda)
\]
for a negligible function $\epsilon$. The obfuscator has the following properties.
\begin{enumerate}
\item \textbf{Correctness.} For every admissible family, the error $\Pr_{c\leftarrow O(1^\lambda,Q),E}[E(c,x)\ne f(x)]$ is negligible, uniformly over all inputs $x$.
\item \textbf{Security.} If two equal-size circuit families with the same input and output lengths admissibly implement the same total classical function, their obfuscations are computationally indistinguishable, even in the presence of arbitrary polynomial-length classical auxiliary information.
\end{enumerate}
\end{definition}

\begin{remark}[Simultaneous correctness]\label{rem:qio-simultaneous}
Any obfuscator satisfying Definition~\ref{def:semantic-qio} can be amplified to satisfy the following stronger correctness guarantee: for every admissible family $Q$ implementing $f$ and all sufficiently large $\lambda$,
\[
\Pr_{c\leftarrow O(Q)}\!\left[
 \max_x\Pr_E[E(c,x)\ne f(x)]>2^{-\lambda}
\right]\leq2^{-\lambda}.
\]
To obtain this guarantee, start with an obfuscator $(O_0,E_0)$ satisfying Definition~\ref{def:semantic-qio}, and compile a program with $n$ input bits and $\ell\geq1$ output bits independently an odd number $L=O(n+\lambda+\log(\ell+1))$ of times. Regard the tuple as one classical encoding, and evaluate by taking the majority separately at each output bit, using independent workspace for each evaluation. For each fixed input and output bit, the joint compilation-and-evaluation error is at most $1/10$ for all sufficiently large $\lambda$. A Chernoff bound makes the error of its majority at most $e^{-c_0L}$ for an absolute constant $c_0>0$. Taking a union bound over output bits, applying Markov's inequality over the compiled tuple, and then taking a union bound over inputs gives
\[
\Pr_{\mathbf c}\!\left[
 \max_x\Pr_E[E(\mathbf c,x)\ne f(x)]>2^{-\lambda}
\right]
\leq 2^{n+\lambda}\ell e^{-c_0L}\leq2^{-\lambda},
\]
where the last inequality holds for a sufficiently large constant in $L$. A polynomial hybrid preserves obfuscation security, and the new algorithms remain polynomial time on all syntactically valid inputs. Thus we may assume the stronger correctness without loss of generality.
\end{remark}

\section{Relations among \textsf{qOWF} definitions and related primitives}
\label{sec:def-three-ev}

We study the relationships among the three \textsf{qOWF} definitions and the related primitives. In summary, we show that $\textsf{qOWF}_1$s and $\textsf{qOWF}_2$s both imply $\textsf{qOWF}_3$, \textsf{qPRG} imply $\textsf{qOWF}_2$, and $\textsf{qOWF}_3$ is equivalent to \textsf{EV-OWPuzz}. We also explain the implications from encryption and point-function obfuscation to puzzles.

\subsection{Relation among \textsf{qOWF} notions.} 

We show the following straightforward result. 

\begin{proposition}
        Every $\textsf{qOWF}_1$ or $\textsf{qOWF}_2$ is also a $\textsf{qOWF}_3$. 
\end{proposition}

\begin{proof}
$\textsf{qOWF}_1$ has the same security requirement as  $\textsf{qOWF}_3$ and has a stricter correctness requirement.

Furthermore, a uniform sampler is one of the samplers permitted for $\textsf{qOWF}_3$, giving the implication from $\textsf{qOWF}_2$.
\end{proof} 

\subsection{\textsf{qPRG} and \texorpdfstring{$\textsf{qOWF}_2$}{\textsf{qOWF}-2}.}

\begin{proposition}
\label{lem:qprg-to-owf}
Suppose a QPT generator $G$ takes a uniform $r(\lambda)$-bit seed and has classical output length $s(\lambda)\geq r(\lambda)+1$. Suppose its output is computationally indistinguishable from $U_{s(\lambda)}$, and its determinism error averaged over uniform seeds is negligible. Then a $\textsf{qOWF}_2$ exists.
\end{proposition}

\begin{proof}
Take $t=\lambda$ independent seed blocks and define $H(x_1,\ldots,x_t)=(G(x_1),\ldots,G(x_t))$ using fresh invocations. The average error is at most $t$ times that of $G$, hence negligible. A polynomial hybrid replacing one block at a time shows that $H(U_{tr})$ is indistinguishable from $U_{ts}$.

If $I$ inverts $H$, distinguish a string $y$ by running $I(y)$, evaluating $H$ on its output, and returning one exactly when the result equals $y$. On $H(U_{tr})$, this is exactly its inversion probability. On a uniform string it is at most
\[
 2^{-ts}\sum_{y,x}\Pr[I(y)=x]\Pr[H(x)=y]
 \leq 2^{-ts}\sum_x\sum_y\Pr[H(x)=y]
 =2^{tr-ts}\leq2^{-\lambda}.
\]
Verification is independent of the inverter conditional on its classical output. A nonnegligible inversion probability would therefore distinguish the two output distributions.
\end{proof}

The negligible average-correctness convention is used in the \textsf{qPRG} definitions of Bouaziz--Ermann et al.~\cite[Definition~3.8]{BEHMV26} and Barhoush~\cite[Section~3.2, Definition~2]{Bar26}. Consequently, Theorem~\ref{thm:separation} rules out these \textsf{qPRG}s as well. The every-input convention in Barhoush, Nishimaki, and Yamakawa~\cite[Section~2.1]{BNY25} is stronger. 

\subsection{\texorpdfstring{$\textsf{qOWF}_3$}{\textsf{qOWF}-3} and \textsf{EV-OWPuzz}.}

$\textsf{qOWF}_3$ captures exactly the existence of efficiently verifiable one-way puzzles. 

\begin{proposition}[$\textsf{qOWF}_3$ and efficiently verifiable puzzles]\label{prop:def-three-ev}
\textsf{EV-OWPuzz}s exist if and only if $\textsf{qOWF}_3$s exist. 
\end{proposition}

\begin{proof}
Assume there exists a $\textsf{qOWF}_3$ $(S,E)$ and let $f(x)$ be its most likely output. Write $\eta=\Pr_{x\leftarrow S}[E(x)\ne f(x)]$. We will show the existence of a \textsf{EV-OWPuzz}.

We let the sampling algorithm of the puzzle sample $x\leftarrow S$ and $y\leftarrow E(x)$, and output $(y,x)$ as the puzzle and answer. On a proposed answer $x'$, the verifier of the puzzle runs $E(x')$ and accepts exactly when the result is $y$. Two independent evaluations at the same honestly sampled input disagree with probability at most $2\eta$, by comparing both with $f(x)$. The puzzle therefore has negligible correctness error. Its solving probability is exactly the inversion probability. 

For the converse, take a puzzle $(\mathsf{Samp},V)$ with correctness error $\eta$. We build a $\textsf{qOWF}_3$ $(S,E)$ as follows. Let $S$ sample $(p,a)\leftarrow\mathsf{Samp}(1^\lambda)$ and output $(p,a)$. The evaluator $E(p,a)$ runs $t=\lambda$  verifications $V(p,a)$ and outputs $p$ if every verification accepts; otherwise it outputs $\bot$. Put $u(p,a)=v_p(a)^t$, where $v_p(a)=\Pr[V(p,a)=1]$. The union bound gives
\[
\mathop{\mathbb{E}}\limits_{(p,a)\leftarrow\mathsf{Samp}(1^\lambda)}[1-u(p,a)]\leq t\eta.
\]
The most-likely-output function satisfies $f(p,a)=p$ when $u(p,a)>1/2$, and $f(p,a)=\bot$ when $u(p,a)<1/2$, with ties broken lexicographically. It is easy to check that $$\mathop{\mathbb{E}}\limits_{x\gets S(1^\lambda)}\Pr[E(x)\ne f(x)]\leq t\eta$$ 
as required for $\textsf{qOWF}_3$.

Given an inverter against this $\textsf{qOWF}_3$, solve a puzzle $p$ by invoking the inverter on $p$, parsing its proposed input as $(p',a')$, and returning $a'$ only when $p'=p$. Let $\sigma_A$ denote the resulting solver's verification probability. In the experiment, a challenge equals $\bot$ with probability at most $t\eta$. When the challenge is a puzzle, the inversion check accepts with probability $u(p,a')\leq v_p(a')$. Hence the inversion probability is at most $t\eta+\sigma_A$.
This bound is negligible by puzzle security, which completes the proof of the converse direction.
\end{proof}

\subsection{Point-function obfuscation implies \textsf{EV-OWPuzz}.} \textsf{qPKE} with classical communication implies \textsf{EV-OWPuzz}: use \textsf{qPKE} as secret-key encryption and apply Chung, Goldin, and Gray~\cite[Section~2.5, Theorem~8]{CGG24}. Point-function obfuscation also suffices for \textsf{EV-OWPuzz}, as the following lemma shows.

\begin{lemma}\label{lem:vbb-to-ev}
Point-function obfuscation implies \textsf{EV-OWPuzz}.
\end{lemma}
\begin{proof}
Assume there exists a point-function obfuscation $(\Obf,\Eval)$. Write $P_a(x)=\mathbf{1}\{x=a\}$. Choose a uniform point $a\in\bits^\lambda$, with $\lambda\geq1$, and sample $p\leftarrow\Obf(1^\lambda,a)$. We build a \textsf{EV-OWPuzz} as follows. The puzzle is $p$, its answer is $a$, and verification runs $\Eval(p,a')$. Correctness follows from the correctness of obfuscation. We now show security. 

Fix a proposed puzzle solver $B$. We use the solver to distinguish whether the first bit of a solution to the puzzle is $b\in \{0,1\}$. For each $b\in\{0,1\}$, consider the one-bit adversary $A_b$ that runs $a'\leftarrow B(p)$ and accepts exactly when $a'$ has its first bit $b$, and verification accepts. Definition~\ref{vbb:def} bounds the verification error even when the input is chosen after seeing $p$, because it takes the maximum over inputs before averaging over programs.

Let $S_b$ be the simulator for $A_b$, which works for every hidden point. For a point chosen uniformly from either first-bit half of $\bits^\lambda$, replacing the simulator's point oracle by the all-zero oracle changes its average acceptance probability by at most $O(q_b2^{-(\lambda-1)/2})$, where $q_b$ bounds its queries. To see this, replace queries one at a time, using the query states with the all-zero oracle: their expected squared weight on the hidden point is at most $2^{-(\lambda-1)}$. This is the standard quantum query comparison of~\cite{BBBV97}. The simulator therefore has negligibly different average acceptance probabilities on the two halves. By the security of point-function obfuscation, the same is true of $A_b$. Its acceptance is negligible on the half with first bit $1-b$, so it is also negligible on the half with first bit $b$. Summing over $b=0,1$ gives exactly the solver's verification success, which is therefore negligible. 
\end{proof}

\section{Separating \textsf{qOWF}s from \textsf{qPKE} and Point-function Obfuscation}
\label{chap:separations}
In this part, we present a quantum black-box separation between \textsf{qOWF}s and both \textsf{qPKE} and point-function obfuscation. More precisely, we will prove the following theorem:
\begin{theorem}[\textsf{qPKE} and point-function obfuscation without \textsf{qOWF}s]
\label{thm:separation}
There exists a self-inverse and real quantum unitary oracle relative to which a perfectly correct $\textsf{qPKE}$ scheme with classical keys and ciphertexts (\Cref{def:pke}) and a point-function obfuscator with classical obfuscated programs exist (\Cref{vbb:def}), but neither $\textsf{qOWF}_1$s (\Cref{def:one}) nor $\textsf{qOWF}_2$s (\Cref{def:two}) exist.
\end{theorem}

We first prove the general barrier theorem in \Cref{sec:general-barrier}, and then in \Cref{sec:construction} construct the oracle supporting \textsf{qPKE} and point-function obfuscation, define a base distribution $\mathcal P$ on it, and provide the constructions of \textsf{qPKE} and point-function obfuscations in that oracle world, together with their correctness bounds. Section~\ref{sec:primitive-security} proves security of the two constructions in the oracle world, with respect to a general class of distributions close enough to $\mathcal P$. Section~\ref{sec:conditioned-oracle} then utilizes the generalized barrier theorem from \Cref{sec:general-barrier} and defines a conditioning procedure on the base oracle $\mathcal P$ to construct a selected distribution $\mathcal Q$, such that no \textsf{qOWF}s exist with respect to $\mathcal Q$ with probability $1$, and the relative entropy of $\mathcal Q$ from the original distribution $\mathcal P$ is bounded, thus by the security guarantees of the \textsf{qPKE} and point-function obfuscation construction in Section~\ref{sec:primitive-security}, we also conclude that the same constructions are secure even in $\mathcal Q$ with probability $1$, which implies the existence of a fixed oracle with this separation, and hence completes the proof. 

\subsection{General barrier theorem}
\label{sec:general-barrier}

Let $(\mathcal X,d)$ be a metric space, let $\mathcal D$ be a measurable space, and let $\nu$ be a distribution on $\mathcal D$. Measurable sets are events to which probabilities are assigned; a map is measurable if the inputs sent into any measurable set also form a measurable set. No distance on $\mathcal D$ is required.

Let $V:\mathcal D\to\mathcal X$ assign a point $V(a)$ to each $a\in\mathcal D$. A fixed measurable map $F:\mathcal D\times\mathcal D\times[0,1]\to\mathcal X$ specifies how to interpolate from $V(a)$ to $V(b)$. It takes two inputs $a,b\in\mathcal D$ and a parameter $t\in[0,1]$, and returns the point $F(a,b,t)\in\mathcal X$, with $F(a,b,0)=V(a)$ and $F(a,b,1)=V(b)$. Once $a,b,t$ are fixed, the resulting point is fixed; randomness comes from sampling these inputs.
For subsets $S,T\subseteq\mathcal X$, write $d(S,T)=\inf_{x\in S,y\in T}d(x,y)$. For a collection of sets, we call the region outside their union the \emph{gap}.

\begin{theorem}[Barrier for distributions obtained by interpolation]
\label{thm:interpolated-barrier}
Let $A,B$ be independent samples from $\nu$, let $U$ be an independently sampled uniform random variable in $[0,1]$, and let $\mu$ be the distribution of $F(A,B,U)$ where $F$ is a map defined above. Suppose there exists $v_0>0$ such that $d(F(a,b,s),F(a,b,t))\leq v_0|s-t|$  for all $a,b\in\mathcal D$, $s,t\in[0,1]$, and suppose $S_1,\ldots,S_r\subseteq\mathcal X$ are measurable sets with $d(S_i,S_j)\geq\delta$ for $i\ne j$, for some $r\geq1$ and $\delta>0$. Then
\[
 \mu(\mathcal X\setminus\bigcup_{i=1}^r S_i)\geq\frac{\delta}{3v_0}\left(g^2-\sum_i p_i^2\right),
\]
where $p_i=\mu(S_i)$ and $g=\sum_{i=1}^r p_i$. 
\end{theorem}

For example, take $\mathcal D=\{0,1\}$, $\mathcal X=[0,1]$ with the usual distance, uniform $\nu$, $V(a)=a$, and $F(a,b,t)=(1-t)a+tb$. Then $\mu$ assigns probability $1/4$ to each endpoint and its remaining probability $1/2$ uniformly over the interval.

In the proof of our main separation result (\Cref{thm:separation}), we will use the following corollary of Theorem~\ref{thm:interpolated-barrier}. 
\begin{corollary}[A set of large probability]
\label{cor:barrier-large-set}
Under the hypotheses of Theorem~\ref{thm:interpolated-barrier}, let $\varepsilon_{\mathrm{bad}}=\mu(\mathcal X\setminus\bigcup_{i=1}^r S_i)<1$ be the probability that a sample lies outside every $S_i$. Then there is an index $j$ such that
\[
 \mu(S_j)\geq 1-\varepsilon_{\mathrm{bad}}-\frac{3v_0}{\delta}\frac{\varepsilon_{\mathrm{bad}}}{1-\varepsilon_{\mathrm{bad}}}.
\]
If $\delta/v_0$ is fixed and the probability outside all the sets is small, then one of the sets has probability close to one.
\end{corollary}

\begin{proof}
Put $p_i=\mu(S_i)$, $g=1-\varepsilon_{\mathrm{bad}}>0$, and $p_{\max}=\max_i p_i$. Since $\sum_i p_i^2\leq gp_{\max}$, Theorem~\ref{thm:interpolated-barrier} gives $\varepsilon_{\mathrm{bad}}\geq(\delta/(3v_0))g(g-p_{\max})$. Rearranging gives $p_{\max}\geq g-(3v_0/\delta)\varepsilon_{\mathrm{bad}}/g$. Substitute $g=1-\varepsilon_{\mathrm{bad}}$ and choose $j$ with $p_j=p_{\max}$.
\end{proof}

To prove the theorem, we construct random paths whose endpoints are independent samples from $\mu$ and whose position at every fixed time also has distribution $\mu$. Proposition~\ref{prop:barrier} derives a barrier bound from these properties and a bound on how quickly the position changes. Lemma~\ref{lem:stationary-paths} constructs the required paths from the interpolation map $F$ in Theorem~\ref{thm:interpolated-barrier}. 

\subsubsection{Paths between independent samples}
\label{sec:barrier}

We now consider an arbitrary probability distribution $\mu$ on a metric space $(\mathcal X,d)$.

Imagine moving continuously from a point $x_{\rm start}$ to a point $x_{\rm end}$. Write $X_t$ for the position at $t\in[0,1]$, so that $X_0=x_{\rm start}$ and $X_1=x_{\rm end}$. We call $t$ the \emph{time parameter} of the path. It parametrizes this mathematical interpolation; it is not the running time of an algorithm. A \emph{random path} is an entire path chosen according to a probability distribution $\Pi$. Thus $\Pi$ is a distribution on paths, while $\mu$ is a distribution on points. Once a path is chosen, all its positions $X_t$ are fixed.

We require the endpoints to be independent samples from $\mu$, and every fixed intermediate point $X_t$ to have distribution $\mu$ as well. We also require movement to be bounded: there is a number $v>0$ such that an interval of parameter length $h$ changes the position by distance at most $vh$. We call $v$ a \emph{speed bound}, or equivalently a \emph{change bound}. The following definition makes these requirements precise without assuming that the path has a derivative.

\begin{definition}[Paths with unchanged distributions]
\label{def:stationary-paths}
The space $(\mathcal X,d)$ with probability distribution $\mu$ admits paths with change bound $v>0$ if there exists a probability distribution $\Pi$ on continuous paths $X:[0,1]\to\mathcal X$ with the following properties. We sample $X\sim\Pi$ and write $X_t=X(t)$.
\begin{enumerate}
\item \textbf{Independent endpoints.} Under $X\sim\Pi$, the initial point $X_0$ and final point $X_1$ are independent, and each has distribution $\mu$.
\item \textbf{Marginal intermediate distributions.}\label{it:marginal} For every fixed $t\in[0,1]$ and every measurable set $A\subseteq\mathcal X$,
\[
 \Pr_{X\sim\Pi}[X_t\in A]=\mu(A).
\]
\item \textbf{Bounded change.}\label{it:bounded-change} With probability one over $X\sim\Pi$,
\[
 d(X_s,X_t)\leq v|s-t|
 \qquad\text{for all }s,t\in[0,1].
\]
\end{enumerate}
We require the evaluation map $(X,t)\mapsto X_t$ to be jointly measurable, so that sampling $X\sim\Pi$ and an independent time gives well-defined probabilities and expectations. The first two requirements are exact conditions on $\Pi$; the third holds with $\Pi$-probability one.
\end{definition}

The second property is also called having the same \emph{marginal distribution} at every time. It concerns each predetermined time separately; it does not assert independence between intermediate times. Neighboring points on a continuous path are strongly related. The third property says, in particular, that traveling a distance of at least $\delta$ requires at least $\delta/v$ units of parameter time.

The first ingredient towards proving \Cref{thm:interpolated-barrier} is the following proposition that converts the existence of suitable random paths into a barrier bound. In particular, the proposition bounds the probability outside finitely many separated sets. Its many-set form will also bound how much probability can lie outside the largest of these sets.

\begin{proposition}[Barrier from paths]
\label{prop:barrier}
Suppose $(\mathcal X,d,\mu)$ admits paths with change bound $v$. Let $S_1,\ldots,S_r$ be finitely many measurable sets separated by distance at least $\delta>0$. Write
\[
 p_i=\mu(S_i),\qquad g=\sum_i p_i,\qquad
 B=\mathcal X\setminus\bigcup_i S_i,\qquad b=\mu(B)=1-g.
\]
Then
\[
 b\ge \frac{\delta}{v}\left(g^2-\sum_i p_i^2\right).
\]
\end{proposition}

\paragraph{Proof template.}
We sample a path from $\Pi$ and an independent uniform time along it. Independent endpoints determine how often the path must cross between different sets. The speed bound makes every such crossing occupy an interval of length at least $\delta/v$, while the unchanged intermediate distributions convert the probability of sampling this interval into a lower bound on the gap probability. Definition~\ref{def:stationary-paths} packages these properties into a reusable criterion. The method has precedents: Barhoush~\cite[Section~4, Lemma~2]{Bar26} interpolates between fixed oracle pairs, and Lott--Villani~\cite[Definitions~1.2 and~1.6, and Theorem~2.5]{LV07} use random paths with independent endpoints and bounds on the distributions encountered along them. Here we require the distribution at each fixed time to remain unchanged and allow general paths with bounded change. The barrier follows directly from these properties, without invoking the geometric inequalities used for Haar distributions in~\cite[Section~5.1]{BEHMV26}.

\begin{proof}
Fix a distribution $\Pi$ on paths satisfying Definition~\ref{def:stationary-paths}. Sample $X\sim\Pi$, and independently choose $T$ uniformly from $[0,1]$. For each fixed path $X$, define
\[
 L_B(X)=\Pr_T[X_T\in B],
\]
$L_B(X)$ is the probability that at a uniformly sampled time on this fixed path, the point lies in the gap $B$. Equivalently, it is the fraction of parameter time spent in $B$. 

\paragraph{Probability of sampling the gap.}
Next we can write the average value of $L_B(X)$ over $X\sim\Pi$ as the overall probability over random $X\sim\Pi$ and independently uniform $T\in[0,1]$ that $X_T\in B$. Since $T$ is chosen independently, the same probability can be written as an expectation over $T$, of the probability that $X_T\in B$ for random $X\sim \Pi$. This last probability term is $b$, because $T$ is sampled independently of $X$, and at every fixed time $t$, the marginal-intermediate-distribution condition (\Cref{it:marginal}) guarantees that $\Pr_{X\sim\Pi}[X_t\in B]=b$. Combining these, we get
\[
 \mathop{\mathbb E}\limits_{X\sim\Pi}[L_B(X)]=\Pr_{X\sim\Pi,T}[X_T\in B]=\mathop{\mathbb E}\limits_T\left[\Pr_{X\sim \Pi}[X_T\in B]\right]=\mathop{\mathbb E}\limits_T\left[\Pr_{Y\sim \mu}[Y\in B]\right]=\mathop{\mathbb E}\limits_T[b]=b.
\]
Hence, in order to complete the proof of \Cref{prop:barrier}, it is enough to show that
\begin{equation}\label{eq:to-be-shown}
    \mathop{\mathbb E}\limits_{X\sim\Pi}[L_B(X)] \geq \frac{\delta}{v}\left(g^2-\sum_i p_i^2\right).
\end{equation}

\paragraph{Time required for a crossing.}
Fix a path $X$ in the support of $\Pi$ satisfying \Cref{it:bounded-change}, such that $X_0\in S_i$ and $X_1\in S_j$ for some $i\ne j$. We identify one interval during which this path crosses the gap. Let $\overline S_k$ denote the closure of $S_k$, obtained by including all limits of sequences of points in $S_k$. The closures remain separated by at least $\delta$, because distance is continuous.

Let $\tau$ be the first time the path reaches any $\overline S_k$ with $k\ne i$, and let $\sigma$ be the last time at or before $\tau$ that it belongs to $\overline S_i$. These times exist, because the path begins in $S_i$ and ends in another set, and continuity makes the relevant sets of hitting times closed subsets of $[0,1]$. A nonempty closed subset of a closed bounded interval contains its smallest and largest elements, so these first and last times are attained. We use closures only to justify the boundary times; we do not assume that adding boundary points preserves the probabilities of the sets.

For every $t\in(\sigma,\tau)$, the path is outside $\overline S_i$ by the definition of $\sigma$, and outside every other closure by the definition of $\tau$. Thus it is in $B$ throughout this interval. Its endpoints lie in different closures, so the separation and change bounds give
\[
 \delta\leq d(X_\sigma,X_\tau)\leq v(\tau-\sigma).
\]
On this fixed path, the uniform time $T$ lies in $(\sigma,\tau)$ with probability $\tau-\sigma$. Consequently $L_B(X)\geq\tau-\sigma\geq\delta/v$. The first inequality is because a path may cross the gap more than once; but we only count the one crossing in the interval $(\sigma,\tau)$, which gives the required lower bound.

\paragraph{Probability of crossing under $\Pi$}
Let $\mathcal E$ be the event, over $X\sim\Pi$, that the initial point $X_0$ and final point $X_1$ belong to different sets among $S_1,\ldots,S_r$. By \Cref{def:stationary-paths}, $X_0,X_1$ are distributed as independent samples from $\mu$. Their independence under $\Pi$ gives
\[
 \Pr_{X\sim\Pi}[\mathcal E]=\Pr_{X_0\sim\mu,X_1\sim\mu}[\exists i\neq j, X_0\in S_i, X_1\in S_j]=\sum_{i\ne j}p_ip_j
 =\left(\sum_i p_i\right)^2-\sum_i p_i^2
 =g^2-\sum_i p_i^2.
\]
We already showed that every path in $\mathcal E$ has $L_B(X)\geq\delta/v$. On every other path, $L_B(X)\geq0$. Therefore, we get
\[
 \mathop{\mathbb E}\limits_{X\sim\Pi}[L_B(X)]
 \geq\frac{\delta}{v}\Pr_{X\sim\Pi}[\mathcal E]
 =\frac{\delta}{v}\left(g^2-\sum_i p_i^2\right),
\]
which completes the proof of \Cref{eq:to-be-shown} and hence the proposition.
\end{proof}

To apply Proposition~\ref{prop:barrier}, we must specify a distribution $\Pi$ on paths satisfying Definition~\ref{def:stationary-paths}. The next lemma constructs this distribution from the hypotheses of Theorem~\ref{thm:interpolated-barrier}. We join several transitions in succession, then choose a random starting position and a random amount of progress along them.

\begin{lemma}[Constructing the path distribution]
\label{lem:stationary-paths}
Under the hypotheses of Theorem~\ref{thm:interpolated-barrier}, there is an explicit distribution $\Pi$ on continuous paths $X:[0,1]\to\mathcal X$ satisfying Definition~\ref{def:stationary-paths} with change bound $v=3v_0$. Thus the distribution $\mu$ in that theorem satisfies the path hypothesis of Proposition~\ref{prop:barrier}. For a finite product of such distributions, independent choices of paths give change bound $3\max_i v_{0,i}$ when distance is the maximum of the coordinate distances. Coordinates fixed at prescribed values may be included.
\end{lemma}

\begin{proof}
Fix $\nu,V,F,\mu$, and $v_0$ as in Theorem~\ref{thm:interpolated-barrier}. We first specify how to sample an entire path from $\Pi$, then verify the three properties in Definition~\ref{def:stationary-paths}.

\paragraph{Sampling a path from $\Pi$.}
Let $A_0,\ldots,A_4$ be five independent samples from $\nu$. Independently sample $U$ uniformly in $[0,1)$ and $T$ uniformly in $[2,3)$. For every $s\in[0,1]$, set
\[
 k_s=\lfloor U+sT\rfloor,\qquad
 u_s=U+sT-k_s,\qquad
 X_s=F(A_{k_s},A_{k_s+1},u_s).
\]
These choices determine the entire function $s\mapsto X_s$. Define $\Pi$ to be the distribution of this function. Since $0\leq U+sT<4$, we have $k_s\in\{0,1,2,3\}$, so the five samples suffice. Half-open sampling intervals give the same uniform distributions as the corresponding closed intervals and avoid endpoint conventions.

To interpret the construction, join the four transitions from $V(A_0)$ to $V(A_1)$, from $V(A_1)$ to $V(A_2)$, and so on, assigning one unit of progress to each transition. At time $s$, the total progress is $U+sT$: the integer $k_s$ selects the current transition, and $u_s\in[0,1)$ gives the position within it. Thus $U$ is the starting position and $T$ is the total amount of progress during the time interval $[0,1]$. At a boundary between transitions, the endpoint identities for $F$ make both formulas agree at the common point $V(A_k)$. Hence every sampled path is continuous. The formula is jointly measurable in the sampled data and $s$, since $F$ is measurable and $k_s$ has only four possible values. This supplies the measurability required in Definition~\ref{def:stationary-paths}.

\paragraph{Bounded change.}
Within a single transition, a progress change of size $h$ changes the point by at most $v_0h$. If an interval crosses a transition boundary, divide it at that boundary and use the triangle inequality. The changes add to at most $v_0$ times the total progress. Since progress increases at rate $T\leq3$, we obtain
\[
 d(X_s,X_t)\leq v_0T|s-t|\leq3v_0|s-t|.
\]

\paragraph{Marginal intermediate distributions.}
Fix a time $s$ and first condition on $T$. The fractional part $u_s=(U+sT)\bmod1$ is uniform in $[0,1)$: adding a fixed number to a uniform position and reducing modulo one preserves uniformity. The index $k_s$ depends only on $U,T$, which are independent of all the samples. Conditional on any choice of $U,T$, the selected pair $(A_{k_s},A_{k_s+1})$ consists of independent $\nu$-samples. Consequently the pair is independent of the position $u_s$, and $X_s$ has exactly the distribution $\mu$.

\paragraph{Independent endpoints.}
The initial point is $X_0=F(A_0,A_1,U)$. Since $T\geq2$, the final transition index satisfies $k_1\geq2$, so the final transition uses neither sampled variable $A_0$ nor $A_1$. Moreover, conditional on any initial position $U$, the fractional part $u_1=(U+T)\bmod1$ is uniform, because $T$ is uniform on an interval of length one. Hence $u_1$ is independent of $U$. Conditional on $U,T$, the final pair of samples has distribution $\nu\times\nu$, independently of the initial pair. It follows that the final triple $(A_{k_1},A_{k_1+1},u_1)$ is independent of $(A_0,A_1,U)$ and has distribution $\nu\times\nu\times U_{[0,1)}$, where $U_{[0,1)}$ denotes the uniform distribution on $[0,1)$. Thus $X_0$ and $X_1$ are independent samples from $\mu$. Both choices in the construction matter: advancing by at least two transitions ensures that the initial and final transitions use disjoint sampled variables, and randomizing the fractional progress avoids sharing the initial position.

We have therefore constructed the required distribution $\Pi$: its endpoints are independent $\mu$-samples, every fixed-time marginal is $\mu$, and its change bound is $3v_0$. For a finite product, perform this construction independently in each coordinate and use the same time interval $[0,1]$. At every fixed time the coordinate distributions remain independent, and the two entire endpoints are independent. Taking the maximum of the coordinate change bounds gives $3\max_i v_{0,i}$. A coordinate fixed at a prescribed value contributes zero change.
\end{proof}

\begin{proof}[Proof of Theorem~\ref{thm:interpolated-barrier}]
Under the theorem's hypotheses, Lemma~\ref{lem:stationary-paths} constructs a distribution $\Pi$ satisfying Definition~\ref{def:stationary-paths} with $v=3v_0$. Apply Proposition~\ref{prop:barrier} to this distribution of paths to obtain the stated bound.
\end{proof}

For finite products of these distributions, Theorem~\ref{thm:interpolated-barrier} and Corollary~\ref{cor:barrier-large-set} hold with $v_0=\max_i v_{0,i}$, using the maximum coordinate distance.

\subsubsection{Relation to the Haar barrier theorem of \cite{BEHMV26}}
\label{sec:barrier-consequences}

Bouaziz--Ermann et al.~\cite[Section~5.1, Theorems~2.1 and~5.3, and Appendix~A.1]{BEHMV26} assume a specific metric probability space: a product of spheres with product Haar measure. They obtain their barrier bound from this geometry by enlarging one set through neighborhoods of increasing radius. They lower-bound the rate at which its probability grows using a Cheeger bound, which controls the size of the boundary relative to the probabilities on its two sides. Their proof obtains that bound from Lichnerowicz's theorem and Buser's inequality. The radius measures the width of a neighborhood; it is not the time parameter of a random path between independently sampled endpoints.

The path condition in Proposition~\ref{prop:barrier} is an intermediate proof tool. Combining it with Lemma~\ref{lem:stationary-paths} gives Theorem~\ref{thm:interpolated-barrier}, whose hypotheses directly describe how the distribution is generated. This gives a formulation parallel to the Haar result: the distribution belongs to a specified class, and the required path properties are proved from that description.

The Haar setting satisfies the path condition as well. In each sphere, connect independent Haar samples by the shortest great-circle arc, the shortest route along the sphere, traversed at a constant rate over $[0,1]$. Except on an event of probability zero, this arc is uniquely defined. Applying a unitary to the endpoints applies the same unitary to every intermediate point, so every fixed intermediate point is Haar distributed. Independent choices in each coordinate preserve the product distribution. Each arc has length at most $\pi$. Since Euclidean distance is at most arc length, these paths satisfy Definition~\ref{def:stationary-paths} with $v=\pi$ for the maximum coordinate Euclidean distance used in~\cite{BEHMV26}.

To obtain the two-set bound, let $S_0,S_1$ be measurable sets separated by distance at least $\delta>0$. Write $p_i=\mu(S_i)$ for $i\in\{0,1\}$, $m=\min\{p_0,p_1\}$, and $\varepsilon=1-p_0-p_1$ for the probability outside both sets. Proposition~\ref{prop:barrier} gives
\[
 \varepsilon\geq\frac{2\delta}{\pi}p_0p_1
 \geq\frac{\delta}{\pi}m(1-\varepsilon),
\]
where the second inequality uses $2\max\{p_0,p_1\}\geq p_0+p_1=1-\varepsilon$. If $m=0$, the desired bound is immediate. Otherwise $\delta\leq2<\pi$, since the Euclidean distance between unit vectors is at most two. Rearranging and using $m\leq1$ yields
\[
 \mu\bigl(\mathcal X\setminus(S_0\cup S_1)\bigr)
 =\varepsilon\geq\frac{\delta m}{\pi+\delta m}
 \geq\frac{\delta}{2\pi}\min\{\mu(S_0),\mu(S_1)\}.
\]
This recovers the dependence on the distance and the smaller set's probability in~\cite[Theorem~5.3]{BEHMV26}, with a different constant that is independent of the dimensions and the number of spheres. For the interpolation distributions considered here, Lemma~\ref{lem:stationary-paths} instead constructs paths with $v=3v_0$, yielding Theorem~\ref{thm:interpolated-barrier}.

\subsection{Base Oracle Distribution and Constructions of \textsf{qPKE} and Point-function Obfuscation}
\label{sec:construction}
We first define a base distribution $\mathcal{P}$ of oracles with respect to which we show a secure \textsf{qPKE} and Point-function Obfuscation, up to negligible probability over the distribution. Then in Section~\ref{sec:conditioned-oracle}, we condition on favorable events to obtain a new distribution $\mathcal Q$, where \textsf{qOWF}s do not exist, while the \textsf{qPKE} and Point-function Obfuscation constructions remain secure, up to some small probability. 
\subsubsection{Base Oracle distribution}\label{sec:base-oracle}
At each length, the oracle provides three operations: $O_{\mathrm{samp}}$, $O_{\mathrm{enc}}$, and $O_{\mathrm{dec}}$. We first describe the classical data used by the oracles, mainly consisting of three objects: a \emph{row}, \emph{tables} and an \emph{interpolation parameter}. We then specify the quantum states and oracle operations determined by these data. After the oracles are described, \Cref{sec:primitive-constructions} gives the \textsf{qPKE} and point-function obfuscation constructions using these operations.

\paragraph{Classical data and their distribution.}
Fix a length $n\geq1$, and set $N=2^n$, $ Y=2^{8n}$, $ M=2^{2n}$, $ Z=2^{12n}.$ There are $N$ p-keys, $Y$ possible s-keys, and $Z$ possible c-texts, where a row corresponding to a p-key $j\in [N]$ is a classical tuple $ (s,L_0,L_1)$, $ s\in[Y],\quad L_0,L_1\subseteq[Z],\quad |L_0|=|L_1|=M,\quad L_0\cap L_1=\varnothing.$ We call $s$ the \emph{s-key}, and $L_0$ and $L_1$ are two sets of \emph{c-texts}, where for each bit $b\in\bits$, $L_b$ is the set of c-texts hiding bit $b$. For each $j\in [N]$, we sample a row $(s,L_0,L_1)$ by independently sampling $s\xleftarrow{\$}[Y]$, $ L_0\xleftarrow{\$}\binom{[Z]}{M}$, $ L_1\xleftarrow{\$}\binom{[Z]\setminus L_0}{M}.$ 
The two c-text sets are sampled as a uniform ordered pair of disjoint sets; they are not independent of each other. 

\paragraph{Sampling the two tables and the interpolation parameter.}
For each $j\in[N]$, independently sample two rows from the distribution above, named $(s_j^A,A_{j,0},A_{j,1})$, and $(s_j^B,B_{j,0},B_{j,1}).$ Collect the first rows into table $A$ and the second rows into table $B$, i.e., $A=\bigl((s_j^A,A_{j,0},A_{j,1})\bigr)_{j\in[N]}$ and $B=\bigl((s_j^B,B_{j,0},B_{j,1})\bigr)_{j\in[N]}$. All row draws are independent across p-keys and between tables. We call their pair at p-key $j$ the \emph{pair of rows}. Within each table, the two c-text sets are disjoint; across tables, c-text sets may intersect and the s-keys $s_j^A,s_j^B$ may coincide.

Let $\mathcal A_j=A_{j,0}\cup A_{j,1}$ and $\mathcal B_j=B_{j,0}\cup B_{j,1}$. A pair of rows is \emph{good} if $\mathcal A_j\cap\mathcal B_j=\varnothing$, and is bad with probability at most $4M^2/Z$. Independently of the tables $A_n,B_n$, sample independent uniform $i_n\in\{1,2,3\}$ and $t_n\in[0,1]$, and set $\tau_n=(i_n-1+t_n)/3$. Then $\tau_n$ is uniform in $[0,1]$. The \emph{base distribution} $\mathcal P$ is the resulting distribution of $(A_n,B_n,\tau_n)_{n\geq1}$, independent across lengths, with $A_n,B_n,\tau_n$ independent within each length. All three operations and all rows at length $n$ use the same $\tau_n$; algorithms query the resulting operations without supplying this parameter. Below we suppress the length subscript and write $A,B,\tau$.
\paragraph{States and response angles.}
Each table specifies uniform c-text states and response angles:
\[
\begin{gathered}
 \ket{\phi^A_{j,b}}=\frac1{\sqrt M}\sum_{c\in A_{j,b}}\ket c,
 \qquad
 \ket{\phi^B_{j,b}}=\frac1{\sqrt M}\sum_{c\in B_{j,b}}\ket c,\\[3pt]
 \alpha^A_j(c)=
 \begin{cases}
 0&c\in A_{j,0},\\
 \pi/2&c\in A_{j,1},\\
 \pi/4&c\notin\mathcal A_j,
 \end{cases}
 \qquad
 \alpha^B_j(c)=
 \begin{cases}
 0&c\in B_{j,0},\\
 \pi/2&c\in B_{j,1},\\
 \pi/4&c\notin\mathcal B_j.
 \end{cases}
\end{gathered}
\]
$O_{\mathrm{dec}}$ uses angles $0,\pi/2,\pi/4$ to return zero, one, or an unbiased bit when its response qubit starts in $\ket0$. To accommodate c-texts from both tables, define the shared angle
\[
 \gamma_j(c)=
 \begin{cases}
 \alpha^A_j(c)&c\in\mathcal A_j\setminus\mathcal B_j,\\
 \alpha^B_j(c)&c\in\mathcal B_j\setminus\mathcal A_j,\\
 \pi/4&\text{otherwise}.
 \end{cases}
\]
Thus a c-text $c$ in exactly one table uses that table's angle, while a c-text in both tables or neither uses the default $\pi/4$. For $1/3\leq\tau\leq2/3$, define
\[
\begin{aligned}
 \ket{g_j(\tau)}&=
 \begin{cases}
 \cos\!\left(\frac{\pi(3\tau-1)}2\right)\ket{s_j^A}+\sin\!\left(\frac{\pi(3\tau-1)}2\right)\ket{s_j^B}&s_j^A\ne s_j^B,\\
 \ket{s_j^A}&s_j^A=s_j^B,
 \end{cases}\\[3pt]
 \ket{\phi_{j,b}(\tau)}&=
 \frac{\cos\!\left(\frac{\pi(3\tau-1)}2\right)\ket{\phi^A_{j,b}}
       +\sin\!\left(\frac{\pi(3\tau-1)}2\right)\ket{\phi^B_{j,b}}}
 {\sqrt{1+2\sin\!\left(\frac{\pi(3\tau-1)}2\right)\cos\!\left(\frac{\pi(3\tau-1)}2\right)
                    \langle\phi^A_{j,b}\mid\phi^B_{j,b}\rangle}}.
\end{aligned}
\]

The $O_{\mathrm{samp}}$ state is $\ket\Psi=N^{-1/2}\sum_{j\in[N]}\ket j\ket{\psi_j}$. The parameter $\tau$ determines its components $\ket{\psi_j}$, the $O_{\mathrm{enc}}$ states $\ket{\phi_{j,b}}$, and the $O_{\mathrm{dec}}$ angles $\beta_{j,y}(c)$. For $y\notin\{s_j^A,s_j^B\}$, always set $\beta_{j,y}(c)=\pi/4$. For distinct s-keys $s_j^A,s_j^B$, use the following formulas.
\begin{enumerate}
\item \textbf{$0\leq\tau\leq1/3$: sample from $A$ and interpolate to the shared response rule.} Set $\ket{\psi_j}=\ket{s_j^A}, \ket{\phi_{j,b}}=\ket{\phi^A_{j,b}}, \beta_{j,s_j^A}(c)=(1-3\tau)\alpha^A_j(c)+3\tau\gamma_j(c), \beta_{j,s_j^B}(c)=(1-3\tau)\pi/4+3\tau\gamma_j(c).$

\item \textbf{$1/3\leq\tau\leq2/3$: interpolate both states and use the shared response rule.} Set $\ket{\psi_j}=\ket{g_j(\tau)}$, $\ket{\phi_{j,b}}=\ket{\phi_{j,b}(\tau)}$, and $\beta_{j,s_j^A}(c)=\beta_{j,s_j^B}(c)=\gamma_j(c)$.
\item \textbf{$2/3\leq\tau\leq1$: sample from $B$ and interpolate to its response rule.} Set $\ket{\psi_j}=\ket{s_j^B},\ket{\phi_{j,b}}=\ket{\phi^B_{j,b}},\beta_{j,s_j^A}(c)=(3-3\tau)\gamma_j(c)+(3\tau-2)\pi/4,\beta_{j,s_j^B}(c)=(3-3\tau)\gamma_j(c)+(3\tau-2)\alpha^B_j(c).$
\end{enumerate}
If $s_j^A=s_j^B$, keep the same states and assign their common s-key the single angle $(1-3\tau)\alpha^A_j(c)+3\tau\gamma_j(c)$ for $0\leq\tau\leq1/3$, $\gamma_j(c)$ for $1/3\leq\tau\leq2/3$, and $(3-3\tau)\gamma_j(c)+(3\tau-2)\alpha^B_j(c)$ for $2/3\leq\tau\leq1$. The formulas agree at $\tau=1/3$ and $\tau=2/3$. 

\paragraph{Oracle operations.}
Each answer register for $O_{\mathrm{samp}}$ or $O_{\mathrm{enc}}$ also has a distinguished basis state $\ket\blank$ orthogonal to all valid answers. The following unitaries are defined on arbitrary inputs, act as the identity on unused encodings, and extend linearly to superpositions. $O_{\mathrm{enc}}$ and $O_{\mathrm{dec}}$ preserve their control registers.
\begin{enumerate}
\item \textbf{$O_{\mathrm{samp}}$.} On the joint answer register for $(j,y)$, define
\[
 O_{\mathrm{samp},n}=I-(\ket\blank-\ket\Psi)(\bra\blank-\bra\Psi).
\]
It swaps $\ket\blank$ with $\ket\Psi$ and fixes their orthogonal complement. Measuring the prepared state $\ket\Psi$ gives a uniform p-key $j$ and an s-key distributed as the measurement of $\ket{\psi_j}$.
\item \textbf{$O_{\mathrm{enc}}$.} On control registers $(j,b)$ and an answer register for the c-text $c$, define
\[
\begin{aligned}
 O_{\mathrm{enc},n}
 ={}&\sum_{j\in[N],\,b\in\bits}\ket{j,b}\bra{j,b}\otimes
 \Bigl[I-(\ket\blank-\ket{\phi_{j,b}})
          (\bra\blank-\bra{\phi_{j,b}})\Bigr]\\
 &+\Pi_{\rm bad}\otimes I,
\end{aligned}
\]
where $\Pi_{\rm bad}$ projects onto malformed controls. At each valid $(j,b)$, the answer operation swaps $\ket\blank$ with $\ket{\phi_{j,b}}$ and fixes their orthogonal complement. Measuring the prepared answer gives a c-text $c$ associated with bit $b$ at p-key $j$.
\item \textbf{$O_{\mathrm{dec}}$.} On controls $(j,y,c)\in[N]\times[Y]\times[Z]$ and one response qubit, define
\[
\begin{aligned}
 O_{\mathrm{dec},n}\ket{j,y,c,0}
 &=\cos\bigl(\beta_{j,y}(c)\bigr)\ket{j,y,c,0}
   +\sin\bigl(\beta_{j,y}(c)\bigr)\ket{j,y,c,1},\\
 O_{\mathrm{dec},n}\ket{j,y,c,1}
 &=\sin\bigl(\beta_{j,y}(c)\bigr)\ket{j,y,c,0}
   -\cos\bigl(\beta_{j,y}(c)\bigr)\ket{j,y,c,1}.
\end{aligned}
\]
Malformed controls are unchanged. Starting with response $\ket0$, measurement returns one with probability $\sin^2(\beta_{j,y}(c))$. Every other s-key uses the default angle $\pi/4$ for all c-texts $c$.
\end{enumerate}
These three operations use $O(n)$ qubits and form the oracle block at length $n$. In addition, fix a $\PSPACE$-complete language $L_{\PSPACE}$ independently of all sampled data and provide its reversible membership oracle $U_{\PSPACE}\ket{z,v}=\ket{z,v\oplus\mathbf 1_{\{z\in L_{\PSPACE}\}}}$, for $v\in\bits$. This oracle has no access to the other oracles.

For the following norm estimates, abbreviate the reflections used above by
\[\numberthis
 H_\psi=I-(\ket\blank-\ket\psi)(\bra\blank-\bra\psi),
 \qquad
 Q(\theta)=\begin{pmatrix}\cos\theta&\sin\theta\\
                         \sin\theta&-\cos\theta\end{pmatrix},
 \label{eq:reflections}
\]
where $\ket\psi$ is a normalized real state orthogonal to $\ket\blank$. Thus $O_{\mathrm{samp}}$ is $H_\Psi$, the answer operation of $O_{\mathrm{enc}}$ at $(j,b)$ is $H_{\phi_{j,b}}$, and the response operation of $O_{\mathrm{dec}}$ at $(j,y,c)$ is $Q(\beta_{j,y}(c))$.

\subsubsection{Properties of the Oracle and the base distribution on it}
Next, we state the following immediate properties of the oracle distribution.
\begin{observation}[Immediate properties of the oracle]
\label{obs:oracle-properties}
The definitions give the following properties.
\begin{itemize}
\item \textbf{Matching endpoints.} For any three tables $A,B,C$ at the same length, the transition $A\to B$ at $\tau=1$ gives exactly the same oracle operations as the transition $B\to C$ at $\tau=0$, because both endpoints depend only on $B$. Thus the transitions $A\to B\to C$ join continuously, allowing us to concatenate transitions in the path construction of Section~\ref{sec:general-barrier}.
\item \textbf{Symmetry.} Every operation is a real symmetric involution, since its defining reflections have these properties.
\item \textbf{Uniform p-key.} For every $\tau$, $O_{\mathrm{samp}}$ yields a uniform p-key $j$, since each normalized $\ket{\psi_j}$ has coefficient $1/\sqrt N$ in $\ket\Psi$.
\item \textbf{Default response.} By definition, $\beta_{j,y}(c)=\pi/4$ whenever $y\notin\{s_j^A,s_j^B\}$.
\item \textbf{Correctness.} At every good pair of rows at $j$, the formulas give $\beta_{j,y}(c)=\pi b/2$ for honestly sampled $y,c$, so $O_{\mathrm{dec}}$ recovers $b$ with certainty, including when $s_j^A=s_j^B$.
\end{itemize}
\end{observation}
Next, we establish that the unitaries $U(\tau)$ are in fact Lipschitz in $\tau$ with the Lipschitz constant independent of the length $n$ and the choice of tables $A,B$.
\begin{lemma}[Bounded change of the interpolation]
\label{lem:edge-properties}
There is a constant $v_0$, independent of the length $n$ and the choice of tables $A,B$, such that $\|U(\tau)-U(\tau')\|\leq v_0|\tau-\tau'|$ for every oracle operation $U$ along the transition from $A$ to $B$ and all $\tau,\tau'\in[0,1]$.

\end{lemma}

\begin{proof}
For the response and state-preparation reflections,
\[
\|Q(\theta)-Q(\theta')\|\le |\theta-\theta'|,
\qquad
\|H_\psi-H_\phi\|\le 2\sqrt 2\,\|\ket\psi-\ket\phi\|.
\]
Within each of the intervals $[0,1/3]$, $[1/3,2/3]$, and $[2/3,1]$, the response angles $\beta_{j,y}(c)$ and s-key states $\ket{\psi_j}$ have derivatives in $\tau$ bounded by absolute constants. Orthogonality of the terms with distinct p-keys gives the same bound for $\ket\Psi$. On $[1/3,2/3]$, the numerator defining $\ket{\phi_{j,b}(\tau)}$ has a bounded derivative and a norm of at least one, so its normalized state also has a bounded derivative; outside this interval, the c-text states are constant. Direct sums take the maximum of component norms, and the formulas agree at the joins, so the triangle inequality extends the bound across the entire transition. 
\end{proof}

Next, we state the following lemma regarding the replacement of a random coordinate of an oracle distribution with independently sampled variables for each p-key, such as our base distribution $\mathcal P$.
\begin{lemma}[Replacing a random coordinate]
\label{lem:hidden-subset-entropy}
Let $\mathcal P=\mathcal P_W\prod_{j=1}^N\mathcal P_j$ and
$D(\mathcal R\Vert\mathcal P)\le L$.
Sample $(W,X_1,\ldots,X_N)$ from $\mathcal R$, independently choose
a uniform $h$-subset $H\subseteq[N]$, and choose a uniform $J\in H$,
where $1\le h\le N$.

Let $\mu$ be the distribution of $(W,X_{-H},H,J,X_J)$, where
$X_{-H}=(X_j)_{j\notin H}$.
Let $\mu'$ be obtained by replacing $X_J$ with an independent
sample from $\mathcal P_J$. Then
\begin{equation}
 D(\mu\Vert\mu')\le\frac Lh,
 \label{eq:hidden-subset-entropy}
\end{equation}
and consequently,
\begin{equation}
 \SD(\mu,\mu')\le\sqrt{\frac{L}{2h}}.
 \label{eq:hidden-subset-transfer}
\end{equation}
\end{lemma}
\begin{proof}
Fix $H$. Under $\mathcal P$, the coordinates $X_H$ are independent of $(W,X_{-H})$ and have distribution $\prod_{j\in H}\mathcal P_j$. The chain rule in Lemma~\ref{lem:kl-chain-rule}, together with nonnegativity of relative entropy, therefore gives
\[
\mathop{\mathbb E}\limits_{W,X_{-H}}
D\!\left(
\mathcal R_{X_H\mid W,X_{-H}}
\middle\Vert
\prod_{j\in H}\mathcal P_j
\right)
\leq D(\mathcal R\Vert\mathcal P)
\leq L.
\]Apply Lemma~\ref{lem:kl-product} to each conditional distribution of $X_H$, and average over the conditioning variables. This yields
\[
\sum_{j\in H}
\mathop{\mathbb E}\limits_{W,X_{-H}}
D\!\left(
\mathcal R_{X_j\mid W,X_{-H}}
\middle\Vert
\mathcal P_j
\right)
\leq L.
\]Since $J$ is uniform in $H$, averaging over $J$ divides this sum by $h$. Averaging also over $H$ proves~\eqref{eq:hidden-subset-entropy}.
By construction, $\mu$ and $\mu'$ have the same marginal on $(W,X_{-H},H,J)$. Applying the chain rule once more gives
\[
D(\mu\Vert\mu')
=
\mathop{\mathbb E}\limits_{H,J,W,X_{-H}}
D\!\left(
\mathcal R_{X_J\mid W,X_{-H}}
\middle\Vert
\mathcal P_J
\right)
\leq \frac{L}{h}.
\]
Finally, combining this with Pinsker's inequality (Lemma~\ref{lem:pinsker}) yields the desired Eq~\eqref{eq:hidden-subset-transfer}. Adjoining independent samples preserves this distance, and applying the same classical or quantum experiment cannot increase the difference between its output probabilities.
\end{proof}


\subsubsection{Construction of \textsf{qPKE} and Point-function Obfuscation}
\label{sec:primitive-constructions}

We now define the two primitives supported by the oracle. Keys, ciphertexts, and obfuscated programs are classical strings, and the quantum algorithms use the quantum oracles $O_{\mathrm{samp}}, O_{\mathrm{enc}}$ and $O_{\mathrm{dec}}$ defined above.

\paragraph{\textbf{The Public Key Encryption Scheme.}} At security parameter $\lambda$, use the operations at length $n=\lambda$:
\begin{itemize}
\item $\KeyGen(1^\lambda)$ applies $O_{\mathrm{samp},\lambda}$ to blank and measures $(j,y)$. It outputs $\mathsf{pk}=j$ and $\mathsf{sk}=(j,y)$.
\item $\Enc(j,b)$ applies $O_{\mathrm{enc},\lambda}$ to $\ket{j,b,\blank}$, measures the c-text $c$, and outputs the ciphertext $(j,c)$.
\item $\Dec((j,y),(j,c))$ applies $O_{\mathrm{dec},\lambda}$ to $\ket{j,y,c,0}$ and measures the response bit.
\end{itemize}
Define encryption on malformed inputs to return a fixed ciphertext, and decryption to return zero if an input is malformed or the p-keys in the secret key and ciphertext differ.

\begin{lemma}[\textsf{qPKE} correctness under a relative-entropy bound]
\label{lem:pke-correctness}
Let $\mathcal P_n$ be the base distribution of oracle descriptions at length $n=\lambda$. For any potentially different distribution $\mathcal R$ on these descriptions satisfying $D(\mathcal R\Vert\mathcal P_n)\leq L$, it holds that for every $b\in\{0,1\}$, the expected honest decryption error of the \textsf{qPKE} construction, averaged over $\mathcal R$ satisfies
\[
\mathop{\mathbb E}\limits_{\mathcal O\leftarrow\mathcal R}
\left[
\Pr_{\substack{
(\mathsf{pk},\mathsf{sk})\leftarrow\mathsf{KeyGen}^{\mathcal O}(1^\lambda)\\
\mathsf{ct}\leftarrow\mathsf{Enc}^{\mathcal O}(\mathsf{pk},b)
}}
\left[
\mathsf{Dec}^{\mathcal O}(\mathsf{sk},\mathsf{ct})\ne b
\right]
\right]
\leq \frac{4M^2}{Z}+\sqrt{\frac{L}{2N}}.
\]

\end{lemma}
\begin{proof}
Observation~\ref{obs:oracle-properties} gives exact honest decryption whenever the sampled pair of rows is good. Since key generation chooses the p-key uniformly for every oracle description, the decryption error for either message bit is at most the fraction of bad pairs of rows. Apply Lemma~\ref{lem:hidden-subset-entropy} with $h=N$, taking the coordinates to be the $N$ pairs of rows and $W=\tau$ to be the interpolation parameter. Replacing a uniformly chosen pair by an independent base sample changes its probability of being bad by at most $\sqrt{L/(2N)}$. Under the base distribution, a union bound over the four cross-table pairs of c-text sets gives probability at most $4M^2/Z$ that $\mathcal A_J\cap\mathcal B_J\ne\varnothing$. Thus the expected fraction of bad pairs under $\mathcal R$, and hence the expected decryption error, is at most the stated bound.
\end{proof}

\paragraph{\textbf{Point-function obfuscation.}}
At security parameter $\lambda$, use the oracle block at length $n=2\lambda$. Write each p-key as $j=(j_1,j_2)$, where $j_1,j_2\in\bits^\lambda$ are its prefix and suffix, respectively. The algorithms are as follows.
\begin{itemize}
\item $\Obf(1^\lambda,a)$: run $\KeyGen(1^{2\lambda})$ to obtain p-key $j=(j_1,j_2)$ and secret key $(j,s)$, then run $\Enc(j,1)$ to obtain ciphertext $(j,c)$. Output $p=(j_1,d,s,c)$, where $d=j_2\oplus a$.
\item $\Eval(p,x)$: for $p=(j_1,d,s,c)$, set $j_x=(j_1,x\oplus d)$. Apply $O_{\mathrm{dec},2\lambda}$ at controls $(j_x,s,c)$ to each of $\lambda$ response qubits initialized to zero. Output one exactly when all measured responses are one. Reject malformed programs or inputs.
\end{itemize}
For $x=a$, evaluation uses the original p-key $j$. The suffix $j_2$ appears in $p$ only through $d$.

\paragraph{Oracle data for p-keys with a common prefix.}
There are $G=2^\lambda$ possible values of $j_1$ and, for each $j_1$, there are $C=2^\lambda$ possible values of $j_2$, giving $N=GC=2^{2\lambda}$ p-keys. Identifying $\bits^\lambda$ with $[2^\lambda]$, we write $j_1\in[G]$ and $j_2\in[C]$. At this length, an s-key lies in $[Y]$, where $Y=2^{16\lambda}$, and each c-text set contains $M=2^{4\lambda}$ elements of $[Z]$, where $Z=2^{24\lambda}$. For each $j_1\in[G]$, collect the $C$ pairs of rows corresponding to p-keys with that prefix into the tuple
\[
\begin{aligned}
 X_{j_1}=\Bigl(&
 (s^A_{(j_1,j_2)},A_{(j_1,j_2),0},A_{(j_1,j_2),1}),\\
 & (s^B_{(j_1,j_2)},B_{(j_1,j_2),0},B_{(j_1,j_2),1})
 \Bigr)_{j_2\in[C]}.
\end{aligned}
\]
For a fixed program $p=(j_1,d,s,c)$, varying the input $x$ changes only the suffix of $j_x$, so every evaluation uses decryption data from $X_{j_1}$. Under $\mathcal P$, the tuples $X_1,\ldots,X_G$ are independent and identically distributed, and independent of $\tau$ and the other oracle blocks.

Call $X_{j_1}$ good if its $2C$ s-key values are all distinct and $\mathcal A_{(j_1,j_2)}\cap\mathcal B_{(j_1,j_2)}=\varnothing$ for every $j_2\in[C]$; otherwise call it bad. Define the fraction of bad tuples by
\[
 \rho_\lambda=\frac1G\sum_{j_1\in[G]}
 \mathbf 1\{X_{j_1}\text{ is bad}\}.
\]

\begin{lemma}[Point-function obfuscation correctness under a relative-entropy bound]
\label{lem:vbb-correctness}
Let $\mathcal P_n$ be the base distribution of oracle descriptions at length $n=2\lambda$, and let $\mathcal R$ be a distribution on these descriptions satisfying $D(\mathcal R\Vert\mathcal P_n)\leq L$. Define
\[
 \epsilon^{\textsf{VBB}}(\mathcal O)
 =\max_{a\in\bits^\lambda}
   \mathop{\mathbb E}\limits_{p\leftarrow\Obf^{\mathcal O}(1^\lambda,a)}
   \max_{x\in\bits^\lambda}
   \Pr[\Eval^{\mathcal O}(p,x)\ne P_a(x)].
\]
Then
\[
 \mathop{\mathbb E}\limits_{\mathcal O\leftarrow\mathcal R}
   \epsilon^{\textsf{VBB}}(\mathcal O)
 \leq O\left(\frac{C^2}{Y}+\frac{CM^2}{Z}
             +\sqrt{\frac{L}{G}}\right)+2^{-\lambda}.
\]
For every fixed oracle, $\epsilon^{\textsf{VBB}}(\mathcal O)\leq\rho_\lambda+2^{-\lambda}$.
\end{lemma}

\begin{proof}
A program $p=(j_1,d,s,c)$ with good $X_{j_1}$ accepts its point $a$ with certainty: at the original p-key $j=(j_1,j_2)$, Observation~\ref{obs:oracle-properties} makes $O_{\mathrm{dec}}$ return one on the published s-key $s$ and c-text $c$, so all repetitions accept. At an input $x\ne a$, the evaluator uses $j_x=(j_1,x\oplus d)\ne j$. The published s-key differs from both s-keys at $j_x$, so $\beta_{j_x,s}(c)=\pi/4$ and the acceptance probability is $2^{-\lambda}$. This holds for every input to the program. Since $O_{\mathrm{samp}}$ chooses $j_1$ uniformly for every oracle description, the maximum over points of the average program error is at most $\rho_\lambda+2^{-\lambda}$.

Apply Lemma~\ref{lem:hidden-subset-entropy} to $X_1,\ldots,X_G$, with $h=G$ and $W=\tau$. For a uniform $j_1\in[G]$, replacing $X_{j_1}$ by an independent sample from $\mathcal P_{X_{j_1}}$ changes its probability of being bad by at most $\sqrt{L/(2G)}$. Under $\mathcal P$, a union bound over s-key collisions and bad pairs of rows gives
\[
 \mathop{\mathbb E}\limits_{\mathcal R}\rho_\lambda
 \leq \frac{\binom{2C}{2}}{Y}+\frac{4CM^2}{Z}
       +\sqrt{\frac{L}{2G}}.
\]
Combining the two bounds completes the proof of the lemma.
\end{proof}

The correctness requirement in Definition~\ref{vbb:def} also suffices for coherent adaptive evaluation of the same honestly generated program. For a fixed program with maximum error $\epsilon_p$, keep the input register containing $x$, compute and copy the answer, and uncompute. {The resulting operation approximates the point oracle within norm $2\sqrt{\epsilon_p}$.} Since the input register containing $x$ is retained, the same bound holds on superpositions, including inputs entangled with other registers. Jensen's inequality and a polynomial query hybrid give negligible average error for polynomially many adaptive evaluations of one program.

Correctness averages over honestly generated programs; it need not hold for every program string.

\Cref{sec:pke-security,sec:point-obf} prove security for these constructions under distributions with bounded relative entropy from the base law.

\subsection{Security of \textsf{qPKE} and Point-function Obfuscation}
\label{sec:primitive-security}

We prove security for the constructions of Section~\ref{sec:primitive-constructions} under a general class of distributions including the distribution $\mathcal P$. For each primitive, we first prove security with respect to $\mathcal{P}$, and then transfer the bound to any distribution $\mathcal R$ with bounded relative entropy from $\mathcal{P}$. 

\subsubsection{\textsf{qPKE}}
\label{sec:pke-security}

At length $n$, the two table rows at p-key $j$ from Section~\ref{sec:construction} form the pair of rows
\[
 X_j=\bigl((s_j^A,A_{j,0},A_{j,1}),\,(s_j^B,B_{j,0},B_{j,1})\bigr).
\]
Let $W$ contain the interpolation parameter $\tau$ and the accessible oracle data at other lengths. The reference distribution is $\mathcal P_W\prod_j\mathcal P_{X_j}$; operations at other lengths may be arbitrary functions of $W$. 

For an adversary $\mathcal A$, define
\[
 \begin{aligned}
 p_{\mathcal A,b}^{\mathcal O}(j)
 &=\Pr\!\left[\mathcal A^{\mathcal O}\bigl(j,\Enc^{\mathcal O}(j,b)\bigr)=1\right],\\
 a_{\mathcal A}(\mathcal O)
 &=\frac1N\sum_{j=1}^N
   \left(p_{\mathcal A,0}^{\mathcal O}(j)-p_{\mathcal A,1}^{\mathcal O}(j)\right).
 \end{aligned}
\]
Here $\mathcal A^{\mathcal O}(j,\mathsf{ct})$ denotes the full adversarial execution, including preprocessing, with the public key $j$ and ciphertext $\mathsf{ct}$ delivered at their prescribed times. Thus $a_{\mathcal A}(\mathcal O)$ is the difference between the probabilities that the adversary outputs one in the two challenge experiments.

\begin{lemma}[\textsf{qPKE} with independently sampled  oracle data for the challenge p-key]
\label{lem:fresh-record-pke}
Fix the interpolation parameter $\tau$ as defined in \Cref{sec:base-oracle}, the challenge p-key $J$, and arbitrary external data, including all oracle data outside $X_J$. If $X_J$ is sampled independently according to the base distribution $\mathcal P$, then every adversary making at most $q$ queries, including preprocessing, satisfies
\[
 \left|\mathop{\mathbb E}\limits_{X_J\leftarrow\mathcal P_{X_J}}
   \left[p_{\mathcal A,0}^{\mathcal O}(J)-p_{\mathcal A,1}^{\mathcal O}(J)\right]\right|
 \leq O\!\left(\frac q{\sqrt N}+\frac q{\sqrt Y}
          +\frac q{\sqrt M}+\frac{M^2}{Z}\right).
\]
\end{lemma}

\begin{proof}
For each $b\in\bits$, define hybrids $\mathsf H_b^{(0)},\ldots,\mathsf H_b^{(4)}$, each outputting the adversary's bit, and write $\pi_b^{(k)}=\Pr[\mathsf H_b^{(k)}=1]$. These probabilities average over $X_J$ and the experiment's randomness; $J$ and the external data remain fixed. All oracle changes apply throughout the execution, including preprocessing.

\textbf{Hybrid $\mathsf H_b^{(0)}$: the real experiment.} Sample $X_J\leftarrow\mathcal P_{X_J}$ and give the adversary public key $J$ and a ciphertext encrypting $b$. Thus $\pi_b^{(0)}=\mathop{\mathbb E}\limits_{X_J}p_{\mathcal A,b}^{\mathcal O}(J)$.

\textbf{Hybrid $\mathsf H_b^{(1)}$: remove $J$ from $O_{\mathrm{samp}}$.} Modify $\mathsf H_b^{(0)}$ so that every adversarial $O_{\mathrm{samp}}$ query uses the state obtained by removing $N^{-1/2}\ket J\ket{\psi_J}$ from $\ket\Psi$ and normalizing the remainder. The removed squared norm is $1/N$, so the state and its reflection change by $O(N^{-1/2})$. Comparing the queries one at a time gives $|\pi_b^{(1)}-\pi_b^{(0)}|\leq O(q/\sqrt N)$.

\textbf{Hybrid $\mathsf H_b^{(2)}$: use default $O_{\mathrm{dec}}$ responses at $J$.} Modify $\mathsf H_b^{(1)}$ by setting every angle $\beta_{J,y}(c)$ to $\pi/4$. Only controls with $y\in\{s_J^A,s_J^B\}$ change. The query states in $\mathsf H_b^{(2)}$ are independent of these two uniform s-keys, so their expected squared weight on the changed controls is at most $2/Y$. Comparing queries using these states and applying Cauchy--Schwarz gives $|\pi_b^{(2)}-\pi_b^{(1)}|\leq O(q/\sqrt Y)$. $O_{\mathrm{samp}}$ and $O_{\mathrm{dec}}$ now use no data from $X_J$.

\textbf{Hybrid $\mathsf H_b^{(3)}$: condition on disjoint c-text sets.}
Run $\mathsf H_b^{(2)}$ conditioned on the four sets $A_{J,0},A_{J,1},B_{J,0},B_{J,1}$ being pairwise disjoint. This event fails with probability at most $4M^2/Z$, giving $|\pi_b^{(3)}-\pi_b^{(2)}|\leq4M^2/Z$. The sets are now uniform ordered disjoint $M$-subsets of $[Z]$. For $1/3\leq\tau\leq2/3$, measuring $\ket{\phi_{J,b}(\tau)}$ chooses table $A$ or $B$ with probabilities $\cos^2(\pi(3\tau-1)/2)$ and $\sin^2(\pi(3\tau-1)/2)$, followed by a uniform c-text from $A_{J,b}$ or $B_{J,b}$. For $\tau\leq1/3$, only $A$ is used; for $\tau\geq2/3$, only $B$ is used. The table choice is independent of $b$.

An equivalent way to sample $\mathsf H_b^{(3)}$ is to choose a uniform c-text $C\in[Z]$ and four uniform disjoint helper sets of size $M$ in $[Z]\setminus\{C\}$. In the set selected by the table choice and message $b$, replace a uniform c-text $D$ by $C$, and use $(J,C)$ as the challenge ciphertext. For each resulting tuple and selected $C$, the removed c-text $D$ can be any of the $Z-4M$ c-texts outside the four resulting sets, with equal probability. Hence this sampling description leaves the distribution unchanged.

\textbf{Hybrid $\mathsf H_b^{(4)}$: keep the helper sets unchanged.}
Use the preceding sampling description, but omit the replacement of $D$ by $C$; the challenge ciphertext is still $(J,C)$. The selected c-text state $\ket{\phi^A_{J,b}}$ or $\ket{\phi^B_{J,b}}$ changes by $\sqrt{2/M}$. For $1/3\leq\tau\leq2/3$, the normalization denominator is one because the sets remain disjoint, so the same bound holds for $\ket{\phi_{J,b}(\tau)}$. The $O_{\mathrm{enc}}$ reflections change by $O(M^{-1/2})$, giving $|\pi_b^{(4)}-\pi_b^{(3)}|\leq O(q/\sqrt M)$.

The distribution of $\mathsf H_b^{(4)}$ is independent of $b$: the challenge is $(J,C)$, the helper sets were sampled independently of $b$, and $O_{\mathrm{samp}}$ and $O_{\mathrm{dec}}$ do not use $X_J$. Thus $\pi_0^{(4)}=\pi_1^{(4)}$, and the triangle inequality gives
\[
 \begin{aligned}
 |\pi_0^{(0)}-\pi_1^{(0)}|
 &\leq\sum_{b\in\bits}\sum_{k=1}^4
       |\pi_b^{(k)}-\pi_b^{(k-1)}|\\
 &\leq O\!\left(\frac q{\sqrt N}+\frac q{\sqrt Y}
                  +\frac q{\sqrt M}+\frac{M^2}{Z}\right).
 \end{aligned}
\]
This is the claimed bound.
\end{proof}

\begin{lemma}[\textsf{qPKE} under a relative-entropy bound]
\label{lem:conditioned-pke}
Suppose $\mathcal R$ has relative entropy at most $L$ from $\mathcal P_W\prod_j\mathcal P_{X_j}$, and let $q\geq1$ bound the adversary's queries and honest preparation calls. For every integer $1\leq h\leq N$,
\[\numberthis
 \left|\mathop{\mathbb E}\limits_{\mathcal O\leftarrow\mathcal R}a_{\mathcal A}(\mathcal O)\right|
 \leq O\left(q\sqrt{\frac hN}+\sqrt{\frac Lh}+\frac{M^2}{Z}\right).
 \label{eq:conditional-pke}
\]
Moreover, the same bound  with $L+\ln2$ in place of $L$, holds for $\mathop{\mathbb E}\limits_{\mathcal O\leftarrow\mathcal R}|a_{\mathcal A}(\mathcal O)|$. 
\end{lemma}

\begin{proof}
For each $b\in\bits$, define hybrids $\mathsf H_b^{(0)},\mathsf H_b^{(1)},\mathsf H_b^{(2)}$, each outputting the adversary's bit. Write $\pi_b^{(k)}=\Pr[\mathsf H_b^{(k)}=1]$. All replacements below apply throughout the execution, including preprocessing.

\textbf{Hybrid $\mathsf H_b^{(0)}$: the real experiment under $\mathcal R$.} Draw the oracle description from $\mathcal R$ and choose $J$ uniformly, independently of it. Privately choose a uniform $h$-subset $H\subseteq[N]$ containing $J$. Give the adversary the public key $J$ and an encryption of $b$. This is the real experiment: $O_{\mathrm{samp}}$ gives every p-key probability $1/N$, the secret key is not otherwise used, and $H$ has no effect. In particular, $\pi_0^{(0)}-\pi_1^{(0)}=\mathop{\mathbb E}\limits_{\mathcal R}a_{\mathcal A}(\mathcal O)$. Comparing messages $0$ and $1$ suffices, since equal messages give no difference and reversing them can be absorbed into the adversary's output rule.

\textbf{Hybrid $\mathsf H_b^{(1)}$: replace the other pairs of rows in $H$.} Next, we want to use the guarantees of \Cref{lem:hidden-subset-entropy} to move to a hybrid where we replace the variables for the challenge p-key $J$ with independently sampled variables from $\mathcal{P}$ in order to invoke \Cref{lem:fresh-record-pke}. However, note that \Cref{lem:hidden-subset-entropy} only ensures negligible statistical distance for records in the absence of the variables $X_{H\setminus\{J\}}=(X_j)_{j\in H\setminus\{J\}}$. Thus, we first need to replace the variables with independent samples, which is what we do next. Specifically, we sample them independently from the base distribution $\mathcal P$.

In particular, modify $\mathsf H_b^{(0)}$ by replacing $X_j$, for $j\in H\setminus\{J\}$, by independent samples according to the base distribution $\mathcal P$. The variable $H$ remains hidden from the adversaries. Each p-key contributes squared mass $1/N$ to $\ket\Psi$, so its state and $O_{\mathrm{samp}}$ reflection change by $O(\sqrt{h/N})$. $O_{\mathrm{enc}}$ and $O_{\mathrm{dec}}$ change only at p-keys in $H\setminus\{J\}$. Conditional on $J$, the original query states are independent of $H$, and their expected squared weight on these p-keys is at most $(h-1)/(N-1)$. Comparing queries using these states gives $|\pi_b^{(1)}-\pi_b^{(0)}|\leq O(q\sqrt{h/N})$, including superpositions over operation names and lengths. For $h=1$, no pair of rows changes and the cost is zero.

\textbf{Hybrid $\mathsf H_b^{(2)}$: replace the challenge pair of rows.} Modify $\mathsf H_b^{(1)}$ by replacing $X_J$ with an independent base sample. Both the oracle and the challenge ciphertext use this new pair of rows. To compare the hybrids, condition on $W$, the outside pairs of rows $X_{-H}$, and the independent replacement data. The original pairs of rows in $H\setminus\{J\}$ no longer occur in $\mathsf H_b^{(1)}$. Equivalently, $H$ was sampled uniformly first and $J$ uniformly within $H$. Lemma~\ref{lem:hidden-subset-entropy} bounds the change in their joint distribution, averaged over the conditioning data, and gives
\[
 |\pi_b^{(2)}-\pi_b^{(1)}|\leq\sqrt{\frac{L}{2h}}.
\]
The conditioning is used only for this comparison, after the preceding query bound; the conditioned data are not given to the adversary.

In $\mathsf H_b^{(2)}$, the challenge pair of rows is independent of all remaining data. Lemma~\ref{lem:fresh-record-pke} therefore gives
\[
 |\pi_0^{(2)}-\pi_1^{(2)}|
 \leq O\!\left(\frac q{\sqrt N}+\frac q{\sqrt Y}
                  +\frac q{\sqrt M}+\frac{M^2}{Z}\right).
\]
Consequently,
\[
 |\pi_0^{(0)}-\pi_1^{(0)}|
 \leq\sum_{b\in\bits}\sum_{k=1}^2
       |\pi_b^{(k)}-\pi_b^{(k-1)}|
       +|\pi_0^{(2)}-\pi_1^{(2)}|.
\]
Substituting the bounds and using $Y,M\geq N$ proves~\eqref{eq:conditional-pke}.

\textbf{Expected absolute difference.} Append $\sigma=\mathbf 1[a_{\mathcal A}(\mathcal O)<0]$ to $W$, using an independent uniform bit in the reference distribution. This increases the relative entropy by at most $\ln2$. Give $\sigma$ to an adversary $\mathcal A'$ that runs $\mathcal A$ and flips its output when $\sigma=1$. Then
\[
 a_{\mathcal A'}(\mathcal O)=(-1)^\sigma a_{\mathcal A}(\mathcal O)
 =|a_{\mathcal A}(\mathcal O)|.
\]
The bit $\sigma$ need not be computable; it is used only for this comparison. The hidden subset remains independent of the original query states, and in $\mathsf H_b^{(2)}$ the $X_J$ is independent of $\sigma$ as well. Applying the same hybrids with budget $L+\ln2$ bounds $\mathop{\mathbb E}\limits_{\mathcal R}|a_{\mathcal A}(\mathcal O)|$.

\end{proof}

\subsubsection{Point-function Obfuscation}
\label{sec:point-obf}

We use $X_{j_1},G,C,N$ from \Cref{sec:primitive-constructions}. Let $W$ contain the interpolation parameter $\tau$ and the accessible oracle data at other lengths. The base distribution is $\mathcal P_W\prod_{k\in[G]}\mathcal P_{X_k}$.

\paragraph{The simulator.}
For an adversary $\mathcal A$, the simulator $\mathcal S$ has quantum access to the common oracle $\mathcal O$ and the point oracle $O_a$. It uses $O_{\mathrm{samp},2\lambda}$ to sample a key $(j_1,j_2',s)$, then uses $O_{\mathrm{enc},2\lambda}$ to sample a c-text $c$ encrypting one at p-key $(j_1,j_2')$. It independently chooses uniform $d\in\bits^\lambda$ and sets $p=(j_1,d,s,c)$. It then runs $\mathcal A$ internally, provides $p$ as the input program to the adversary, and answers its oracle queries as follows. 
{All queries are forwarded unchanged to the corresponding operation of $\mathcal O$, except encryption and decryption queries at length $2\lambda$ whose p-key prefix is $j_1$, which are handled as follows. These rules are applied coherently to each basis component of a query.}
\begin{itemize}
\item \textbf{$O_{\mathrm{enc}}$ queries at $2\lambda$.} For an $O_{\mathrm{enc},2\lambda}$ query with p-key $(j_1,u)$ and message bit $b$, the simulator uses $O_a$ to test whether $u\oplus d=a$. In particular, it coherently computes $u\oplus d$ and runs it through $O_a$ coherently. Conditioned on accept, it applies $O_{\mathrm{enc}}$ at p-key $(j_1,j_2')$ and the same bit $b$, and if the test rejects, it applies the original operation at $(j_1,u)$, coherently. 
\item \textbf{$O_{\mathrm{dec}}$ queries at $2\lambda$.} For an $O_{\mathrm{dec},2\lambda}$ query with p-key $(j_1,u)$, s-key $y$, and c-text $c'$, the simulator performs the same test with $O_a$ to check if $u\oplus d= a$. If it accepts, it applies the decryption response operation of $\mathcal O$ at $(j_1,j_2',y,c')$. If it rejects, it applies the original operation at $(j_1,u,y,c')$. 
\end{itemize}

\begin{remark}[{Universality of the simulator.}]
 {The simulator is a uniform procedure that, given the code of $\mathcal A$ and an explicit polynomial running-time bound, runs $\mathcal A$ with polynomial overhead and returns its output bit. For each fixed uniform QPT adversary, this gives a uniform QPT simulator as required by Definition~\ref{vbb:def}.}
\end{remark}

The proof template for the simulation security is the same as that in the case of \textsf{qPKE} in \Cref{sec:pke-security}. We start with showing simulation security (as in \Cref{vbb:def}) when the p-key prefix corresponding to the challenge obfuscation has independently sampled oracle data, see \Cref{vbb:fresh-group}. Next, as done for \textsf{qPKE} in \Cref{lem:conditioned-pke},  we reduce the security game in the bounded relative-entropy distribution case to the independently sampled oracle data case, that establishes security in the bounded relative-entropy case, see \Cref{vbb:conditioned-security}. We note that the proof of \Cref{vbb:conditioned-security} is very similar to that \Cref{lem:conditioned-pke} with the main difference that instead of working over $N$ p-keys, we need to work with $G$ p-key prefixes, which results in similar bounds as in \Cref{lem:conditioned-pke} but with $N$ replaced with $G$.

For a fixed oracle $\mathcal O$ and point $a$, let $p_{\mathcal A,a}^{\mathcal O}(j_1)$ and $p_{\mathcal S,a}^{\mathcal O}(j_1)$ be the probabilities that the real and simulated experiments output one, respectively, conditioned on the prefix of the sampled p-key being $j_1$. This prefix is uniform in $[G]$ in both experiments. Thus, their overall difference is
\[
\begin{aligned}
 d_a(\mathcal O)
 & =\Pr[\mathcal A^{\mathcal O}(\Obf^{\mathcal O}(1^\lambda,a))=1]
   -\Pr[\mathcal S^{\mathcal O,O_a}(1^\lambda)=1]\\
 & =\frac1G\sum_{j_1\in[G]}
   \left(p_{\mathcal A,a}^{\mathcal O}(j_1)-p_{\mathcal S,a}^{\mathcal O}(j_1)\right).
\end{aligned}
\]

\begin{lemma}[Point-function obfuscation with independently sampled oracle data for the challenge p-key prefix]
\label{vbb:fresh-group}
Fix the interpolation parameter $\tau$, a value $j_1\in[G]$ for the prefix of the p-key, a point $a\in\bits^\lambda$, and arbitrary external data, including all oracle data outside $X_{j_1}$. If $X_{j_1}$ is sampled independently from $\mathcal P_{X_{j_1}}$, then every adversary satisfies
\[
 \left|\mathop{\mathbb E}\limits_{X_{j_1}\leftarrow\mathcal P_{X_{j_1}}}
 \left[p_{\mathcal A,a}^{\mathcal O}(j_1)
       -p_{\mathcal S,a}^{\mathcal O}(j_1)\right]\right|
 \leq O\left(\frac q{\sqrt C}+\frac q{\sqrt N}\right),
\]
where $q$ bounds all oracle calls in the real and simulated experiments, including preparation and preprocessing.
\end{lemma}

\begin{proof}
The proof template is similar to that of the proof of \Cref{lem:fresh-record-pke}. The proof begins as in the proof of \Cref{lem:fresh-record-pke} by modifying the sampling oracle. However, in this case, the obfuscated program reveals the s-key, so the same argument for replacing decryption by default responses does not apply. Instead, we modify encryption and decryption at a p-key with the same prefix and an independently sampled suffix. Exchanging the corresponding pairs of rows then lets us identify the final hybrid with the simulated experiment.

Fix $d\in\bits^\lambda$, put $J=(j_1,a\oplus d)$, and write $\mathcal T=(A,B)$ for the tables at length $2\lambda$. Using~\eqref{eq:reflections}, write
\[
 U_{\mathrm{enc}}^{\mathcal T}(k,b)=H_{\phi_{k,b}^{\mathcal T}},
 \qquad
 U_{\mathrm{dec}}^{\mathcal T}(k,y,z)=Q(\beta_{k,y}^{\mathcal T}(z)).
\]
Define hybrids $\mathsf H^{(0)},\ldots,\mathsf H^{(3)}$, and write $\pi^{(r)}=\Pr[\mathsf H^{(r)}=1]$. Each hybrid changes only the specified operations. All replacements affect the adversary's queries throughout its execution, including preprocessing, at length $2\lambda$.

\textbf{Hybrid $\mathsf H^{(0)}$: the real experiment.} Sample $X_{j_1}\leftarrow\mathcal P_{X_{j_1}}$, independently choose uniform $j_2'\in[C]$, and put $J'=(j_1,j_2')$. Measure $\ket{\psi_J^{\mathcal T}}$ and $\ket{\phi_{J,1}^{\mathcal T}}$ separately to obtain $s,c$. Give $p=(j_1,d,s,c)$ to $\mathcal A$ with the oracle determined by $\mathcal T$.

\textbf{Hybrid $\mathsf H^{(1)}$: remove prefix $j_1$ from the sampling operation.} Replace $O_{\mathrm{samp},2\lambda}$ by $H_{\Psi_{-j_1}^{\mathcal T}}$, where $\ket{\Psi_{-j_1}^{\mathcal T}} =\frac1{\sqrt{N-C}} \sum_{\substack{k=(k_1,k_2)\in[N]\\k_1\ne j_1}} \ket{k}\ket{\psi_k^{\mathcal T}}.$

Measuring $\ket{\Psi^{\mathcal T}}$ gives prefix $j_1$ with probability $C/N=1/G$, giving $|\pi^{(1)}-\pi^{(0)}|\leq O(q/\sqrt G)$.

\textbf{Hybrid $\mathsf H^{(2)}$: change encryption and decryption at $J'$.} For each p-key $k$, replace $U_{\mathrm{enc}}^{\mathcal T}(k,b)$ and $U_{\mathrm{dec}}^{\mathcal T}(k,y,z)$ by
\[
 \widetilde U_{\mathrm{enc}}(k,b)=
 \begin{cases}
 U_{\mathrm{enc}}^{\mathcal T}(J,b),&k=J',\\
 U_{\mathrm{enc}}^{\mathcal T}(k,b),&k\ne J',
 \end{cases}
\]
\[
 \widetilde U_{\mathrm{dec}}(k,y,z)=
 \begin{cases}
 U_{\mathrm{dec}}^{\mathcal T}(J,y,z),&k=J',\\
 U_{\mathrm{dec}}^{\mathcal T}(k,y,z),&k\ne J'.
 \end{cases}
\]
The queries in $\mathsf H^{(1)}$ are independent of $j_2'$, so measuring a query's p-key would give $J'$ with probability at most $1/C$, averaged over $j_2'$. Thus $|\pi^{(2)}-\pi^{(1)}|\leq O(q/\sqrt C)$.

\textbf{Hybrid $\mathsf H^{(3)}$: restore sampling with the exchanged pairs of rows.} Set $\mathcal T'=\operatorname{Swap}_{J,J'}(\mathcal T)$, exchanging the pairs of rows at $J,J'$ in both tables, and replace the sampling operation by $H_{\Psi^{\mathcal T'}}$. Since $\ket{\Psi_{-j_1}^{\mathcal T'}}=\ket{\Psi_{-j_1}^{\mathcal T}}$, the same bound as in $\mathsf H^{(1)}$ gives $|\pi^{(3)}-\pi^{(2)}|\leq O(q/\sqrt G)$.

We have $\ket{\psi_J^{\mathcal T}}=\ket{\psi_{J'}^{\mathcal T'}}$ and $\ket{\phi_{J,1}^{\mathcal T}}=\ket{\phi_{J',1}^{\mathcal T'}}$. For either encryption or decryption, omitting $b$ or $(y,z)$ from the notation,
\[
 \widetilde U(k)=
 \begin{cases}
 U^{\mathcal T'}(J'),&k=J,\\
 U^{\mathcal T'}(k),&k\ne J.
 \end{cases}
\]
Thus $\mathsf H^{(3)}(\mathcal T;d,j_2')\equiv\mathsf{Sim}(\mathcal T';d,j_2')$. For every fixed $d,j_2'$, the tuple $X'_{j_1}$ has distribution $\mathcal P_{X_{j_1}}$, since its pairs of rows are independent and identically distributed. Hence $\mathsf H^{(3)}$ has the simulated distribution.

The real experiment has uniform $d$. Averaging over $d$, applying the triangle inequality, and using $G=C$ gives
\[
\begin{aligned}
 &\left|\mathop{\mathbb E}\limits_{X_{j_1}\leftarrow\mathcal P_{X_{j_1}}}
 \left[p_{\mathcal A,a}^{\mathcal O}(j_1)
       -p_{\mathcal S,a}^{\mathcal O}(j_1)\right]\right|\\
 &\qquad\leq O\left(\frac q{\sqrt G}+\frac q{\sqrt C}\right)
 =O\left(\frac q{\sqrt C}\right),
\end{aligned}
\]
which implies the stated bound.
\end{proof}

\begin{lemma}[Point-function obfuscation under a relative-entropy bound]
\label{vbb:conditioned-security}
Suppose $\mathcal R$ has relative entropy at most $L$ from $\mathcal P_W\prod_{k\in[G]}\mathcal P_{X_k}$. For every uniform QPT adversary $\mathcal A$, let $q\geq1$ bound all oracle calls in the real and simulated experiments, including preparation and preprocessing. Then, for every point $a\in\bits^\lambda$ and every integer $1\leq h\leq G$,
\[
 \left|\mathop{\mathbb E}\limits_{\mathcal O\leftarrow\mathcal R}d_a(\mathcal O)\right|
 \leq O\left(q\sqrt{\frac hG}+\sqrt{\frac Lh}\right).
\]
{Moreover, the same bound with $L+(\lambda+1)\ln2$ in place of $L$ holds for $\mathop{\mathbb E}\limits_{\mathcal O\leftarrow\mathcal R}\max_{a\in\bits^\lambda}|d_a(\mathcal O)|$.}
\end{lemma}

\begin{proof}
The proof is the same as that for \Cref{lem:conditioned-pke}, but adapted to the setting of the point function construction. Fix a point $a$. For each $E\in\{\mathcal A,\mathcal S\}$, define hybrids $\mathsf H_E^{(0)},\mathsf H_E^{(1)},\mathsf H_E^{(2)}$, where $E=\mathcal A$ denotes the real experiment and $E=\mathcal S$ the simulated experiment. Below, each hybrid outputs the adversary's bit, and let $\pi_E^{(k)}:=\Pr[\mathsf H_E^{(k)}=1]$. 

\textbf{Hybrid $\mathsf H_E^{(0)}$: the original experiment under $\mathcal R$.} Draw the oracle description from $\mathcal R$ and independently choose a uniform prefix $j_1\in[G]$. In the real experiment, choose a uniform suffix $j_2\in[C]$, sample $s,c$ honestly at p-key $(j_1,j_2)$ with $c$ encrypting one, set $d=j_2\oplus a$, and give $p=(j_1,d,s,c)$ to $\mathcal A$ with the original oracle. In the simulated experiment, choose a uniform suffix $j_2'\in[C]$, sample $s,c$ in the same way at $(j_1,j_2')$, and choose $d$ independently and uniformly; give the same program format to $\mathcal A$, answering its queries by the simulator's rules above. Since, choosing a uniform $h$-subset $H\subseteq[G]$ containing $j_1$ without using it in the experiment has no effect on the experiment statistics, and since the sampled prefix is uniform for every oracle, we get $\pi_{\mathcal A}^{(0)}-\pi_{\mathcal S}^{(0)}=\mathop{\mathbb E}\limits_{\mathcal R}d_a(\mathcal O)$.

\textbf{Hybrid $\mathsf H_E^{(1)}$: replace the data at the other prefixes in $H$.} Modify $\mathsf H_E^{(0)}$ by replacing each $X_k$, for $k\in H\setminus\{j_1\}$, by an independent sample from $\mathcal P_{X_k}$. Use these replacements in every common-oracle operation, leaving $X_{j_1}$, the prepared program, and the simulator's redirection rule unchanged. The subset $H$ remains hidden. The corresponding query comparison in the proof of \Cref{lem:conditioned-pke}, with $N$ replaced by $G$, applies because the simulator preserves p-key prefixes and gives
\[
 |\pi_E^{(1)}-\pi_E^{(0)}|\leq O\left(q\sqrt{\frac hG}\right).
\]

\textbf{Hybrid $\mathsf H_E^{(2)}$: replace the data at the sampled prefix.} Modify $\mathsf H_E^{(1)}$ by replacing $X_{j_1}$ with an independent sample from $\mathcal P_{X_{j_1}}$. Use this new tuple both to prepare the program as in $\mathsf H_E^{(0)}$ and to answer all oracle queries, including the simulator's redirected queries; retain $W$, $H$, $j_1$, the outside data $X_{-H}$, and the independent replacements at the other prefixes in $H$. Applying \Cref{lem:hidden-subset-entropy} to $X_1,\ldots,X_G$, exactly as in the corresponding comparison in \Cref{lem:conditioned-pke}, gives
\[
 |\pi_E^{(2)}-\pi_E^{(1)}|\leq\sqrt{\frac{L}{2h}}.
\]

In $\mathsf H_E^{(2)}$, the tuple $X_{j_1}$ is independent of all remaining oracle data, so \Cref{vbb:fresh-group} gives $|\pi_{\mathcal A}^{(2)}-\pi_{\mathcal S}^{(2)}|\leq O(q/\sqrt C+q/\sqrt N)$. Combining these comparisons by the triangle inequality yields
\[
 \left|\mathop{\mathbb E}\limits_{\mathcal R}d_a(\mathcal O)\right|
 \leq O\left(q\sqrt{\frac hG}+\sqrt{\frac Lh}
             +\frac q{\sqrt C}+\frac q{\sqrt N}\right).
\]
Since $C=G$, $N=GC$, and $h\geq1$, the last two terms are absorbed into the first.

For the maximum absolute difference, use the final argument of the same \textsf{qPKE} proof, appending to $W$ both a maximizing point $a$ and the sign bit $\sigma=\mathbf1[d_a(\mathcal O)<0]$. An independent uniform point and bit in the reference distribution give entropy budget $L+(\lambda+1)\ln2$, proving the second claim.
\end{proof}

\subsection{Moving to an oracle with no \textsf{qOWF}s}
\label{sec:conditioned-oracle}

Let $\mathcal P$ be the oracle distribution defined in Section~\ref{sec:construction}. We will condition on this distribution to obtain a distribution where \textsf{qOWF}s do not exist. 

\paragraph{Whole-function agreement.} { Fix an evaluator $E$, a parameter $\lambda\ge2$, and a cutoff $a=a(\lambda)=\lceil c\log_2\lambda\rceil$, where $c\ge1$ is fixed. Let $q=q(\lambda)\ge1$ be a polynomial upper bound on the number of oracle queries and the lengths queried by $E$.} Keep the oracle blocks at lengths $n<a$ fixed. Sample the remaining blocks independently as specified by $\mathcal P$. We use this sampling procedure in the next four lemmas, {with this cutoff.}

At one remaining length, $\mathcal P$ samples independent tables $A,B$ and a uniform $\tau\in[0,1]$. Write $F(A,B,\tau)$ for the resulting block. For fixed $A,B$, increasing $\tau$ from zero to one continuously changes the operations defined by $A$ into those defined by $B$. At $\tau=1$, the operations depend only on $B$. They therefore equal the operations at $\tau=0$ for any block whose first table is $B$. Lemma~\ref{lem:edge-properties} also bounds the operator-norm change by a constant times the change in $\tau$.

The running-time bound of $E$ limits the lengths it can query to a finite set. Applying Lemma~\ref{lem:stationary-paths} independently at these lengths at or above the cutoff constructs random continuous paths between two independent samples of those blocks. At every fixed value of the path parameter, those blocks have exactly the distribution specified by $\mathcal P$. The blocks below $a$ and the $\PSPACE$ oracle stay fixed. Define $d(\mathcal O,\mathcal O')$ as the maximum operator-norm difference between corresponding operations at lengths $E$ can query. This also bounds the difference between queries made in superposition over operation names and lengths. Taking the maximum, rather than a sum, keeps the path's change bound independent of the lengths and their number.

Let $D\subseteq\bits^{m(\lambda)}$ be a fixed input set, independent of the remaining blocks. Call an oracle \emph{reliable on $D$} if, on every $x\in D$, some output of $E^{\mathcal O}(x)$ has probability at least $5/6$. These unique outputs define a function on $D$.

\begin{lemma}[Whole-function agreement]
\label{cor:whole-function}
If the oracle is reliable on $D$ with probability at least $1-\varepsilon$, where $0\le\varepsilon\le1/2$, then there is a function $f$ on $D$ such that, with probability at least
\begin{equation}
 1-C(q+1)\varepsilon,
 \label{eq:pke-agreement}
\end{equation}
the oracle is reliable on $D$ and the most likely output of $E$ is $f(x)$ for every $x\in D$. Here $C$ is a constant.
\end{lemma}

\begin{proof}
{ Replacing oracle queries one at a time bounds the change} in the probability of any output event by $2q\,d(\mathcal O,\mathcal O')$; intermediate measurements may be deferred. Partition the reliable oracles into sets computing the same function on $D$. Two different sets disagree at some input, where one output has probability at least $5/6$ in one set and at most $1/6$ in the other. For any oracles $\mathcal O,\mathcal O'$ in different sets, the query bound therefore gives
\[
 d(\mathcal O,\mathcal O')
 \ge \frac{5/6-1/6}{2q}=\frac1{3q}.
\]
Here $d$ is the maximum operator-norm difference defined above.  The reliable function classes form a finite measurable family. Let $\varepsilon_{\mathrm{bad}}$ be the probability that the oracle is not reliable on $D$, so $\varepsilon_{\mathrm{bad}}\leq\varepsilon\leq1/2$. Applying Corollary~\ref{cor:barrier-large-set} and its finite-product extension with $\delta=1/(3q)$, one class has probability at least
\[
 1-\varepsilon_{\mathrm{bad}}-\frac{9v_0q\,\varepsilon_{\mathrm{bad}}}{1-\varepsilon_{\mathrm{bad}}}
 \geq1-(1+18v_0q)\varepsilon.
\]
Choose $f$ to be the common function computed by the oracles in this largest class. Since $v_0$ is an absolute constant, taking $C=1+18v_0$ gives~\eqref{eq:pke-agreement}.
\end{proof}

\paragraph{Recovering reliable answers.} We estimate the fixed blocks and simulate the remaining blocks. Averaging the simulated evaluation in polynomial space then recovers the answers on which most oracles agree.

\begin{lemma}[Recovering reliable answers]
\label{lem:mean-recovery}
Suppose that, for every $x\in D$, some answer $a_x$ satisfies
\[
 \Pr_{\mathcal O\leftarrow\mathcal P}
 \bigl[\Pr[E^{\mathcal O}(x)=a_x]\ge5/6\bigr]
 \ge9/10.
\]
A {QPT} procedure outputs classical estimates $T$ of the operations below length $a(\lambda)$, defining a function $g_T$ computable in polynomial space from $\lambda,T$, such that
\[
 \Pr_T[\forall x\in D,\ g_T(x)=a_x]\ge1-\lambda^{-3}.
\]
\end{lemma}
\begin{proof}
{ At lengths below the logarithmic cutoff $a$, the oracle operations
have polynomial dimension. Using the available oracle
access, we reconstruct these operations by standard Hadamard-test
estimation of their matrix entries~\cite[Section~2.2]{AJL06}. For an operation of dimension $d_0$, use controlled Hadamard tests to estimate each matrix entry to additive error $1/(1024qd_0)$. The estimated matrix is then within $1/(1024q)$ in operator norm. Replacing it by a nearest unitary increases this bound by at most a factor of two. Approximate that unitary by a circuit to error $1/(512q)$ and store the circuit description in $T$. Thus each stored circuit has operator-norm error at most $1/(256q)$. The concentration bound for the Hadamard tests and a union bound over all entries give this guarantee simultaneously with probability at least $1-\lambda^{-3}$, using polynomially many tests. Replacing the queries one at a time then changes any output probability by at most $2q/(256q)=1/128$, uniformly over inputs and the remaining blocks.}

For the simulation at lengths $n\ge a$, condition on $\mathcal A_j\cap\mathcal B_j=\varnothing$ for every pair of rows. The excluded probability is at most
\[
 \sum_{n\ge a}4NM^2/Z
 =\sum_{n\ge a}4\cdot2^{-7n}<\frac1{16}.
\]
Thus conditioning changes any averaged output probability by less than $1/16$. The pairs of rows remain independent under this conditioning.

At each such length, let $I_1,I_2,I_3,I_4$ be four consecutive intervals of size $M$ in $[Z]$. {Under this conditioning, the four sets of c-texts $A_{j,0},A_{j,1},B_{j,0},B_{j,1}$ are sampled} by choosing an independent uniform permutation $\pi_j$ for each pair of rows and setting
\[
 A_{j,b}=\pi_j^{-1}(I_{b+1}),\qquad
 B_{j,b}=\pi_j^{-1}(I_{b+3})
 \quad(b\in\bits).
\]

{ Choose $k=O(q)$ to be twice an upper bound on the total number of coherent queries to the s-keys and permutations used to simulate $E$. We replace each uniform permutation by an almost $k$-wise independent permutation from~\cite{KNR09}: on any at most $k$ distinct inputs, its images are within statistical distance $\delta$ of those of a uniform permutation. Such permutations have descriptions of length polynomial in $n,q,\log(1/\delta)$. For each length, use a $k$-wise independent function of $j$ whose output contains the two s-keys and the random bits specifying $\pi_j$. Thus any $k$ pairs of rows have the same joint distribution as if these strings were chosen independently. The functions at different lengths are sampled independently. Only polynomially many lengths can be queried, so all these choices have a polynomial-length description when $\log(1/\delta)$ is polynomial in $\lambda$.}

Fix $\tau$ at each length. Each operation uses constantly many coherent queries to the s-keys and permutations, including their inverses. Applying $\pi_j^{-1}$ to the prescribed interval superpositions prepares the c-text states. Evaluating $\pi_j(c)$ determines the $O_{\mathrm{dec}}$ operation.

Exact $k$-wise independent permutations preserve the averaged output probabilities for this number of quantum queries. For the almost independent permutations, the bound in~\cite[Theorem~1.1]{AL13} lets us choose $\delta$ with polynomial $\log(1/\delta)$ so that the total replacement error over all pairs of rows is negligible. {By~\cite[Fact~2]{Z12}, a quantum algorithm making at most $k/2$ queries to this function has the same averaged output probabilities whether it is $k$-wise independent or fully random. This applies even when a query is in superposition over $j$.}

{ Replace the uniform choice of $\tau\in[0,1]$ by a uniform choice from the equally spaced points
\[
 \{0,2^{-s},2\cdot2^{-s},\ldots,(2^s-1)2^{-s}\},
\]
where $s$ is polynomial in $\lambda$. Rounding $\tau$ down to this grid changes it by at most $2^{-s}$. Lemma~\ref{lem:edge-properties} and the query hybrid bound therefore bound the change in any output probability by $O(q2^{-s})$. Choose $s$ large enough to make this at most $1/128$.} {For each choice of these functions and $\tau$} all simulated operations and the fixed $\PSPACE$ queries are computable in polynomial space. We can therefore approximate each output-bit probability in polynomial space by~\cite{Wat99}. {Enumerating the polynomial-length descriptions of these functions and the grid indices} and averaging these probabilities still uses polynomial space.

{ Condition on the event that every circuit in $T$ has operator-norm error at most $1/(256q)$; its failure probability is at most $\lambda^{-3}$, as proved above. Choose the permutation, grid, simulated-circuit, and numerical approximation errors to be at most $1/128$ each. Together with the $1/128$ error from replacing the fixed operations, their sum is less than $1/16$. We call $T$ accurate when the stated operator-norm bounds hold.} Including the conditioning error, the computed average then differs from the original one by at most $1/8$, uniformly over inputs. Define $g_T$ by rounding each computed output-bit probability at $1/2$.

For every $x\in D$, each bit of $a_x$ is correct with probability at least $(9/10)(5/6)=3/4$ in the original average. The error bound therefore ensures that rounding recovers $a_x$. Thus accurate estimates give $g_T(x)=a_x$ simultaneously for all $x\in D$.
\end{proof}

\paragraph{Inversion.} Under pointwise correctness, we recover the entire function and find any preimage. Under average correctness, we instead sample a preimage nearly uniformly, to avoid concentrating on inputs where recovery fails.

\begin{lemma}[Inversion under pointwise correctness]
\label{lem:strict-inversion}
Fix the blocks below $a(\lambda)$ as in Lemma~\ref{lem:mean-recovery}, and sample the remaining blocks according to $\mathcal P$. Suppose that, with probability at least $9/10$, the resulting oracle is reliable on the full input domain and computes the same function $f_*$. One QPT inverter with $\PSPACE$ access returns a preimage of every $y\in\Image(f_*)$ with probability at least $1-\lambda^{-3}$.
\end{lemma}

\begin{proof}
Lemma~\ref{lem:mean-recovery} gives $g_T=f_*$ everywhere except with probability $\lambda^{-3}$. Given $y$, use the $\PSPACE$ oracle to find the lexicographically first $x$ satisfying $g_T(x)=y$. 
\end{proof}

We now show inversion under average correctness. Let $K$ be the polynomial factor in Lemma~\ref{cor:whole-function}, and write
\[
 \varepsilon=
 \Pr_{\substack{\mathcal O\leftarrow\mathcal P\\
                x\leftarrow\bits^{m(\lambda)}}}
 \left[
   \max_y\Pr[E^{\mathcal O}(x)=y]<5/6
 \right].
\] 

\begin{lemma}[Inversion under average correctness]
\label{lem:average-inversion}
There exists a QPT inverter with access to $\mathcal O$ and the $\PSPACE$ oracle such that, outside a set of oracles of probability at most $O(K\lambda^2\varepsilon)$, its success probability is at least
\[
 1-O(\eta+\lambda^{-2})
\]
whenever $E^{\mathcal O}$ has average correctness error at most $\eta$ under uniform inputs. 
\end{lemma}
\begin{proof}
For each input $x$, let $\varepsilon_x$ be the probability that the oracle is not reliable on $\{x\}$. Then $\mathop{\mathbb E}\limits_x\varepsilon_x=\varepsilon$. For inputs with $\varepsilon_x\le1/(10K)$, Lemma~\ref{cor:whole-function} gives common answers, which Lemma~\ref{lem:mean-recovery} recovers as $g_T$. The expected fraction of inputs either excluded by this threshold or failing to produce their common answer with probability at least $5/6$ is $O(K\varepsilon)$. By Markov's inequality, outside a set of oracles of probability $O(K\lambda^2\varepsilon)$, this fraction is $O(\lambda^{-2})$. Thus, whenever $T$ is accurate and $E^{\mathcal O}$ has average correctness error $\eta$,
\[
 \Pr_x[g_T(x)\ne f(x)]=O(\eta+\lambda^{-2}).
\]

The inverter computes $T$ independently of the challenge. Given $y$, it samples a nearly uniform preimage under $g_T$, or returns failure if none exists. This requires only polynomially many queries to the $\PSPACE$ oracle.

Since $f$ and $g_T$ agree on all but an $O(\eta+\lambda^{-2})$ fraction of inputs, sampling a uniform preimage of the challenge under $g_T$ inverts $f$ with probability $1-O(\eta+\lambda^{-2})$ for a uniform challenge input. The same bound holds for evaluations of $E$, after including evaluation error, estimation failure, and the negligible sampling error.
\end{proof}

{We now essentially established that on oracles where the evaluator is correct, we can invert it with \textsf{PSPACE} oracle. However, there may be oracles where the correctness does not hold on certain security parameters and, thus, the inversion attack may not succeed. Nevertheless, this still allows us to rule out \textsf{qOWF}s because on such parameters correctness does not hold. In other words, we can select an oracle such that at infinitely many selected security parameters for each evaluator, either security does not hold or correctness does not hold.} 

Moreover, since the number of evaluators is countable, we can actually enumerate over all evaluators and choose one oracle such that for every evaluator, there is an infinite sequence of security parameters where the evaluator does not satisfy correctness or can be inverted. This gives us a single oracle where no \textsf{qOWF} exists. We then show in \cref{sec:primitive-security} that \textsf{qPKE} can be built relative to this oracle which gives us the final separation. 

\paragraph{Selecting the distribution.}

For any distribution $\mathcal R$ on the oracle descriptions of Section~\ref{sec:construction}, write $\mathcal R_{\le r}$ for its restriction to lengths at most $r$.

\begin{theorem}[Oracle selection and absence of one-way functions]
\label{thm:oracle-selection}
\label{lem:owf-selection-failures}
There is a distribution $\mathcal Q$ on the oracle descriptions of Section~\ref{sec:construction} with the following properties.
\begin{enumerate}
\item With $\mathcal Q$-probability one, no $\textsf{qOWF}_1$ or $\textsf{qOWF}_2$ exists.
\item For every $r\geq1$, the relative entropy (KL-divergence) of $\mathcal Q$ from the base distribution $\mathcal P$, restricted to lengths at most $r$, is at most $L(r)=2r$:
\begin{equation}
 D(\mathcal Q_{\le r}\Vert\mathcal P_{\le r})\leq L(r)=2r.
 \label{eq:entropy-prefix}
\end{equation}
\end{enumerate}
\end{theorem}

\begin{proof}
The construction of $\mathcal{Q}$ and its proof follow by applying the preceding lemmas to all uniform evaluators. We enumerate all possible evaluators with explicit polynomial running-time bounds. Consider each evaluator infinitely often for pointwise correctness and infinitely often for average correctness under uniform inputs. Each step of the following procedure (which eventually generates $\mathcal{Q}$) chooses a finite group of new oracle blocks and fixes them permanently.

For the evaluator considered at a given step, use its bounds $q$ and $K\ge2$ and its cutoff $a(\lambda)$ from the preceding lemmas. Set
\[
 \theta(\lambda)=\frac1{10^6K(\lambda)\lambda^8}.
\]
Choose the constant $c$ in $a(\lambda)$ once for this evaluator so that $\ln(1/\theta(\lambda))\le a(\lambda)$ at all selected parameters. These logarithmic cutoffs allow polynomial-time estimation of the operations below the cutoff.

Number all steps by $i=1,2,\ldots$, assigning each evaluator and correctness requirement an infinite subsequence of these indices. When we later analyze a fixed evaluator, we restrict to its assigned indices; $i$ always refers to the full sequence of steps.

At step $i$, use a parameter $\lambda_i\ge i+2$ and write
\[
 \theta_i=\theta(\lambda_i),\qquad
 a_i=a(\lambda_i),\qquad
 b_i=\max\{a_i,q(\lambda_i)\}.
\]
Choose the sequence of evaluators, correctness requirements, and increasing parameters in advance, with $a_{i+1}>\max\{2a_i,b_i\}$. Sample the blocks below $a_1$ from $\mathcal P$. At step $i$, keep the blocks below $a_i$ fixed and define $B_i$ according to the correctness requirement being considered:
\begin{itemize}
\item For pointwise correctness, $B_i$ is the event that some input has no output of probability at least $5/6$.
\item For average correctness, $B_i$ is the event that more than a $\theta_i$ fraction of uniformly distributed inputs have no output of probability at least $5/6$.
\end{itemize}
In both cases, these output probabilities are for the chosen evaluator at parameter $\lambda_i$.

The event $B_i$ depends on the chosen evaluator and the new oracle blocks. The output probabilities defining $B_i$ are over the evaluator's internal randomness; the probability below is over sampling the new blocks from $\mathcal P$. The rule below forces $B_i$ whenever it is sufficiently likely. Otherwise, its small probability lets us apply the preceding lemmas to obtain inversion guarantees.

If
\[
 \Pr_{\mathcal P}[B_i\mid\text{fixed blocks}]
 \ge\theta_i,
\]
sample the blocks at lengths $a_i,\ldots,b_i$ from $\mathcal P$ conditioned on $B_i$. Otherwise, sample them from $\mathcal P$ without conditioning. Sample the blocks at lengths $b_i+1,\ldots,a_{i+1}-1$ independently from $\mathcal P$. All sampled blocks are fixed permanently. The events $B_i$ are measurable and depend only on blocks fixed by this step, so later steps cannot change the evaluator's behavior at $\lambda_i$. The selection procedure need not be efficient. Let $\mathcal Q$ be the distribution of the complete oracle description generated by this procedure.

Conditioning on an event of probability $\rho$ contributes $\ln(1/\rho)$ to relative entropy. Thus, for every choice of the previously fixed blocks, step $i$ contributes at most
\[
 \ln(1/\theta_i)\le a_i.
\]
Unconditioned sampling contributes zero. The chain rule for relative entropy therefore gives
\[
 D(\mathcal Q_{\le r}\Vert\mathcal P_{\le r})
 \le\sum_{i:a_i\le r}a_i\le2r,
\]
where the last inequality follows from $a_{i+1}>2a_i$. If $r$ falls within a step, apply the bound through that step and discard the blocks above $r$; this cannot increase relative entropy.

At an unconditioned step $i$, keep the previously sampled blocks fixed and put $\lambda=\lambda_i$ and $\theta=\theta_i$. Here $K$ is the previously chosen bound for which Lemma~\ref{cor:whole-function} gives agreement probability at least $1-K\varepsilon$ for this evaluator at $\lambda_i$. By the selection rule,
\[
 \Pr_{\mathcal P}[B_i\mid\text{previously fixed blocks}]<\theta,
\]
and the new blocks are sampled directly from $\mathcal P$.

In the pointwise case, Lemma~\ref{cor:whole-function} gives a set of oracles of probability at least $1-K\theta$ that reliably compute the same function on every input. In the average case, the expected fraction of unreliable inputs is at most $2\theta$: outside $B_i$ this fraction is at most $\theta$, and $B_i$ itself has probability below $\theta$. Lemma~\ref{lem:average-inversion} therefore gives an exceptional set of probability $O(K\lambda^2\theta)$.

Let $F_i$ be the event of falling outside the common-function set in the pointwise case, or inside the exceptional set in the average case. Set $F_i=\varnothing$ at conditioned steps. The bounds hold for every choice of the previously fixed blocks, so
\[
 \Pr_{\mathcal Q}[F_i]
 \le O(K\lambda^2\theta)=O(\lambda_i^{-6}).
\]
Since $\lambda_i\ge i+2$, these probabilities are summable. By Borel--Cantelli, only finitely many $F_i$ occur with $\mathcal Q$-probability one. At each conditioned step, the conditional law assigns probability zero to $B_i^c$. A countable union of these null failure events has probability zero, so every conditioned step satisfies its event $B_i$ simultaneously with $\mathcal Q$-probability one.

Fix an oracle for which every conditioned step satisfies $B_i$ and only finitely many $F_i$ occur. Fix an evaluator satisfying one of the correctness requirements with negligible error. At its sufficiently late scheduled steps, conditioning on $B_i$ would force correctness error greater than $1/6$ on some input in the pointwise case, or average error greater than $\theta/6$ in the average case. Both contradict correctness, since $\theta$ is inverse polynomial. Hence these steps are unconditioned and outside $F_i$.

Lemma~\ref{lem:strict-inversion} or Lemma~\ref{lem:average-inversion} now gives a uniform inverter with success probability at least $2/3$ at infinitely many selected security parameters for this evaluator and correctness requirement, such that the inverter's description does not depend on these selected parameters. In the pointwise case (\Cref{def:one}), the guarantee holds on every image, and negligible evaluation error also gives the evaluation guarantee for every input sampler. In the average case (\Cref{def:two}), this success guarantee holds for uniformly sampled inputs. Hence, the inverter contradicts the one-wayness guarantee for both definitions (\Cref{def:one,def:two}) in \Cref{sec:def-owf}.

\end{proof}

\begin{corollary}[Oracle distribution with \textsf{qPKE} and point-function obfuscation without \textsf{qOWF}s]
\label{cor:selected-oracle-separation}
Let $\mathcal Q$ be the distribution on the oracle descriptions of Section~\ref{sec:construction} supplied by Theorem~\ref{thm:oracle-selection}. With $\mathcal Q$-probability one, no $\textsf{qOWF}_1$ or $\textsf{qOWF}_2$ exists, while the \textsf{qPKE} construction of Section~\ref{sec:primitive-constructions} is \textsf{IND-CPA} secure and the point-function obfuscator satisfies Definition~\ref{vbb:def}. 

In particular, there is a fixed oracle $\mathcal O\in\supp(\mathcal Q)$ with all these properties and perfect correctness for \textsf{qPKE}, thus yielding \Cref{thm:separation}. 
\end{corollary}

\begin{proof}[Proof of \Cref{cor:selected-oracle-separation}]\label{pg:thm-main-separation}
Fix a uniform adversary for each primitive, and let $q(\lambda)$ bound the queries in their security experiments, including honest preparation and the obfuscation simulator's calls. Let $r(\lambda)\geq2\lambda$ be a polynomial bound on all accessible oracle lengths. Theorem~\ref{thm:oracle-selection} gives relative-entropy budget $L=2r(\lambda)$.

For \textsf{qPKE}, apply Lemma~\ref{lem:conditioned-pke} with $n=\lambda$, $N=2^\lambda$, and $h=\lceil2^{\lambda/2}\rceil$. Use $L+\ln2$ for the absolute-advantage bound. For correctness, apply Lemma~\ref{lem:pke-correctness} to the block at length $\lambda$ with budget $L$, and bound the maximum of the two correctness errors by their sum. Since $q,r$ are polynomial and $M^2/Z=2^{-8\lambda}$, the expected absolute advantage and correctness error are at most $\poly(\lambda)2^{-\Omega(\lambda)}$.

For point-function obfuscation, the block at length $2\lambda$ has $C=G=2^\lambda$, $N=2^{2\lambda}$, $Y=2^{16\lambda}$, $M=2^{4\lambda}$, and $Z=2^{24\lambda}$. Apply Lemma~\ref{vbb:conditioned-security} with $h=\lceil2^{\lambda/2}\rceil$ and $\Lambda=L+(\lambda+1)\ln2$. Its expected maximum absolute gap is at most $\poly(\lambda)2^{-\Omega(\lambda)}$. Honest preparation and evaluation use only this block, so Lemma~\ref{lem:vbb-correctness} with relative-entropy budget $4\lambda$ gives expected correctness error $2^{-\Omega(\lambda)}$.

Let $e(\mathcal O)$ denote either absolute security gap (maximized over points for obfuscation) or either correctness error. For each positive integer $c$, Markov's inequality gives
\[
 \sum_{\lambda\geq1}\Pr_{\mathcal O\leftarrow\mathcal Q}
 [e(\mathcal O)>\lambda^{-c}]
 \leq\sum_{\lambda\geq1}\lambda^c\,
 \mathop{\mathbb E}\limits_{\mathcal O\leftarrow\mathcal Q}e(\mathcal O)<\infty.
\]
Borel--Cantelli and a countable intersection over uniform adversaries and positive integers $c$ show that all these errors are negligible with probability one. Intersect this event with the probability-one absence of both $\textsf{qOWF}_1$s and $\textsf{qOWF}_2$s from Theorem~\ref{thm:oracle-selection}, and call the resulting event $\mathcal H$. Thus $\mathcal Q(\mathcal H)=1$, with the single obfuscation simulator of Section~\ref{sec:point-obf} working for all points. Independent bitwise encryption extends \textsf{qPKE} to polynomial-length messages by a polynomial hybrid and a union bound.

For the fixed oracle's exact \textsf{qPKE} correctness, let $C_n$ be the event that some pair of rows at length $n$ has $\mathcal A_j\cap\mathcal B_j\ne\varnothing$, and put $\mathcal G=\bigcap_{n\geq1}C_n^c$. At selection step $i$, conditioning increases probabilities by at most $\theta_i^{-1}\leq e^{a_i}$. Using the construction of $\mathcal Q$ and $\mathcal P(C_n)\leq4\cdot2^{-7n}$, we obtain
\[
 \begin{aligned}
 \mathcal Q(C_n)&\leq
 \exp\!\left(\sum_{i:a_i\leq n}a_i\right)\mathcal P(C_n)
 \leq4(e^2/128)^n,\\
 \mathcal Q(\mathcal G^c)&\leq4\sum_{n\geq1}(e^2/128)^n<1.
 \end{aligned}
\]
Hence $\mathcal Q(\mathcal G\cap\mathcal H)=\mathcal Q(\mathcal G)>0$. The support of $\mathcal Q$ has full measure, so choose $\mathcal O\in\mathcal G\cap\mathcal H\cap\supp(\mathcal Q)$. Observation~\ref{obs:oracle-properties} gives exact honest \textsf{qPKE} correctness there, completing the proof.
\end{proof}

\begin{corollary}[Access to unitary implementations]
\label{cor:unitary-access}
There is no uniform relativizing fully black-box construction of a $\textsf{qOWF}_1$ or $\textsf{qOWF}_2$ from either primitive, even with access to full unitary implementations of its algorithms and their inverse, conjugate, transpose, and controlled versions. For point-function obfuscation, the reduction must be valid for the single-program guarantee of Definition~\ref{vbb:def}.
\end{corollary}

\begin{proof}
Both the construction and its security reduction must remain valid relative to a common additional oracle. Defer measurements, retain discarded registers, and store random coins in workspace to obtain a uniform polynomial-size unitary implementation $W^{\mathcal O}$ of each primitive algorithm. Reverse its gates for the inverse, conjugate them for the entrywise conjugate, use $(W^{\mathcal O})^T=(\overline{W^{\mathcal O}})^\dagger$ for the transpose, and control each gate for controlled access. The required oracle operations are available. These implementations act correctly on arbitrary inputs and workspace, while honest challenges keep their private workspace hidden.

Instantiate the proposed construction at the oracle of Corollary~\ref{cor:selected-oracle-separation} and expand its implementation queries into these circuits. The result is a uniform QPT evaluator with polynomial overhead. Either its \textsf{qOWF} correctness fails, or the inversion argument in Theorem~\ref{thm:oracle-selection} supplies an inverter succeeding at infinitely many parameters. The latter contradicts the reduction and primitive security. The inversion guarantee covers every challenge sampler for $\textsf{qOWF}_1$ and the prescribed uniform sampler for $\textsf{qOWF}_2$.
\end{proof}

\section{Constructing \textsf{qPKE} from \textsf{qiO} and \textsf{EV-OWPuzz}}
\label{chap:constructions}
In this section, we construct \textsf{qPKE} from \textsf{EV-OWPuzz} and \textsf{qiO}.
We use the \textsf{qiO} of Definition~\ref{def:semantic-qio}, which obfuscates quantum circuits computing classical functions and produces classical programs. Let $(O,E)$ be its obfuscation and evaluation algorithms, respectively. For brevity, we remove the security parameter $1^\lambda$ from their input syntax. We amplify its correctness as in Remark~\ref{rem:qio-simultaneous}, so that, except for a negligible fraction of obfuscated programs, evaluation is correct with negligible error on every input. Throughout this section, all samplers, verifiers, obfuscation algorithms, and adversaries are QPT unless stated otherwise, with classical inputs and outputs. All lengths are bounded by fixed polynomials in $\lambda$, and programs compared under \textsf{qiO} are padded if required to have the same circuit size along with input and output lengths.

\subsection{Suppression lemma}
\label{sec:ev-qio}
\label{sec:suppression}

For an \textsf{EV-OWPuzz} scheme $(\textsf{Samp},V)$ with correctness error $\eta$, write $v_p(a):=\Pr[V(p,a)=1]$ for any $p,a$, as in Definition~\ref{def:ev-puzzle}. A puzzle verifier may accept many answers, including answers with intermediate acceptance probabilities. The goal of this section is to design a procedure that accepts only one accepting answer and suppresses the acceptance probability of every other input, including other previously valid answers. We do this so that the resulting circuit implements a point function, which can then be used to prove the point-hiding property for \textsf{qiO} that we need to prove security of our \textsf{qPKE} construction given in \Cref{sec:pke-construction}. We follow the threshold selection and hashing argument of Aharonov, Ben-Or, Brand\~ao, and Sattath~\cite[Lemma~39 and Section~5]{ABBS22}, and then use independent repetition to reduce the error.

\begin{lemma}[Suppressing all but one accepting answer]\label{lem:ev-isolation}
For any \textsf{EV-OWPuzz} with answers padded to a fixed length $m=m(\lambda)\geq\lambda+5$, and for any sampled puzzle $p$ that has an answer $a$ satisfying $v_p(a)\geq2/3$ (which happens with probability at least $1-3\eta$ over the sampling of $p$), there is an efficient randomized procedure that with probability at least inverse polynomial in $m$, 
constructs a polynomial-size quantum circuit $W_p$ such that there is an answer $a^\star$ satisfying 
$\Pr[W_p(a^\star)=1]\geq1-2^{-\lambda}$, $
\Pr[W_p(a)=1]\leq2^{-\lambda}$ for every $a\ne a^\star,$\footnote{i.e., $W_p$ computes a point function at the point $a^\star$} and $v_p(a^\star)>1/3$. 
\end{lemma}

\begin{proof}
We first repeat $V(p,\cdot)$ independently $O(\log m)$ times and take a majority. Write $q(a)$ for the resulting acceptance probability, and choose the repetition count so that
\[
v_p(a)\leq1/3\Longrightarrow q(a)\leq\frac{1}{8m},\qquad
v_p(a)\geq2/3\Longrightarrow q(a)\geq1-\frac{1}{8m}.
\]
Note that this does not give any guarantee on answers with previously intermediate acceptance probabilities. Next, we guess $j$ uniformly from $\{1,\ldots,m-2\}$ and $k$ uniformly from $\{0,\ldots,m\}$, and sample a pairwise independent hash $h:\{0,1\}^m\to\{0,1\}^{k+2}$ using a uniform binary matrix and an independent uniform offset. Thus, distinct inputs have independent uniform hash values. We want the zero bucket, i.e., the set of answers that hash to zero, to keep one answer with acceptance at least $(j+1)/m$ and no other answer with acceptance greater than $j/m$. Define $W_p$ to reject inputs whose hash is nonzero; otherwise, it repeats the amplified verifier independently $O(m^2\lambda)$ times and accepts when the fraction of accepting trials exceeds $(j+1/2)/m$.

We now bound the probability that these guesses succeed. For $j=1,\ldots,m-1$, let $S_j=\{a:q(a)\geq j/m\}$. These sets are nested and $S_{m-1}$ is nonempty. Some $j\in\{1,\ldots,m-2\}$ satisfies $|S_j|\leq4|S_{j+1}|$; otherwise $|S_1|>4^{m-2}|S_{m-1}|>2^m$. For this $j$, the sets
\[
A=\{a:q(a)\geq(j+1)/m\},\qquad B=\{a:q(a)>j/m\}
\]
satisfy $\varnothing\ne A\subseteq B$ and $|B|\leq4|A|$.

Fix this guess of $j$ and a guess satisfying $2^k\leq|B|<2^{k+1}$. Write $M=2^{k+2}$. For each $a\in A$, pairwise independence and a union bound give
\[
\begin{aligned}
&\Pr[h(a)=0\text{ and }h(b)\ne0\text{ for all }b\in B\setminus\{a\}]\\
&\hspace{2em}\geq\frac1M-\frac{|B|-1}{M^2}\geq\frac{1}{2M}.
\end{aligned}
\]
These events are disjoint for distinct $a\in A$. Since $M\leq4|B|$, their total probability is at least $|A|/(8|B|)\geq1/32$. The guesses of $j$ and $k$ contribute factors $1/(m-2)$ and $1/(m+1)$, giving success probability at least $1/[32(m-2)(m+1)]$, which is inverse polynomial in $m$.

On this event, the zero bucket (i.e., the set of answers that hash to zero) contains one answer $a^\star$ with acceptance at least $(j+1)/m$, and every other retained answer has acceptance at most $j/m$. Hoeffding's inequality applied to the midpoint test therefore gives error at most $2^{-\lambda}$ on every input. Moreover, $q(a^\star)\geq2/m>1/(8m)$ implies $v_p(a^\star)>1/3$ by the initial amplification.

Finally, whenever the lemma's hypothesis fails, the verifier rejects every answer with probability greater than $1/3$. Since it rejects an honestly sampled puzzle-answer pair with probability at most $\eta$, such puzzles occur with probability at most $3\eta$.
\end{proof}

\subsection{From \textsf{EV-OWPuzz} to Hiding a Random Point under \textsf{qiO}}
\label{sec:point-hiding}

We now use the suppression lemma to prove the point-hiding property needed for the security of our \textsf{qPKE} scheme, described in \Cref{sec:pke-construction}. Write $P_u(x)=\mathbf{1}\{x=u\}$ for the point function at $u\in\{0,1\}^m$, where $m$ is the answer length fixed in Lemma~\ref{lem:ev-isolation}, and write $Z(x)=0$ for the zero function on the same domain. To apply obfuscation security, we pad all the programs that we will compare to a common size $s=s(\lambda)$, which is a fixed polynomial large enough for every circuit size obfuscated in the proof below. 

\begin{lemma}[Hiding a random point]
\label{lem:ev-point-hiding}
Assume \textsf{EV-OWPuzz} exists and let $(O,E)$ be a \textsf{qiO} scheme satisfying
Definition~\ref{def:semantic-qio}.
For every polynomial $t$, $t$ independent obfuscations of $P_u$ under $O$, for a uniformly sampled $u\leftarrow\bits^m$, is computationally
indistinguishable from $t$ independent obfuscations of $Z$ under $O$.
\end{lemma}

We adapt the single-point extraction argument of Boyle, Chung, and Pass~\cite[Section~6.1, Theorem~6.2]{BCP14} to the quantum circuit obtained from suppression.

\begin{proof}
Suppose a QPT distinguisher $D$ has absolute distinguishing gap at least $\delta(\lambda)=1/q(\lambda)$ for a polynomial $q$ and infinitely many security parameters. Given a puzzle $p$, our inverter runs suppression once to obtain $W$, chooses $r\leftarrow\{0,1\}^m$, and forms the padded circuit $Q_r(x)=W(x\oplus r)$.

By \Cref{lem:ev-isolation}, conditioned on the suppression procedure being successful with answer $a^\star$, $Q_r$ computes the point function at $u=a^\star\oplus r$ with error at most $2^{-\lambda}$ on every input. Since $r$ is uniform, so is $u$. By the indistinguishability security of \textsf{qiO}, conditioned on the successful trial of the suppression procedure, replacing the $t$ independent obfuscations of $P_u$ by obfuscations of $Q_r$, one at a time, changes $D$'s average acceptance probability by a negligible additive term. Moreover, the negligible change in acceptance probability holds for every $W$ in the support of a successful trial and all $r$. Otherwise, select $W,r,a^\star$ for which the change is at least inverse-polynomial at infinitely many security parameters, and compare the acceptance probabilities for the obfuscations of $Q_r,P_u$, i.e., output $1$ on acceptance. This gives a distinguisher that violates the indistinguishability security of \textsf{qiO} (see Definition~\ref{def:semantic-qio}).

Put $\alpha_r=\Pr[D(O(Q_r)_1,\ldots,O(Q_r)_t)=1]$ and $\beta=\Pr[D(O(Z)_1,\ldots,O(Z)_t)=1]$. For all sufficiently large parameters at which $D$ has gap at least $\delta$, we have $|\mathop{\mathbb E}\limits_r[\alpha_r-\beta]|\geq\delta/2$. Since $|\alpha_r-\beta|\leq1$, at least a $\delta/4$ fraction of choices of $r$ satisfy $|\alpha_r-\beta|\geq\delta/4$.

For each input position $i$, the inverter forms
\[
Q_{r,i}(x)=
\begin{cases}
Q_r(x),&x_i=1,\\
0,&x_i=0,
\end{cases}
\]
and pads it to size $s$. If $u_i=1$, this circuit computes the same point function as $Q_r$; if $u_i=0$, it computes zero, in both cases with error at most $2^{-\lambda}$ on every input. Let $\alpha_{r,i}$ be $D$'s acceptance probability on $t$ independent obfuscations of $Q_{r,i}$. By the indistinguishability security of \textsf{qiO}, $\alpha_{r,i}$ is negligibly close to $\alpha_r$ when $u_i=1$ and to $\beta$ when $u_i=0$.

The inverter estimates $\alpha_r$, $\beta$, and all $\alpha_{r,i}$ to accuracy $\delta/100$, using fresh independent obfuscations and independent runs of $D$, and then taking the fraction of accepting runs. Taking $O(\delta^{-2}\log(m+2))$ samples per estimate makes all estimates accurate with probability at least $3/4$ by standard Chernoff bounds and a union bound over the $m+2$ estimates. Next, the inverter sets $\widehat u_i=1$ if the estimate of $\alpha_{r,i}$ is closer to that of $\alpha_r$, and sets $\widehat u_i=0$ otherwise. For any $r$ satisfying $|\alpha_r-\beta|\geq\delta/4$ and sufficiently large $\lambda$, the negligible differences above are less than $\delta/100$. When all estimates are accurate (which conditioned on the suppression procedure being successful, happens with probability at least $3/4$), the distance to the correct estimate is at most $3\delta/100$, while the distance to the other is at least $\delta/4-3\delta/100$. Thus, conditioned on successful suppression and a choice of $r$ satisfying $|\alpha_r-\beta|\geq\delta/4$, $\widehat u=u$, and hence the inverter outputs $\widehat u\oplus r=a^\star$, with probability at least $3/4$ over the independent obfuscations and runs of $D$. 
Note that an honestly sampled puzzle satisfies the suppression hypothesis with probability at least $1-3\eta$. For each such puzzle, suppression succeeds with probability at least $\gamma=\gamma(\lambda)$, an inverse-polynomial bound supplied by Lemma~\ref{lem:ev-isolation}. Moreover, conditioned on successful suppression, there are at least a $\delta/4$ fraction of choices of $r$  such that $|\alpha_r-\beta|\geq\delta/4$, and hence, permits the inverter to recover $a^\star$ with probability at least $3/4$. Finally, we know that $a^\star$ is accepted by the puzzle verifier with probability at least $1/3$. Hence the probability that the inverter outputs an answer that $V(p,\cdot)$ (the puzzle verifier's verification with respect to $p$) accepts, is at least
\[
(1-3\eta)\,\gamma\,\frac{\delta}{4}\,\frac34\,\frac13
=\frac{(1-3\eta)\gamma\delta}{16}
\]
on infinitely many security parameters, contradicting the \textsf{EV-OWPuzz} security since $\frac{(1-3\eta)\gamma\delta}{16}$ is non-negligible in $\lambda$.

Finally, the inverter is clearly a uniform QPT algorithm, since suppression and obfuscation run in polynomial time on every trial, which completes the proof of the lemma.
\end{proof}

\subsection{\textsf{qPKE} construction}
\label{sec:pke-construction}

We can now state the construction and prove security of it using Lemma~\ref{lem:ev-point-hiding}. As noted in the technical overview, the \textsf{EV-OWPuzz} assumption only shows up in the security proof to prove point hiding; the \textsf{qPKE} algorithms themselves only use \textsf{qiO}.

\begin{theorem}[Public-key encryption from puzzles and quantum obfuscation]\label{thm:ev-qio-pke}
There exists a construction of \textsf{IND-CPA} public-key encryption with classical keys and ciphertexts (as defined in \Cref{def:pke}) given an \textsf{EV-OWPuzz} scheme (\Cref{def:ev-puzzle}) and a \textsf{qiO} scheme (as defined in Definition~\ref{def:semantic-qio}).
\end{theorem}

\begin{proof}
It suffices to encrypt one bit. Let $m,s$ be as in Lemma~\ref{lem:ev-point-hiding}. Choose a polynomial size $s'$ large enough for a circuit containing an obfuscation of a size-$s$ program and running the evaluator $E$. Such a bound exists because $O$ and $E$ run in polynomial time. We fix $s$ first, then $s'$, and pad all ciphertext programs to size $s'$.
\begin{enumerate}
\item $\KeyGen(1^\lambda)$: choose $u\leftarrow\{0,1\}^m$, set $\mathsf{sk}=u$, and output $\mathsf{pk}\leftarrow O(P_u)$, using the padded point circuit.
\item $\Enc(\mathsf{pk},b)$: construct the quantum circuit $C_{\mathsf{pk},b}(x)=b\wedge E(\mathsf{pk},x)$, padded to size $s'$, and output $c\leftarrow O(C_{\mathsf{pk},b})$.
\item $\Dec(u,c)$: output $E(c,u)$.
\end{enumerate}
All keys and ciphertexts are classical. The ciphertext circuit runs the quantum evaluator $E$, which is why we need obfuscation for quantum circuits.

For correctness, Remark~\ref{rem:qio-simultaneous} ensures that, except with probability $2^{-\lambda}$ over the public key, $E(\mathsf{pk},x)$ computes $P_u(x)$ with error at most $2^{-\lambda}$ on every input. For each such key, $C_{\mathsf{pk},b}$ computes $b\wedge P_u$ with the same error bound. The correctness of \textsf{qiO} then gives $E(c,u)=b$ with negligible error. This bound holds uniformly over these keys because Definition~\ref{def:semantic-qio} requires correctness for every family of such circuits. The exceptional keys contribute only negligible error.

For security, first replace the public key with $\mathsf{pk}\leftarrow O(Z)$. Lemma~\ref{lem:ev-point-hiding} with $t=1$ makes this change computationally indistinguishable. The reduction can run the adversary and generate its challenge ciphertext from the given public key, so the lemma applies to the full encryption experiment.

By Remark~\ref{rem:qio-simultaneous}, except with probability $2^{-\lambda}$, this public key evaluates to zero with error at most $2^{-\lambda}$ on every input. Both $C_{\mathsf{pk},0}$ and $C_{\mathsf{pk},1}$ then compute zero with negligible error on every input, and they also have the same size $s'$ due to padding. Hence, their obfuscations are indistinguishable by \textsf{qiO}, even given $\mathsf{pk}$ as classical auxiliary information. As above, the guarantee holds uniformly over these keys, and exceptional keys contribute negligible probability. Finally, replace the public key by an obfuscation of a random point again. These changes give indistinguishable encryptions of zero and one under the real key distribution.

For polynomially many ciphertexts, replace them one at a time; public encryption lets the reduction generate all other ciphertexts itself. Encrypting each bit of a polynomial-length message separately and changing the encrypted bits one at a time gives \textsf{IND-CPA} security for polynomial-length messages.
\end{proof}

\section{AI Methodology}
The concepts and proof ideas were mainly generated by the authors. However, the proof of the barrier theorem (\Cref{thm:interpolated-barrier}) was produced through interaction with ChatGPT 6 Astra. We also used ChatGPT 6 Astra to help write the other proofs, improve presentation, and check for mistakes. The authors take responsibility for any mistakes or issues.

\bibliographystyle{alpha}
\bibliography{references}
\end{document}